\documentclass[a4paper,11pt]{article}
\pdfoutput=1 
\usepackage{jheppub} 
\usepackage{amsmath,amsfonts,amssymb,amsthm}
\usepackage[british]{babel}
\usepackage{bbm,amscd,bbold}
\usepackage{array}
\usepackage[svgnames,table]{xcolor}
\usepackage[T1]{fontenc}
\usepackage[utf8]{inputenc}
\usepackage{lmodern} 
\usepackage{mathrsfs}
\usepackage{verbatim}
\usepackage{pdfsync}
\usepackage[final]{microtype}
\usepackage{adjustbox}
\usepackage{enumitem}
\usepackage{slashed}
\usepackage{booktabs} 
\usepackage{caption}
\usepackage{appendix}
\usepackage{chngcntr}
\usepackage{apptools}
\usepackage{romanbar}
\usepackage{tikz,pgfplots}
\pgfplotsset{compat=1.18}
\usepgfplotslibrary{fillbetween}
\usetikzlibrary{calc,arrows,arrows.meta,cd,shapes.geometric,positioning,through,decorations.markings,decorations.text}
\usepackage[hyperpageref]{backref}
\renewcommand*{\backref}[1]{}
\renewcommand*{\backrefalt}[4]{%
 \ifcase #1%
 \or [Page~#2.]%
 \else [Pages~#2.]%
 \fi%
}
\theoremstyle{plain}
\newtheorem{lemma}{Lemma}
\newtheorem{proposition}[lemma]{Proposition}

\theoremstyle{definition}

\newcommand{\g}{\mathfrak{g}}

\newcommand{\fk}{\mathfrak{k}}

\renewcommand{\d}{\partial}

\newcommand{\so}{\mathfrak{so}}

\newcommand{\su}{\mathfrak{su}}

\newcommand{\fsl}{\mathfrak{sl}}

\renewcommand{\t}{\mathfrak{t}}

\let\squigto\rightsquigarrow

\newcommand{\q}{\boldsymbol{\pi}}

\newcommand{\x}{\boldsymbol{x}}

\newcommand{\bj}{\boldsymbol{j}}
\newcommand{\p}{\boldsymbol{p}}
\newcommand{\bk}{\boldsymbol{k}}
\newcommand{\ba}{\boldsymbol{a}}

\newcommand{\J}{\boldsymbol{J}}
\newcommand{\bzero}{\boldsymbol{0}}
\newcommand{\bB}{\boldsymbol{B}}
\newcommand{\bD}{\boldsymbol{D}}
\newcommand{\bP}{\boldsymbol{P}}
\newcommand{\bom}{\boldsymbol{m}}

\newcommand{\eH}{\mathscr{H}}

\newcommand{\eO}{\mathcal{O}}

\newcommand{\Cl}{C\ell}

\renewcommand{\Re}{\operatorname{Re}}
\renewcommand{\Im}{\operatorname{Im}}
\newcommand{\ad}{\operatorname{ad}}

\newcommand{\Tr}{\operatorname{Tr}}

\newcommand{\1}{\mathbb{1}}
\renewcommand{\AA}{\mathbb{A}}

\newcommand{\CP}{\mathbb{CP}}

\newcommand{\RR}{\mathbb{R}}

\newcommand{\NN}{\mathbb{N}}

\newcommand{\ZZ}{\mathbb{Z}}

\newcommand{\CC}{\mathbb{C}}

\newcommand{\Aut}{\operatorname{Aut}}

\newcommand{\GL}{\operatorname{GL}}

\newcommand{\SO}{\operatorname{SO}}
\newcommand{\SU}{\operatorname{SU}}
\newcommand{\U}{\operatorname{U}}

\newcommand{\Spin}{\operatorname{Spin}}

\newcommand{\End}{\operatorname{End}}

\newcommand{\zbar}{\overline{z}}
\usepackage{pdflscape}
\usepackage{braket} 
\usepackage{bm} 

\providecommand*{\pd}{\partial}
\renewcommand*{\pd}{\partial}

\providecommand*{\jb}{{\bm{j}}}
\renewcommand*{\jb}{{\bm{j}}}

\newcommand{\bd}{{\bm{d}}}

\newcommand{\bu}{\boldsymbol{u}}
\newcommand{\bv}{\boldsymbol{v}}

\usepackage{tensor} 

\providecommand*{\pd}{\mathop{}\!\partial}
\renewcommand*{\pd}{\mathop{}\!\partial}

\newcolumntype{L}{>{$}l<{$}}

\title{\boldmath Quantum Carroll/fracton particles}

 \author[a,1]{José Figueroa-O'Farrill,\note{ORCID: \href{https://orcid.org/0000-0002-9308-9360}{0000-0002-9308-9360}}}
 \author[b,c,2]{Alfredo Pérez\note{ORCID: \href{https://orcid.org/0000-0003-0989-9959}{0000-0003-0989-9959}}}
 \author[a,3]{and Stefan Prohazka\note{ORCID: \href{https://orcid.org/0000-0002-3925-3983}{0000-0002-3925-3983}}}

\affiliation[a]{Maxwell Institute and School of Mathematics, The University of Edinburgh, James Clerk Maxwell Building, Peter Guthrie Tait Road,
  Edinburgh EH9 3FD, Scotland, United Kingdom}
\affiliation[b]{Centro de Estudios Científicos (CECs), Avenida Arturo Prat 514, Valdivia, Chile}
\affiliation[c]{Facultad de Ingeniería, Arquitectura y Diseño, Universidad San Sebastián, sede Valdivia, General Lagos 1163, Valdivia 5110693, Chile}

 \emailAdd{j.m.figueroa@ed.ac.uk}
 \emailAdd{alfredo.perez@uss.cl}
 \emailAdd{stefan.prohazka@ed.ac.uk}

 \abstract{We classify and relate unitary irreducible representations
   (UIRs) of the Carroll and dipole groups, i.e., we define elementary
   quantum Carroll and fracton particles and establish a
   correspondence between them.  Whenever possible, we express the
   UIRs in terms of fields on Carroll/aristotelian spacetime subject
   to their free field equations.

   We emphasise that free massive (or ``electric'') Carroll and
   fracton quantum field theories are ultralocal field theories and
   highlight their peculiar and puzzling thermodynamic features. We
   also comment on subtle differences between massless and
   ``magnetic'' Carroll field theories and discuss the importance of
   Carroll and fractons symmetries for flat space holography.}

\begin{document}
\maketitle

\section{Introduction}
\label{sec:introduction}

In this work we classify and relate quantum Carroll and fracton
particles, i.e., we classify unitary irreducible representations of
the Carroll~\cite{Levy1965,SenGupta1966OnAA} and
dipole~\cite{Gromov:2018nbv} groups and describe them, whenever
possible, as fields in Carroll spacetime or the aristotelian spacetime
underlying the fracton system. This is a continuation of our earlier
paper \cite{Figueroa-OFarrill:2023vbj}, henceforth referred to as
Part~I, where we studied the classical elementary systems with Carroll
and dipole (for us, ``fracton''\footnote{We will not be able to give full
  justice to the broad field of ``fractons'' for which we refer the
  reader to~\cite{Nandkishore:2018sel,Pretko:2020cko,Grosvenor:2021hkn} for
  reviews. \label{fn:1}}) symmetries. In particular, in this paper we lift the
Carroll/fracton correspondence proposed in Part~I from the classical
to the quantum realm.

A Lie group $G$ is said to be a symmetry of a quantum mechanical model
if the underlying Hilbert space of states admits a unitary
representation\footnote{This might be relaxed to only projective
  representations of $G$, but we can always restrict to honest
  representations by passing to a central extension of $G$.  We shall
  assume that we have done this from now on.  For the case of the
  Carroll and dipole groups there is no need to do this, since there
  are no nontrivial central extensions in $3+1$ dimensions and above.}
of $G$.  In this sense it is typical to think of unitary
representations as describing the symmetries of a quantum mechanical
model, what we will somewhat loosely refer to as quantum symmetries in
this paper.  Indeed, some of these unitary representations may be
constructed by geometrically quantising coadjoint orbits and, for some
(but famously not all) Lie groups, all unitary representations arise
in this way.  In particular, if a group acts as symmetries of a given
spacetime, we expect it to be realised as quantum symmetries of any
natural quantum system defined on that spacetime.  Conversely, it is
often the case that unitary representations of a Lie group $G$ can
actually be realised on classical fields defined on a spacetime on
which $G$ acts by symmetries.

In any unitary representation of a Lie group, the elements of the Lie
algebra give rise to hermitian operators which may be thought of as
physical observables of the corresponding quantum system and in whose
spectrum we might be interested. For example, the generator of time
translations gives rise to the hamiltonian, which governs the energy
spectrum of the theory.

Unitary representations need not be irreducible, but they typically
decompose into irreducible components, which we may think of as the
building blocks of the quantum symmetries of the given
group.  We call them the \emph{elementary} quantum systems and they
are the quantum counterpart of the coadjoint orbits describing the
elementary classical systems.  One of the earliest descriptions of
elementary quantum systems is the Wigner
classification~\cite{Wigner:1939cj} of unitary irreducible
representations (UIRs, for short) of the Poincaré group, which
describe the (free) particles we observe in nature, and forms a
cornerstone of relativistic quantum field theory on Minkowski
spacetime.

Thus if we wish to study quantum mechanical models with a certain
symmetry group, it is useful to first understand the elementary
quantum systems of that group.  This motivates the study of UIRs
of the Carroll and dipole groups and allows us to propose a definition
of what we mean by elementary quantum Carroll and fracton particles.

We are able to treat carrollions and fractons simultaneously for the
most part since free complex Carroll and fracton scalar theories and
their symmetries essentially\footnote{The dipole group is a trivial
  central extension of the Carroll group.} coincide~\cite{Bidussi:2021nmp} (see
also~\cite{Marsot:2022imf}). This led us in Part~I to propose a
correspondence, summarised in Table~\ref{tab:carrvsfrac}, between all
elementary carrollions and fractons, which we show in this work to
persist at the quantum level.

\begin{table}
  \centering
  \caption{Conserved quantities in carrollian and fractonic theories~\cite{Figueroa-OFarrill:2023vbj}}
     \begin{tabular}{l  l }
       \toprule
       Carroll particles                        & Fractonic particles    \\ \midrule \rowcolor{blue!7}
       \multicolumn{2}{c}{angular momentum $\bm{J}$}                         \\ 
       \multicolumn{2}{c}{momentum $\bm{P}$}                                \\ \rowcolor{blue!7}
       center-of-mass $\bm{B}$\quad\quad\quad\quad & dipole moment $\bm{D}$    \\ 
       energy  $H$                              & charge      $Q$        \\ \rowcolor{blue!7}
       ---                                      & energy       $H_{F}$   \\ 
       \bottomrule
  \end{tabular}
  \label{tab:carrvsfrac}
\end{table}

Let us emphasise that out of these elementary ingredients one can of
course build composite objects (which are then reducible), to which
the above correspondence extends.  The elementary dipoles should be
contrasted with composite dipoles built out of two elementary
monopoles, both of which have mobility, in contradistinction to
isolated monopoles. They correspond to two Carroll particles of
opposite energy~\cite{Figueroa-OFarrill:2023vbj}.  Seen from this
perspective it is even less surprising that composite Carroll
particles with opposite mass can move~\cite{Bergshoeff:2014jla} (see
also~\cite{deBoer:2021jej,Zhang:2023jbi}). This is a
manifestation of the fact that the
dipole moment for nonzero charge depends on the choice of origin (see,
e.g., \cite{griffiths1999introduction}). Hence a monopole, which has
nonzero charge ($q \neq 0$), cannot move without changing the dipole
moment. As soon as the charge is zero, e.g., by adding another
monopole with opposite charge, mobility in ways that do not change
the total dipole moment is restored. This shows that it can be useful
to think about Carroll and fracton systems from complementary
perspectives.

Since carrollian physics is relevant for flat space holography at
null~\cite{Duval:2014lpa}, timelike~\cite{Figueroa-OFarrill:2021sxz}
or spacelike~\cite{Gibbons:2019zfs} infinity it might be interesting
to understand them from a fracton perspective. We will provide further
comments concerning these interesting topics in
Section~\ref{sec:discussion}, but let us highlight that some of the
structure of flat space holography can already be seen by using only
Carroll symmetries, e.g., there is a radiative and non-radiative
branch, related to the option of having vanishing or non-vanishing
mass. In addition, the equations of the non-radiative branch share
similarities with the massless or magnetic carrollian field theory.

Another motivation comes from the successful and wide-ranging
applications of tools that fall under the banner of ``scattering
amplitudes''.  See, e.g.,~\cite{Travaglini:2022uwo} for a review. To
apply this remarkable toolkit to carrollian or fractonic theories, one
needs first of all an understanding of the possible quantum particles,
i.e., the ``in'' and ``out'' states of scattering amplitudes. The UIRs
classified in this paper provide the complete answer to this question
and it could be interesting to employ this technology.

The construction of UIRs of the Carroll and dipole groups employed in
this paper is essentially that pioneered by Wigner.  The upshot of
this method is that UIRs are carried by fields in momentum space, more
precisely fields defined over orbits on momentum space of the
``homogeneous'' subgroup: e.g., Lorentz in the case of Poincaré.
These fields transform under (unitary, irreducible) representations of
the little group associated to these orbits.  In the Poincaré case,
from which we derive most of our intuition, the orbits are the mass
shells and the little groups are the subgroups of the Lorentz group
which preserve a given massive, massless or tachyonic momentum.  There
is however one crucial difference between the Poincaré and Carroll
groups: in the latter, boosts commute and hence there are other
choices of abelian subgroups from which we can induce.  This is
reflected in the existence of automorphisms of the Carroll
group~\cite{Figueroa-OFarrill:2023vbj} which mix momenta and boosts.
Although we will induce from characters of the translation group in
this work, one could equally induce from characters of the group
generated by boosts and time translations and in some cases, such as
the centrons (in Carroll language) or elementary dipoles (in fracton
language), it would perhaps be more natural to do that and hence
express the corresponding UIRs as fields on centre-of-mass (or
dipole-moment) space.

It is natural to wish to describe these representations in terms of
classical fields defined over the relevant spacetime: Minkowski in the
case of Poincaré, the eponymous spacetime in the case of Carroll or
the aristotelian spacetime in the case of the
fractons~\cite{Bidussi:2021nmp,Jain:2021ibh}. Such fields transform
according to representations of the homogeneous subgroup, which is the
stabiliser of a chosen origin in the spacetime and hence in passing
from the momentum space description to the spacetime description,
there is always a choice to be made: we need to embed the
representation of the little group (the so-called ``inducing
representation'') into some representation of the homogeneous
subgroup, a process known as ``covariantisation'' in the Physics
literature. The embedding representation need not be unitary; although
a typical consideration is that it should be as small as possible and,
in any case, finite-dimensional, in order to arrive at spacetime
fields with a finite number of components. It typically happens that
the embedding representation is of larger dimension than the inducing
representation and hence that the spacetime field has more degrees of
freedom than the momentum space field. One way to cut down to the
required number of degrees of freedom is by imposing field equations
on the spacetime fields, so that the sought-after irreducible
representation is carried not by all spacetime fields, but only by
those which obey their field equations. There is a systematic way to
arrive at the field equations, once the inducing representation has
been covariantised, in terms of a group-theoretical generalisation of
the Fourier transform. It is well known from the case of the Poincaré
group, that this procedure is the origin of many of the familiar
relativistic free field equations: Klein--Gordon, Dirac, Maxwell,
Proca, linearised Einstein,... In this paper we will give similar
descriptions for some of the UIRs of the Carroll and dipole groups. In
particular, we will see that some of the massless low-helicity UIRs of
the Carroll group are given by solutions of the three-dimensional
euclidean Helmholtz, Dirac and topologically massive Maxwell
equations. This is perhaps not surprising in that the massless
helicity UIRs of the Carroll group are actually UIRs of the
three-dimensional euclidean group, but what is perhaps novel is the
interpretation of these well-known partial differential equations as
irreducibility conditions for three-dimensional euclidean fields and,
by extension, for massless carrollian fields (or neutral fractons).

It is worth highlighting the fact that the field equations we find do
not necessarily agree with the Carroll-invariant field equations in
the literature (see, e.g.,
\cite{Duval:2014uoa,Henneaux:2021yzg,deBoer:2021jej}). The fundamental
reason for this discrepancy is our different points of departure. Our
principal aim in this paper is the classification of UIRs of the
Carroll and dipole groups. These representations are given in terms of
fields on momentum space and that suffices for the classification.
Those fields have the precise number of degrees of freedom (roughly
the dimension of the inducing representation) required to describe the
UIRs. To express the UIRs in terms of spacetime fields, we must make a
choice of how to covariantise the inducing representation and our
approach has been to choose the simplest covariantisation: roughly,
the one which adds the smallest number of extra degrees of freedom.
For example, for massive Carroll UIRs (equivalently, charged fracton
UIRs), the inducing representation is already covariant if we demand
that the boosts (equivalently, the dipole generators) act trivially.
This results in no additional field equations beyond the one coming
from having fixed the energy. This may result in unfamiliar/surprising
field transformation laws, since such an economical choice of
covariantisation is not available for the Poincaré group. Indeed, what
makes this possible is that the boosts commute in the homogeneous
Carroll group, but do not in the Lorentz group, where the commutator
of two boosts is a rotation.

By contrast, in the approach where one departs from building invariant
actions for spacetime fields, it is perhaps not obvious (and indeed
would have to be checked) that the resulting field equations project
onto an irreducible subrepresentation of the representation carried by
the off-shell spacetime fields.

This paper is organised as follows. In the remainder of this
Introduction we provide a self-contained summary of the UIRs of the
Carroll (Section~\ref{sec:quant-carr-part}) and dipole
(Section~\ref{sec:quant-fract-part}) groups and whenever possible we
summarise their description in terms of spacetime fields. Readers who
are happy to skip some of the details could then continue with
Section~\ref{sec:carr-fract-quant} where we discuss Carroll and
fracton quantum field theories and finish with the discussion in
Section~\ref{sec:discussion}, where we highlight, e.g., the relevance
for flat space holography.

The more detailed treatment starts in
Section~\ref{sec:review-coadjoint-orbits}, where we review the basic
results of Part~I on the coadjoint orbits of the Carroll group, their
structure and the action of automorphisms.  The construction of the
UIRs starts in Section~\ref{sec:unit-irred-repr}, based on the method
of induced representations described in some detail in
Appendix~\ref{sec:meth-induc-repr}. As this topic is somewhat
technical, we distill from the appendix a sort of algorithm to
construct the UIRs and this is briefly recapped in
Section~\ref{sec:brief-recap-method}.
Section~\ref{sec:induc-repr-homog} outlines the method and checks that
the momentum space orbits admit invariant measures.
Section~\ref{sec:induc-repr} works out the inducing representations:
the UIRs of the little groups of the momentum orbits. Two of the
little groups are themselves semidirect products and require iterating
the method of induced representations. In Section~\ref{sec:summary} we
put everything together and list the UIRs of the Carroll group:
divided into those with nonzero energy (termed ``massive'' in Part~I
and treated in Section~\ref{sec:uirreps-eneq0}) and those with zero
energy (termed ``massless'' in Part~I and treated in
Section~\ref{sec:uirreps-e=0}). In Section~\ref{sec:simpl-descr-massl}
we comment on a more unified description of the massless UIRs as
induced representations from a larger subgroup of the Carroll group.
We end the section with a conjectural correspondence between the UIRs
and the (quantisable) coadjoint orbits. In
Section~\ref{sec:carrollian-fields} we describe some of the UIRs found
in the previous section in terms of classical fields on Carroll
spacetime. This requires a continuation of the brief recap of the
method of induced representations, describing the covariantisation
procedure and the group-theoretical Fourier transform, and contained
in Section~\ref{sec:brief-recap-method-contd}. We then do an example
of a massive carrollian field (in
Section~\ref{sec:an-explicit-example}) and of a massless carrollian
field with helicity (in Section~\ref{sec:another-example}), which are
the only two classes of UIRs which seem to admit a description in
terms of finite-component carrollian fields. We work out the resulting
field equations for the cases of helicities $0$, $\tfrac12$ and $1$
and recover the three-dimensional Helmholtz, Dirac equation and
topologically massive Maxwell equations, respectively.
Section~\ref{sec:fractonic-fields} is devoted to fractonic particles
and fields. In Section~\ref{sec:unit-irred-repr-3} we classify the
UIRs of the dipole group by observing that there is a bijective
correspondence between UIRs of the Carroll group and classes of UIRs
of the dipole group distinguished solely by the fracton energy. We
then describe some of these UIRs in terms of fields on aristotelian
spacetime. We do the example of a charged monopole in
Section~\ref{sec:charg-arist-fields} and of what could be termed a
neutral aristotelion in Section~\ref{sec:neutr-arist-fields}. In
Section~\ref{sec:carr-fract-quant} we discuss Carroll and fracton
quantum field theories in a second-quantisation language, some of
their similarities, and highlight the relation to ultralocal field
theories~\cite{Klauder:1970cs,Klauder:1971zz}. Additionally, we comment on the intricate thermodynamic
properties of these theories and show that there is a subtle
difference between free massless Carroll theories and the magnetic
Carroll theory. In Section~\ref{sec:discussion} we provide a
discussion where we emphasise interesting connections to flat space
holography and other intriguing topics for further exploration. The
paper ends with two appendices: Appendix~\ref{sec:meth-induc-repr}
contains a short review of the method of induced representations,
whereas Appendix~\ref{sec:hopf-charts-su2} contains some formulae in
special coordinates for the $3$-sphere adapted to the Hopf fibration,
which we make use of in our discussion of massless UIRs of the Carroll
group.

\subsection{Summary}
\label{sec:summary-1}

In this section we summarise the quantum Carroll and fracton particles
and their field theories.  In addition to the vacuum sector, UIRs fall
into two broader classes outlined in Table~\ref{tab:introtable}:
\begin{itemize}
\item $\Romanbar{II}$: massive Carroll particles $\hat H= E_{0}$ and
  charged monopoles $\hat Q= q$
\item $\Romanbar{III}-\Romanbar{V}$: massless Carroll particle
  $\hat H=0$ and neutral fractons $\hat Q = 0$
\end{itemize}
In each case the properties of the particles are quite distinct and
many of the curious physical features of carrollian and fracton
theories can be traced back to this fact.

We will present hermitian operators corresponding to our symmetries.
They are related to the skew-hermitian Lie algebra generators via
multiplication by $i$.

\begin{table}[h]
  \centering
    \caption{Unitary irreducible representations of the Carroll and
    monopole/dipole group}
    \label{tab:introtable}
  \begin{adjustbox}{max width=\textwidth}
    \begin{tabular}{L  L | L l | L l}\toprule
                        &                         & \multicolumn{2}{c}{\text{Carroll}} & \multicolumn{2}{c}{Fractons, $E \in \RR$}                                  \\ \midrule \rowcolor{blue!7}
      \Romanbar{I}      & 2s \in \NN_0            & \multicolumn{4}{c}{vacuum sector}                                                                               \\
      \Romanbar{II}     & 2s \in \NN_0            & E_{0} \neq 0                       & massive spin $s$                       & q \neq 0 & monopole spin $s$      \\ \rowcolor{blue!7}
      \Romanbar{III}    & n \in \ZZ, p>0          & \multicolumn{4}{c}{massless helicity $\tfrac{n}{2}\simeq$ aristotelions}                                        \\ 
      \Romanbar{III}'   & n \in \ZZ               & k>0                                & centrons                               & d>0      & elementary dipoles     \\ \rowcolor{blue!7}
      \Romanbar{IV}_\pm & n \in \ZZ, p>0,\pm      & k>0                                & (anti)parallel helicity $\tfrac{n}{2}$ & d>0      & (anti)parallel dipoles \\
      \Romanbar{V}_\pm  & \theta \in (0,\pi), p>0 & k>0                                & generic massless                       & d>0      & generic dipole         \\ \rowcolor{blue!7}
      \bottomrule
    \end{tabular}
  \end{adjustbox}
  \caption*{By $\Romanbar{I}$ to $\Romanbar{V}$ we enumerate
    inequivalent UIRs of the Carroll and dipole groups. They
    broadly fall into two classes with different physical properties:
    $\Romanbar{II}$ massive carrollions and charged monopoles versus
    $\Romanbar{III}-\Romanbar{V}$ massless carrollions and neutral
    fractons.                                                                                                                                                       \\
    In this table $\NN_0$ denotes the non-negative natural numbers
    (i.e., including zero) and $\ZZ$ the integers. This shows that
    spin $s$ and helicity $\tfrac{n}{2}$ are quantised.  All other
    quantities are real. It is implicit that there is an additional,
    but mostly irrelevant, phase for the fractons. Further
    explanations are given in the summary
    Section~\ref{sec:summary-1}.}
\end{table}

\subsubsection{Quantum Carroll particles and fields}
\label{sec:quant-carr-part}

We are not the first to study the UIRs of the Carroll group, the massive
UIRs were already constructed by Lévy-Leblond~\cite{Levy1965} and the
massless sector was highlighted in~\cite{deBoer:2021jej} (see also
Appendix A in~\cite{Duval:2014uoa}), but this work provides the first
classification. A similar feat has already been accomplished for the
Poincaré~\cite{Wigner:1939cj} and Galilei/Bargmann groups (we refer to
the review~\cite{levygalgr} and references therein) and this work
closes the final gap for quantum symmetries based on the maximally symmetric
affine spacetimes~\cite{Figueroa-OFarrill:2018ilb}.

We will now provide a summary of quantum Carroll particles, i.e., UIRs
of the Carroll group, and discuss some of their properties. Further
details are presented in Section~\ref{sec:unit-irred-repr}. Broadly
speaking they fall into two classes, massive carrollions
($\hat H=E_{0}$) and massless carrollions ($\hat H=0$), which have
very distinctive features.

We parametrise the Carroll group as
\begin{equation}
  g=g(R,\bv,\ba,s) = e^{s H} e^{\ba \cdot \bP} e^{\bv \cdot \bB} R
\end{equation}
where $R$ is a rotation, $\bB$ denote the Carroll boost generators,
$\bP$ the generators of spatial translations and $H$ the generator of
time translations. The Carroll Lie algebra is given by ($i,j,k=1,2,3$)
\begin{align}
  \label{eq:3-carroll-algebra-sum}
  [J_i, J_j] &= \epsilon_{ijk} J_k  &
  [J_i, B_j] &= \epsilon_{ijk} B_k &
  [J_i, P_j] &= \epsilon_{ijk} P_k &
  [B_i,P_j]&=\delta_{ij} H \, .
\end{align}

The quantum Carroll particles, by which we mean UIRs of the
(simply-connected) Carroll group, fall into several different classes
listed below.  The notation $\NN_0$ denotes the non-negative integers
and $V_s$ stands for the complex spin-$s$ irreducible representation
of $\Spin(3) \cong \SU(2)$, of dimension $2s+1$.

\begin{description}
\item[$\Romanbar{I}(s)$ vacuum sector] with $2s \in \NN_0$ with
  underlying Hilbert space $V_s$. When $s=0$ this
  represents the vacuum, whereas for $s>0$ these are spinning 
  vacua. In this representation only the rotations act nontrivially
  and they do so via the spin-$s$ irreducible representation.

\item[$\Romanbar{II}(s,E_0)$ massive spin $s$] with $2s \in \NN_0$ and
  $E_0 \in \RR \setminus\{0\}$~\cite{Levy1965}. This shows that the
  spin of massive quantum Carroll particles is indeed quantised. The
  underlying Hilbert space is given by square-integrable functions
  $\psi \in L^2(\AA^3,V_s)$ and $\p \in \AA^3$ parametrises the
  hyperplane in momentum space with $E =E_0$ (see
  Figure~\ref{fig:mom_lim}).

  The unitary action of $G$ is given by\footnote{We could have
     added an additional label to our wavefunctions such that
    the specific $E=E_{0}$ hyperplane is explicit, e.g., we could have
    written $\psi_{E_{0}}(\p)$. To reduce clutter we will leave the
    energy implicit.}
  \begin{equation}
    (g \cdot \psi)(\p) = e^{i ( E_0 s + \p \cdot \ba)} \rho(R) \psi(R^{-1}(\p + E_0\bv)) 
  \end{equation}
  where $R \mapsto \rho(R)$ denotes the spin-$s$ representation of
  $\Spin(3)$ and the inner product
  \begin{equation}
    \label{eq:massive-inner-product-sum}
    (\psi_1,\psi_2) = \int_{\AA^3} d^3p \left<\psi_1(\p),\psi_2(\p)\right>_{V_s},
  \end{equation}
  where $\left<-,-\right>_{V_s}$ is an $\SU(2)$-invariant hermitian inner
  product on $V_s$.  The hermitian operators that correspond to our
  conserved charges are in this basis given by
  \begin{align}
    \label{eq:op-p-basis-sum}
    \hat \J &= - i \p \times \frac{\pd}{\pd \p} + \hat{\bm{S}}  & \hat{\bm{B}} &= -i E_{0} \frac{\pd}{\pd \p} &  \hat H &= E_{0} & \hat{\bm{P}}&= \p   \, ,
  \end{align}
  where $\hat{\bm{S}}$ are the infinitesimal generators of the
  spin-$s$ representation $\rho(R)$. Massive spin-$s$ carrollions can
  then be labeled by
  \begin{equation}
    \hat H = E_{0} \qquad\text{and}\qquad \hat{\bm{S}}^{2} = s (s+1) \, ,
  \end{equation}
  which are multiples of the identity. For massive carrollions we can
  also define a position operator $\hat{\bm{X}}$
  \begin{equation}
    \hat{ \bm{X}} = \frac{1}{E_{0}} \hat{\bm{B}}
  \end{equation}
  which agrees with the intuition that the centre of mass of a massive
  Carroll particle is the energy multiplied by the position and
  satisfies the canonical commutation relations
  \begin{equation}
    [\hat X_{i},\hat P_{j}] = - i \delta_{ij} \, .
  \end{equation}

  We may alternatively diagonalise with respect to $\hat{\bB}$, which
  is related to the above via a Fourier transform
  (see~\eqref{eq:psi-fourier}) and express the representation in the
  ``boost basis'' as
  \begin{equation}
  \label{eq:g-rep-massive-boost}
  (g \cdot \tilde\psi)(\bk) = e^{i ( - E_0 s + \bk \cdot \bv)} \rho(R) \tilde\psi(R^{-1}(\bk - E_0\ba))  \, ,
  \end{equation}
  where the inner product is now given by
  \begin{equation}
    \label{eq:k-massive-inner-product-sum}
    (\tilde\psi_1,\tilde\psi_2) = \int_{\AA^3} d^3k \left<\tilde\psi_1(\bk),\tilde\psi_2(\bk)\right>_{V_s} \, .
  \end{equation}
  We provide further details for this UIR in
  Section~\ref{sec:uirreps-eneq0}.

  This representation can also be described using fields on Carroll
  spacetime, i.e., as massive Carroll field theories. These
  $V_s$-valued fields $\phi(t,\x)$ are obtained from the $V_s$-valued
  momentum space fields $\psi(\p)$ via a group-theoretical Fourier
  transform which, in this case, agrees with the classical Fourier
  transform:
  \begin{equation}
    \phi(t,\x) = e^{-i E_0 t} \int_{\AA^3} d^3p e^{-i\p\cdot\x} \psi(\p) \, .
  \end{equation}
  The field $\phi$ satisfies the obvious (and only) field equation
  \begin{equation}
    \frac{\d\phi}{\d t} = -i E_0 \phi.
  \end{equation}
  The action of the Carroll group on the spacetime fields is given by
  \begin{equation}
    (g \cdot \phi)(t,\x) = \rho(R) \phi(t-s-\bv\cdot(\x -\ba),
    R^{-1}(\x -\ba)) \, 
  \end{equation}
  where we want to emphasise that the fields are scalars under boosts,
  since these act nontrivially only on the coordinates.
  
  When we do not restrict to just one orbit and allow both energies
  $E=\pm E_{0}$ we are led to ultralocal (quantum) field
  theories~\cite{Klauder:1970cs,Klauder:1971zz} or ``electric Carroll
  field'' theories~\cite{Henneaux:2021yzg,deBoer:2021jej}, as
  discussed in more detail in Section~\ref{sec:mass-carr-monop}.
  
\item[$\Romanbar{III}(n,p)$ massless helicity $\frac{n}{2}$] with real
  $p>0$ and $n \in \ZZ$, so the helicity is now quantised. The
  underlying Hilbert space consists of complex-valued
  functions\footnote{More precisely, square-integrable (relative to an
    $\SU(2)$-invariant measure) sections of the line bundle
    $\mathscr{O}(-n)$ over the complex projective line, but the
    description in this summary suffices.} on the
  complex plane, where the action of $G$ is given by
    \begin{equation}
    (g \cdot \psi)(z) = e^{i \ba \cdot \q(z)} \left( \tfrac{\eta +
          \overline{\xi} z}{| \eta + \overline{\xi} z|} \right)^{-n}
      \psi\left( \tfrac{\overline{\eta}z -\xi}{\eta + \overline{\xi} z} \right).
  \end{equation}
  In this case $z$ is a stereographic coordinate on the sphere
  $\|\p \| =p$, the action of time translations and boosts is trivial
  and $\q(z)$, given in equation~\eqref{eq:q-vector}, satisfies
  $\|\q(z) \|^{2}=p^{2}$. The rotation group $\SU(2)$ acts on $z$ via
  linear fractional transformations:
  \begin{align}
    \label{eq:rot}
    R =
    \begin{pmatrix}
      \eta & \xi \\ - \overline \xi & \overline \eta
    \end{pmatrix} \in SU(2) \quad\text{acts as}\quad z \mapsto \tfrac{\eta z + \xi}{\overline\eta - \overline{\xi} z} \, .
  \end{align}
  Consequently, $\hat{\bm{B}} = \hat{H}=0$, however
  $\hat{\bm{P}} = \q(z)$ so that these representations are indeed
  specified by the helicity $\tfrac{n}{2}$ and
  $\hat{\bm{P}}^{2} = p^{2}$. The inner product on the Hilbert space
  is given by
  \begin{equation}
    \label{eq:inner-product-sphere-sum}
    \left<\psi_1, \psi_2\right> := \int_\CC \frac{2i dz \wedge
      d\zbar}{(1+|z|^2)^2} \overline{\psi_1(z)}\psi_2(z) \, .
  \end{equation}

  This UIR can also be described using spacetime fields. One
  possibility is to covariantise the inducing representation of
  $\U(1)$ of weight $n$ with boosts acting trivially into the
  spin-$|n/2|$, representation $V_{|n/2|}$ of $\SU(2)$ as the highest
  (if $n\geq 0$) or lowest (if $n\leq 0$) weight vectors in
  $V_{|n/2|}$. The $V_{|n/2|}$-valued spacetime fields $\phi(t,\x)$
  are given in terms of the momentum space fields $\psi(z)$ by
  \begin{equation}
    \phi(t,\x) = \int_\CC \frac{2i dz \wedge d\zbar}{(1 + |z|^2)^2}
    e^{-i\x \cdot \q(z)} \rho(\sigma(z))\psi(z),
  \end{equation}
  where $\sigma(z) \in \SU(2)$ is defined by
  \begin{equation}
    \sigma(z) = \frac1{\sqrt{1+|z|^2}}
    \begin{pmatrix}
      z & -1 \\ 1 & \zbar
    \end{pmatrix}.
  \end{equation}
  Notice that the spacetime fields do not depend on $t$, all massless
  Carroll fields fulfil
  \begin{align}
    \label{eq:notime}
      \frac{\d\phi}{\d t} = 0 \, ,
  \end{align}
  so they are essentially euclidean three-dimensional fields. The
  action of $G$ factors through the action of the three-dimensional
  euclidean group:
  \begin{equation}
    (g\cdot\phi)(\x) = \rho(R) \phi(R^{-1}(\x - \ba)) \, ,
  \end{equation}
  where we see that boosts act trivially. The additional field
  equations which project to the irreducible subrepresentation can be
  worked out for the lowest values of the helicity. For helicity 0 we
  obtain the Helmholtz equation
  \begin{equation}
    (\bigtriangleup + p^2) \phi(\x) = 0,
  \end{equation}
  for a scalar field, where $\bigtriangleup$ is the laplacian acting
  on functions in three-dimensional euclidean space.  For helicity
  $1/2$, we obtain the three-dimensional euclidean Dirac equation
  \begin{equation}
    \left( \slashed{\d} + i p \right) \phi(\x)= 0,
  \end{equation}
  where now $\phi(\x)$ is a 2-component field taking values in the
  spin-$1/2$ representation of $\SU(2)$.  Finally, for helicity $1$ we
  find the topologically massive Maxwell equation of
  \cite{Deser:1981wh,Deser:1982vy}:
  \begin{equation}
    \nabla \times \phi = p \phi,
  \end{equation}
  where $\phi$ is now a three-dimensional vector field, which is
  metrically dual to the Hodge dual of the Maxwell field-strength.
  
\item[$\Romanbar{III}'(n,k)$ centrons] with real $k>0$ and
  $n \in \ZZ$. They are in many ways analogous to the massless
  carrollions just described, so we will be brief. The underlying
  Hilbert space consists again of complex-valued functions on the
  complex plane and the action of $G$ is given by
  \begin{equation}
    (g \cdot \psi)(z) = e^{i \bv \cdot \q(z)} \left( \tfrac{\eta +
        \overline{\xi} z}{| \eta + \overline{\xi} z|} \right)^{-n}
    \psi\left( \tfrac{\overline{\eta}z -\xi}{\eta + \overline{\xi} z}
    \right),
  \end{equation}
  where $z$ is now a stereographic coordinate on the sphere
  $\|\bk \|=k$ and is $\q(z)$ given in equation~\eqref{eq:q-vector}
  with $p \mapsto k$. In this case the action of the time and spatial
  translations is trivial and the representations are uniquely
  specified by $\hat{\bm{P}} = \hat{H}=0$,
  $\hat{\bm{B}}^{2} = k^{2}$ and $n\in\ZZ$.  The inner product on
  the Hilbert space is again given by~\eqref{eq:inner-product-sphere-sum}.

  We can also write down field theories for the centrons and they are
  mutatis mutandis the same as for the massless carrollions, with the
  interesting twist that they live naturally in ``centre-of-mass
  space''. In this sense they are more reminiscent of internal degrees
  of freedom, such as spin.

\item[$\Romanbar{IV}_\pm(n,p,k)$ (anti)parallel massless helicity
  $\frac{n}{2}$] with $n \in \ZZ$ and real $ p,k >0$.  The underlying
  Hilbert space   is again given by complex-valued functions on the
  complex plane and the action of $G$ is given by
  \begin{equation}
    (g \cdot \psi)(z) = e^{i (\ba \pm \tfrac{k}{p} \bv)\cdot \q(z)} \left( \tfrac{\eta +
        \overline{\xi} z}{| \eta + \overline{\xi} z|} \right)^{-n}
    \psi\left( \tfrac{\overline{\eta}z -\xi}{\eta + \overline{\xi} z} \right).
  \end{equation}
  In this case $z$ is a stereographic coordinate on the sphere
  $\|\p \| =p$. By inspection we see that $\Romanbar{III}(n,p)$ is the
  limit of $\Romanbar{IV}_\pm(n,p,k)$ as $k\to 0$, which results from
  formally putting $k=0$ in the above expression for $g\cdot\psi$.
  The action of time translations is trivial, consequently
  $\hat{H}=0$. For the momentum and centre-of-mass operators we obtain
  $\hat{\bm{P}} = \q(z)$ and $\hat{\bm{B}} = \pm \tfrac{k}{p} \q(z)$,
  respectively. Since $\bk = \pm \tfrac{k}{p} \p$, the sign tells us
  whether   we are in the parallel $(+)$ or antiparallel $(-)$ cases.
  In summary, we can characterise these UIRs by
  $\hat{\bm{P}}^{2} = p^{2}$ and $\hat{\bm{B}}^2 = k^2$ and the sign
  of $\hat{\bm{P}} \cdot \hat{\bm{B}}$.  The inner product on the
  Hilbert space is again given by~\eqref{eq:inner-product-sphere-sum}.
  
\item[$\Romanbar{V}_\pm(p,k,\theta)$ generic massless] with real
  $p, k>0$ and $\theta \in (0,\pi)$. It is interesting to note that
  there are no discrete quantum numbers for the generic massless
  particles. The underlying Hilbert space is $L^2(S^3, \CC)$, which
  are the square-integrable functions on the round $3$-sphere with
  values in a one-dimensional unitary representation of the nilpotent
  subgroup of $G$ generated by boosts and translations. The unitary
  character of this representation is such that
  \begin{equation}
    \chi\left( e^{s H + \ba \cdot \bP} e^{\bv \cdot \bB}  \right) =
    e^{i \left( \ba \cdot \p + \bv \cdot \bk \right)},
  \end{equation}
  where $\p = (0,0,p)$ and $\bk = (k\sin\theta,0,k\cos\theta)$.
  We identify $S^3$ with the $\SU(2)$ subgroup of $G$ and we write the
  action of $g = g(R,\ba,\bv,s) \in G$ on $L^2(S^3,\CC)$ as
  \begin{equation}
    (g \cdot \Psi)(S) = e^{i (\ba \cdot S \p + \bv \cdot S\bk)} \Psi(R^{-1}S),
  \end{equation}
  where $S \in \SU(2)$.  The inner product is given by
  \begin{equation}
    \left<\Psi_1,\Psi_2\right> = \int_{S^3} d\mu(S)
    \overline{\Psi_1(S)}\Psi_2(S),
  \end{equation}
  with $d\mu(S)$ the volume form of a round metric on $S^3$, or
  equivalently a bi-invariant Haar measure on $\SU(2)$. This
  representation breaks up as the orthogonal direct sum of two UIRs:
  $L^2_\pm(S^3,\CC)$, where $\Psi \in L^2_\pm(S^3,\CC)$ if and only if
  $\Psi(-S) = \pm \Psi(S)$ for all $S \in \SU(2)$.  The sign labels
  two inequivalent quantisations of the same coadjoint orbit, a
  phenomenon typically associated to a disconnected stabiliser, which
  is indeed the underlying reason here too as discussed in
  Section~\ref{sec:simpl-descr-massl}.

  In summary, apart from that sign, the representation is uniquely
  specified by
  \begin{align}
    \label{eq:generic-spec}
    \hat H &= 0  & \hat{\bm{P}}^{2}&= p^{2} & \hat{\bm{B}}^{2} &= k^{2} & \hat{\bm{P}} \cdot \hat{\bm{B}} = p k \cos \theta \, ,
  \end{align}
  where $\hat{\bm{P}} = S\p$ and $\hat{\bm{B}}=S\bk$.
  
  From the point of view of path integral quantisation the subtleties
  in the quantisation of these orbits derives from the intricate
  constraint structure of the orbits, which are basically the
  classical analog of~\eqref{eq:generic-spec}, as in Section~3.5 of
  Part~I.
\end{description}

\subsubsection{Quantum fracton particles and fields}
\label{sec:quant-fract-part}

In this section we summarise the UIRs of the dipole group and some of
their field-theoretic realisations on aristotelian spacetime.  To the
best of our knowledge there have been no attempts towards a
classification of unitary irreducible representations of the dipole
group. The dipole Lie algebra\footnote{A better name might be
  monopole-dipole algebra, since the particles described by this
  algebra include monopoles as well as dipoles.} is given by
\begin{align}
  \label{eq:fracton-algebra-intro}
  [J_i, J_j] &= \epsilon_{ijk} J_k  &
  [J_i, P_j] &= \epsilon_{ijk} P_k &
  [J_i, D_j] &= \epsilon_{ijk} D_k &
  [D_i,P_j]&=\delta_{ij} Q \, ,
\end{align}
with an additional generator $H_F$ which is central and most notably
the exchange of center-of-mass $B_{i}$ with dipole moment $D_{i}$ and
Carroll energy $H$ with charge $Q$, as shown in
Table~\ref{tab:carrvsfrac}.

Let us now discuss the quantum generalisation of the correspondence
between Carroll and fracton
particles~\cite{Figueroa-OFarrill:2023vbj}. As explained in
Section~\ref{sec:fractonic-fields}, UIRs of the dipole group are in
bijective correspondence with the UIRs of the Carroll group, except
that we extend them to a UIR of the dipole group by declaring that
$e^{s H_F}$ should act via the unitary character
$\chi(e^{s H_F}) = e^{i s E}$ for some $E\in \RR$ which is to be
interpreted as the fracton energy. We therefore use the same notation,
but replacing the Carroll energy $E_0$ with the monopole charge $q$
and the magnitude of the centre-of-mass $k$ with the magnitude of the
dipole moment $d$ and adding a label $E$: hence the UIRs of the dipole
group are $\Romanbar{I}(s,E)$, $\Romanbar{II}(s,q,E)$,
$\Romanbar{III}(n,p,E)$, $\Romanbar{III}'(n,d,E)$,
$\Romanbar{IV}_\pm(n,p,d,E)$ and $\Romanbar{V}_\pm(p,d,\theta,E)$, as
can be seen in Table~\ref{tab:introtable}.

\paragraph{Monopoles with charge $q$ and spin $s$.}
\label{sec:monopoles}

For example, monopoles of charge $q\neq 0$ and spin $s$ (where $2s$ is
a non-negative integer) are given by the UIR $\Romanbar{II}(s,q,E)$.
The underlying Hilbert space is given by square-integrable functions
$\psi \in L^2(\AA^3,V_s)$ which means that they are functions on the
hyperplane in momentum space $\p \in \AA^3$ with fixed charge $q$ and
energy $E$ (see Figure~\ref{fig:mom_lim} with $E \mapsto q$). They are
valued in $V_s$, i.e., in the complex spin-$s$ UIR of $\SU(2)$. The
action of the dipole group is given as follows. If we let
\begin{equation}
  \label{eq:g-dipole-group}
  g=g(R,\bom,\ba,\theta,s) = e^{s H_F + \theta Q + \ba \cdot \bP}
  e^{\bom \cdot \bD} R
\end{equation}
denote the generic element of the dipole group, we have, for $\psi$ a
$V_s$-valued field, that
\begin{equation}
  \label{eq:dipole-action}
  (g \cdot \psi)(\p) = e^{i ( q \theta +  E s + \p \cdot \ba)} \rho(R) \psi(R^{-1}(\p + q\bom)) \, .
\end{equation}
Let us emphasise that the dipole transformation acts as expected
$\p \mapsto \p + q\bom$ and $\rho(R)$ is a manifestation of the fact
that these are spin $s$ monopoles. We could have written
$\psi_{q,E}(\p)$ to emphasise that our functions are restricted to
these specific charge $q$ and energy $E$. The infinitesimal action
of~\eqref{eq:dipole-action} is given by
\begin{align}
  \label{eq:fract-op}
  \hat \J &= - i \p \times \frac{\pd}{\pd \p} + \hat{\bm{S}} & \hat{Q} &= q   & \hat{\bm{D}} &= -i q \frac{\pd}{\pd \p} &  \hat H &= E & \hat{\bm{P}}&= \p   \, ,
\end{align}
which are hermitian operators with respect to the inner product
\begin{equation}
  \label{eq:inner-prod-p}
  (\psi_1,\psi_2) = \int_{\AA^3} d^3p \left<\psi_1(\p),\psi_2(\p)\right>_{V_s} \, .
\end{equation}
The UIRs can then be uniquely labeled by
\begin{align}
  \hat{Q} &= q  & \hat H &= E & \hat{\bm{S}}^{2} &= s (s+1) \, ,
\end{align}
which are multiples of the identity. We can
also define a position operator $\hat{\bm{X}}$
\begin{align}
  \hat{\bm{X}} = \frac{1}{q} \hat{\bm{D}}
\end{align}
which agrees with the intuition that the dipole moment is given by the
charge times the position. The position operator satisfies the
canonical commutation relation
\begin{align}
  [\hat X_{i},\hat P_{j}] = - i \delta_{ij} \, .
\end{align}

We may alternatively diagonalise with respect to
$\hat{\bD}$, which is related to the above
via a Fourier transform (see~\eqref{eq:psi-fourier}) and express the
representation in the ``dipole basis''
\begin{equation}
  \label{eq:g-rep-massive-boost-dipole}
  (g \cdot \tilde\psi)(\bd) = e^{i ( - E s + q\theta +\bd \cdot \bom)} \rho(R) \tilde\psi(R^{-1}(\bd - q\ba))  \, ,
\end{equation}
where the dipole moment is, as expected, shifted by the
translations.\footnote{This choice of basis was also employed in
  Appendix D in~\cite{Jensen:2022iww}.}  The inner product is then
given by
\begin{equation}
  \label{eq:k-massive-inner-product-sum-dipole}
  (\tilde\psi_1,\tilde\psi_2) = \int_{\AA^3} d^3d \left<\tilde\psi_1(\bd),\tilde\psi_2(\bd)\right>_{V_s} \, .
\end{equation}

We may describe these UIRs in terms of fields on the aristotelian
spacetime~\cite{Bidussi:2021nmp,Jain:2021ibh} with coordinates
$(t,\x)$. The field $\phi(t,\x)$ is obtained from $\psi(\p)$ via a
Fourier transform
\begin{equation}
  \phi(t,\x) = e^{-iE t} \int_{\AA^3} d^3p e^{-i\p\cdot\x} \psi(\p)~,
\end{equation}
and the action of the generic element $g$ of the dipole group in
equation~\eqref{eq:g-dipole-group} is given by
\begin{equation}
  (g \cdot \phi)(t,\x) = e^{i q(\theta + \bom\cdot (\x - \ba))} \rho(R)
  \phi(t-s, R^{-1}(\x - \ba)) \, ,
\end{equation}
with $\rho$ the spin-$s$ representation of $\SU(2)$. In particular,
pure charge and dipole transformations act as expected via a phase
\begin{align}
  \label{eq:dipolrot}
  \phi(t,\x) \mapsto   e^{i q(\theta + \bom\cdot\x)} 
  \phi(t, \x) \, .
\end{align}
The only field equation is
\begin{equation}
  \frac{\d \phi}{\d t} = -iE \phi.
\end{equation}

Readers who are happy to skip the details of the classification of the
UIRs could continue with our discussion of Carroll and fracton quantum
field theories in Section~\ref{sec:carr-fract-quant}.

\begin{figure}
\centering
\begin{tikzpicture}
  \node at (7,0) {\includegraphics[width=0.5 \textwidth]{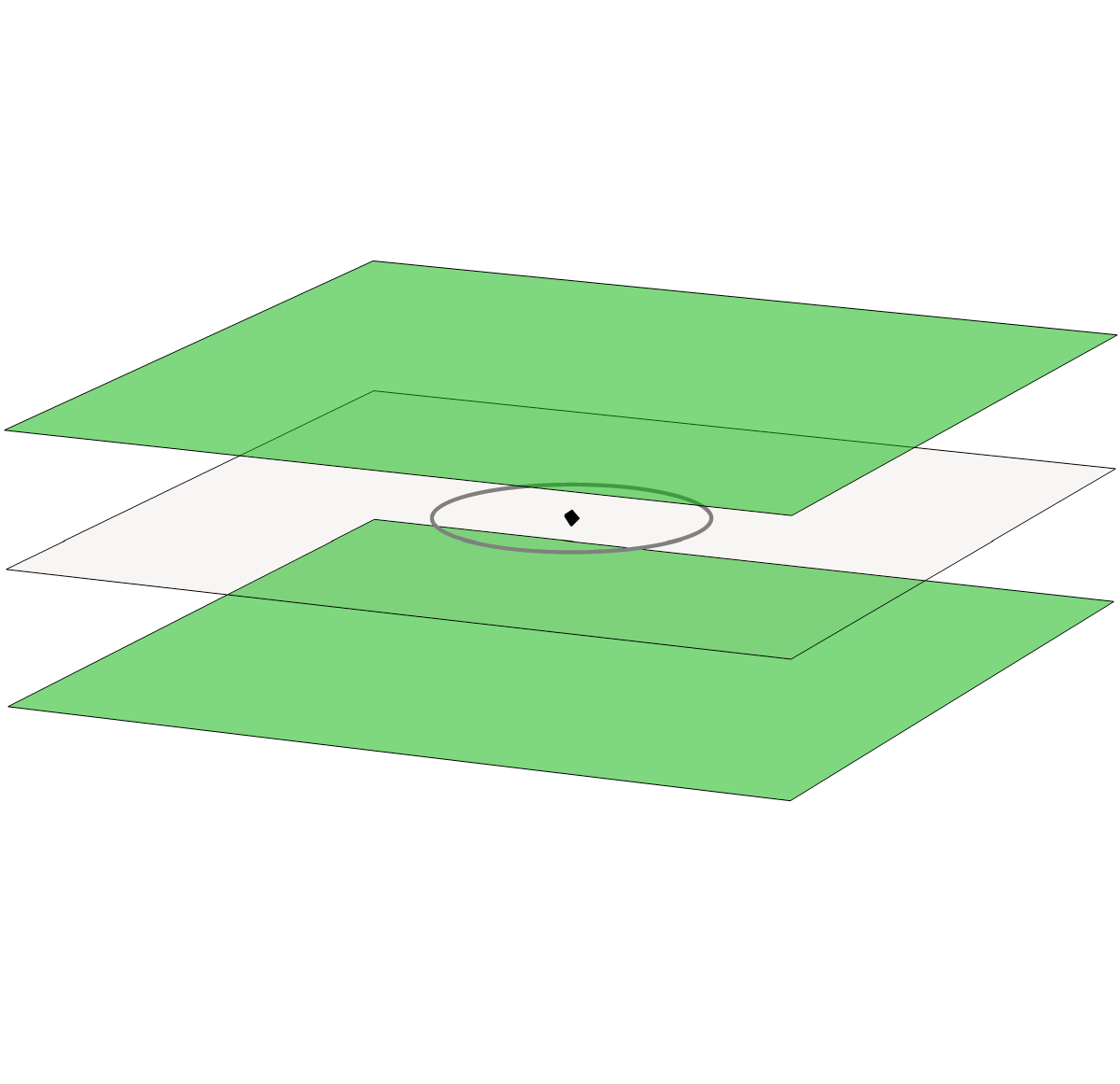}};
  \draw[-Stealth,thick] (2.7,-3) -- (2.7,3);
  \node[right] at (2.7,2.9) {$E$};
\end{tikzpicture}
\caption{This figure shows the coadjoint orbits of Carroll group in
  momentum space $(E,\p)$. Broadly they fall into two classes depending on vanishing or nonvanishing $E$.\\
  When $E=E_{0}\neq 0$, and since the energy is a Casimir, the orbits
  are given by three dimensional planes (depicted in green). When
  $E=0$ the orbits are given by $\|\bm{p}\| = const.$ two spheres, one
  of which we have represented as a black circle. The whole $E=0$
  plane is foliated by such spheres, while the origin $\|\bm{p}\| = 0$
  is the dot in the middle.
  \\
  Let us emphasise that this figure only represents the $(E,\p)$ part
  of the full dual space $(\jb,\bv,\p,E)$ and the complete structure
  of the orbits is more intricate and involves spin degrees of
  freedom.}
    \label{fig:mom_lim}
\end{figure}

\section{Review of coadjoint orbits of the Carroll group}
\label{sec:review-coadjoint-orbits}

In Part~I we classified the coadjoint orbits of the
($3+1$)-dimensional Carroll group.  The Carroll Lie algebra $\g$ is
the ten-dimensional real Lie algebra spanned by $J_i,B_i,P_i,H$ where
$i=1,2,3$ subject to the following non-zero brackets:
\begin{align}
  \label{eq:3-carroll-algebra}
  [J_i, J_j] &= \epsilon_{ijk} J_k  &
  [J_i, B_j] &= \epsilon_{ijk} B_k &
  [J_i, P_j] &= \epsilon_{ijk} P_k &
  [B_i,P_j]&=\delta_{ij} H \, ,
\end{align}
where the Levi-Civita symbol $\epsilon_{ijk}$ is normalised so that
$\epsilon_{123}=1$.  The connected Carroll group $G$ is isomorphic to
a semidirect product $G \cong K \ltimes T$, where $K \cong \SO(3)
\ltimes \RR^3$ and $T \cong \RR^4$.  The Lie algebra $\t$ of $T$ is
spanned by $P_i, H$ and the Lie algebra $\fk$ of $K$ by $J_i, B_i$.
The group $K$ is isomorphic to the (connected) three-dimensional
euclidean group.

Elements $\alpha \in \g^*$ in the dual of the Carroll Lie algebra are
parametrised by the ``momenta'' of classical particles; that is,
$\alpha = (\bj, \bk, \p, E)$ where
$\bj = \left<\alpha,\boldsymbol{J}\right>$ is the angular momentum,
$\bk = \left<\alpha, \boldsymbol{B}\right>$ is the centre of mass,
$\p = \left<\alpha, \boldsymbol{P}\right>$ is the linear momentum and
$E = \left<\alpha, H\right>$ is the energy.  Coadjoint orbits belong
to several classes distinguished in the first instance by the value of
the Casimir elements $H$ and $W^2$, which is the euclidean norm of
\begin{align}
  \label{eq:W}
  W_i := H J_i + \epsilon_{ijk}P_j B_k \, .
\end{align}
On
$\alpha = (\bj, \bk, \p, E)$,
\begin{equation}
  H(\alpha) = E \qquad\text{and}\qquad W^2(\alpha) = \| E \bj + \p
  \times \bk\|^2.
\end{equation}
Notice that since the energy is constant on each orbit, it is
trivially bounded below, that being a typical physical requirement.
See also Figure~\ref{fig:mom_lim}.

In Part~I, we arrived at the classification of coadjoint
orbits displayed in Table~\ref{tab:coadjoint-orbits}.

\begin{table}
  \centering
  \caption{Coadjoint orbits of the Carroll group}
  \setlength{\extrarowheight}{3pt}
  \resizebox{\linewidth}{!}{
    \begin{tabular}{>{$}l<{$}>{$}l<{$}>{$}c<{$}>{$}l<{$}}
      \toprule
      \multicolumn{1}{l}{\#} & \multicolumn{1}{c}{Orbit representative} & \dim\mathcal{O}_\alpha & \multicolumn{1}{c}{Equations for orbits}\\
                             & \multicolumn{1}{c}{$\alpha=(\bj, \bk, \p,E)$} & & \\ \midrule \rowcolor{blue!7}
      1  & (\bzero, \bzero, \bzero,E_0) & 6 & E =E_0 \neq 0, E_0 \bj + \p \times \bk = \bzero\\ 
      2  & (S\bu, \bzero,  \bzero , E_0)  & 8 & E =E_0\neq 0,\|\bj+E_0^{-1}\p \times\bk\| = S >0 \\ \midrule \rowcolor{blue!7}
      3  & (\bzero,\bzero,\bzero,0)  & 0  & E=0,\p=\bzero,\bk=\bzero,\bj=\bzero)\\     
      4  & (j\bu,\bzero,\bzero,0) & 2 & E=0,\p=\bzero,\bk=\bzero,\|\bj\|=j>0 \\ \rowcolor{blue!7}
      5  & ( h \bu, k \bu, \bzero, 0) & 4 & E=0,\p=\bzero,\|\bk\|=k>0, \bj\cdot\bk = h\|\bk\| \in \RR \\
      6  & (h \bu, \bzero, p \bu, 0) & 4 & E=0,\bk=\bzero,\|\p\|=p>0, \bj \cdot \p = h \|\p\|\in \RR \\     \rowcolor{blue!7}
      7_\pm & (h\bu,\pm k\bu, p\bu, 0) & 4 & E=0,\|\p\|=p>0, \|\bk\|=k>0, \p\cdot\bk =\pm p k, \bj \cdot \p = h \|\p\|\in \RR \\
      \rowcolor{blue!7}      
      8  &  (\bzero, k \cos\theta \bu + k\sin\theta \bu_\perp, p \bu,0)  & 6 & E=0,\|\p\|=p>0, \|\bk\|=k>0, \p\cdot\bk = pk\cos\theta, \theta \in  (0,\pi) \\
      \bottomrule
    \end{tabular}
  }
  \caption*{This table provides an overview of the coadjoint orbits of
    the Carroll group. As indicated by the horizontal
    line they are separated into orbits with $E\neq 0$ and orbits with
    $E=0$.  The second column displays an orbit representative: the
    notation is such that $\bu \in \RR^3$ represents a fixed unit-norm
    vector and in the last row $\bu_\perp\in \RR^3$ is a second unit-norm
    vector perpendicular to $\bu$.  The third column is the dimension
    of the orbit and the last column provides the equations
    which define the orbits.}
  \label{tab:coadjoint-orbits}
\end{table}

We also determined the structure of the coadjoint orbits as
homogeneous fibre bundles over the $K$-orbits in $\t^*$ and that plays
an important rôle in the construction of induced representations.
Briefly, we write $\alpha \in \g^*$ as
$(\kappa,\tau) \in \fk^* \oplus \t^*$.  Since $K$ acts on $T$, it acts
on $\t$ and hence on $\t^*$.  We let $\eO_\tau = K \cdot \tau$ denote
the $K$-orbit of $\tau$ in $\t^*$. Let $K_\tau$ denote the stabiliser
of $\tau$ in $K$ and let $\fk_\tau$ be its Lie algebra.  We let
$\kappa_\tau \in \fk_\tau^*$ denote the restriction of $\kappa$ to
$\fk_\tau$ and let $\eO_{\kappa_\tau}$ denote its $K_\tau$-coadjoint
orbit.  Then as explained, for example in \cite{MR387499} (see also
\cite{Oblak:2016eij}) the $G$-coadjoint orbit of $\alpha =
(\kappa,\tau)$ is the fibred product
\begin{equation}
  \begin{tikzcd}
    \eO_\alpha \arrow[r] \arrow[d] & T^*\eO_\tau \arrow[d] \\
    K\times_{K_\tau} \eO_{\kappa_\tau} \arrow[r] & \eO_\tau\\
  \end{tikzcd}
\end{equation}
over $\eO_\tau$ of the cotangent bundle $T^*\eO_\tau$ and the
homogeneous fibre bundle over $\eO_\tau$ whose fibre is the
$K_\tau$-coadjoint orbit of $\kappa_\tau$.  A more standard notation
for that fibred product would be
\begin{equation}
  \eO_\alpha = T^*\eO_\tau\times_{\eO_\tau} (K \times_{K_\tau} \eO_{\kappa_\tau}).
\end{equation}
It is a symplectic manifold of dimension
$2 \dim \eO_\tau + \dim \eO_{\kappa_\tau}$ and carries a
Carroll-invariant symplectic structure.  For every coadjoint orbit of
the Carroll group, one can determine $\tau$, $\eO_\tau$, $K_\tau$ and
the structure of the orbit.  This was done in Part~I, from where we
borrow Table~\ref{tab:structure-orbits}.  
\begin{table}[h]
  \centering
    \caption{Deconstructing the coadjoint orbits}
    \label{tab:structure-orbits}
  \begin{adjustbox}{max width=\textwidth}
    \begin{tabular}{>{$}l<{$}|*{8}{>{$}c<{$}}}
      \# & \alpha \in \g^* & \tau \in \t^* & \eO_\tau & K_\tau & \kappa \in \fk^* & \kappa_\tau \in \fk_\tau^* & \eO_{\kappa_\tau} & \eO_\alpha\\\toprule \rowcolor{blue!7}
      1& (\bzero,\bzero,\bzero,E_0\neq0) & (\bzero,E_0) & \AA^3_{E=E_0} & \SO(3) & (\bzero,\bzero) & \bzero &  \{\bzero\} & T^*\AA^3 \\
      2& (\bj\neq\bzero,\bzero,\bzero,E_0\neq0) & (\bzero,E_0) & \AA^3_{E=E_0} & \SO(3) & (\bj,\bzero) & \bj &  S^2_{\|j\|} & T^*\AA^3 \times_{\AA^3} (K \times_{K_\tau} S^2) \\ \rowcolor{blue!7}
      3& (\bzero,\bzero,\bzero,0) & (\bzero,0) & \{(\bzero,0)\} & K & (\bzero,\bzero) & (\bzero,\bzero) &  \{(\bzero,\bzero)\} & \{0\}\\
      4& (\bj\neq\bzero,\bzero,\bzero,0) & (\bzero,0) & \{(\bzero,0)\} & K & (\bj,\bzero) & (\bj,\bzero) &  S^2_{\|\bj\|} & S^2 \\ \rowcolor{blue!7}
      5& (\bj,\bk\neq\bzero,\bzero,0)_{\bj\times\bk=\bzero} & (\bzero,0) & \{(\bzero,0)\} & K & (\bj,\bk) & (\bj,\bk) & T^*S^2_{\|\bk\|} & T^*S^2\\
      6& (\bj,\bzero,\p\neq 0,0)_{\bj\times \p = \bzero} & (\p,0) & S^2_{\|\p\|} & \SO(2) \ltimes \RR^3 & (\bj,\bzero) & (\bj,\bzero) & \{(\bj,\bzero)\} & T^*S^2\\ \rowcolor{blue!7}
      7_\pm& (\bj,\bk\neq\bzero,\p\neq\bzero,0)_{\bk \times \p = \bj \times \bk = \bzero} & (\p,0) & S^2_{\|\p\|} & \SO(2) \ltimes \RR^3 & (\bj,\bk) & (\bj,\bk) & \{(\bj,\bk)\} & T^*S^2\\
      8& (\bzero,\bk,\p,0)_{\bk\times\p \neq \bzero} & (\p,0) & S^2_{\|\p\|} & \SO(2) \ltimes \RR^3 & (\bzero,\bk) & (\bzero,\bk) & T^*S^1_{\|\bk\|} & T^*S^2 \times_{S^2} (K \times_{K_\tau} T^*S^1)\\
      \bottomrule
    \end{tabular}
  \end{adjustbox}
  \vspace{1em}
\end{table}

In Part~I we showed that automorphisms of $G$ induce
symplectomorphisms between coadjoint orbits (provided with their
natural $G$-invariant Kirillov--Kostant--Souriau symplectic
structure). Inner automorphisms preserve the coadjoint orbit, whereas
outer automorphisms relate different coadjoint orbits.  For instance,
all the four-dimensional coadjoint orbits of the Carroll group (cases
$5,6,7_\pm$) with the same value of $\|\bj\|$ are related by
automorphisms.  In the same way, automorphisms also relate different
representations of $G$.  If $\rho: G \to U(\eH)$ is a unitary
representation of $G$ on a Hilbert space $\eH$ and
$\varphi \in \Aut(G)$ is an automorphism, we may twist $\rho$ by
$\varphi$ to arrive another representation $\rho^\varphi$ defined
simply by pre-composition: $\rho^\varphi(g) = \rho(\varphi(g))$ for
all $g \in G$.  Notice that by construction, $\rho^\varphi$ is a
representation on the same underlying Hilbert space.  Again if
$\varphi$ is an inner automorphism, so $\varphi(g) = h g h^{-1}$ for
some $h \in G$, then
$\rho^\varphi(g) = \rho(h) \circ \varphi(g) \circ \rho(h)^{-1}$, so
that the two representations are unitarily equivalent.  However if
$\varphi$ is outer, then $\rho$ and $\rho^\varphi$ need not be
equivalent and, indeed, often they are not.

The outer automorphisms of the Carroll group $G$ were determined in
Part~I.  They are given by $\begin{pmatrix} \alpha & \beta \\ \gamma &
  \delta \end{pmatrix} \in \GL(2,\RR)$ acting on $\g$ as follows:
\begin{equation}
  \label{eq:autos-on-g}
  J_i \mapsto J_i,\qquad B_i \mapsto \alpha B_i + \beta P_i, \qquad
  P_i \mapsto \gamma B_i + \delta P_i \qquad\text{and}\qquad H \mapsto
  \Delta H,
\end{equation}
where $\Delta = \alpha\delta - \beta \gamma \neq 0$ the determinant of the
matrix.  The dual action on $\g^*$ is given as follows:  $(\bj,\bk,\p,E) \mapsto
(\bj', \bk', \p', E')$ with
\begin{equation}
  \label{eq:autos-on-g-star}
  \bj' = \bj,\qquad \bk' = \frac{\delta \bk - \beta
    \p}{\Delta},\qquad \p' = \frac{\alpha \p - \gamma \bk}{\Delta}
  \qquad\text{and}\qquad E' = \frac{E}{\Delta}.
\end{equation}
We will find it convenient to also work out the action of
automorphisms on the group $G$ in our choice of parametrisation:
\begin{equation}
  \label{eq:g-param}
  g(R,\bv,\ba,s) = e^{s H} e^{\ba \cdot \bP} e^{\bv \cdot \bB} R.
\end{equation}
One finds after a short calculation (using the
Baker--Campbell--Hausdorff formula) that
\begin{equation}
  \label{eq:autos-on-G}
  g(R,\bv,\ba,s) \mapsto g(R,\gamma \ba + \alpha \bv, \delta \ba +
  \beta \bv, s \Delta + \tfrac12 (\gamma\delta \|\ba\|^2 + \alpha\beta
  \|\bv\|^2) + \beta\gamma \ba\cdot\bv).
\end{equation}

\section{UIRs of the Carroll group}
\label{sec:unit-irred-repr}

We now discuss UIRs of the Carroll group.  Since the Carroll group is
a (regular) semidirect product $K \ltimes T$, with $T$ abelian, it
follows from Mackey's Imprimitivity Theorem (see, e.g.,
\cite[Ch.~17]{MR0495836}) that all such representations are obtained
via the method of induced representations, departing from a unitary
one-dimensional representation of $T$ and a unitary irreducible
representation of its ``little group''.  Furthermore, as shown by
Rawnsley \cite{MR387499}, these are precisely the representations
arising via the geometric quantisation of the coadjoint orbits.  Our
approach is via induced representations, rather than the geometric
quantisation of the coadjoint orbits, but the correspondence with
coadjoint orbits provides a useful guide.  Coadjoint orbits were
described in Section~\ref{sec:review-coadjoint-orbits} and the method
of induced representations is described in
Appendix~\ref{sec:meth-induc-repr}.

It is convenient to consider the universal cover of the
Carroll group, which shares the coadjoint orbits with the Carroll
group.  From here onwards, we shall let $G$ denote the universal cover
of the the Carroll group, whose maximal compact subgroup is
$\Spin(3)\cong \SU(2)$, the universal cover of $\SO(3)$.  Just as the
Carroll group, its universal cover is a semi-direct product
\begin{equation}
  G \cong (\Spin(3) \ltimes \RR^3) \ltimes (\RR^3 \oplus
  \RR) = K \ltimes T,
\end{equation}
where $K$ now denotes the universal cover of the homogeneous Carroll
group (isomorphic to the universal cover of the three-dimensional
euclidean group) and $T \cong \RR^4$ is the abelian normal subgroup of
translations.  We recall that the Casimir elements of the Carroll
group are $H$ (linear) and $\|H \boldsymbol{J} + \boldsymbol{P} \times
\boldsymbol{B}\|^2$ (quartic).

\subsection{Brief recap of the method of induced representations}
\label{sec:brief-recap-method}

Although more details are given in Appendix~\ref{sec:meth-induc-repr},
we briefly recap the method of induced representations for a
semidirect product $K\ltimes T$ with $T$ abelian, emphasising the
procedure, which we list as a sort of algorithm.

\begin{enumerate}[label=(\arabic*)]
\item Pick a complex one-dimensional unitary representation of $T$ or,
  equivalently, an element $\tau \in \t^*$ in the dual of its Lie
  algebra.  Since $T$ is abelian, all complex irreducible
  representations are one-dimensional and the unitary ones are given
  by characters
  \begin{equation}\label{eq:T-char-summary}
    \chi_\tau(\exp X) = e^{i \left<\tau,X\right>}
  \end{equation}
  for all $X \in \t$ and where $\tau \in \t^*$.  Therefore picking a
  one-dimensional unitary representation of $T$ is equivalent to
  picking $\tau \in \t^*$.
  
\item Pick a complex unitary irreducible representation $W$ of the
  stabiliser $K_\tau \subset K$ of $\tau \in \t^*$.  Let
  $\left<-,-\right>_W$ denote a $K$-invariant hermitian inner
  product on $W$.  The $K$-orbit $\eO_\tau$ of $\tau$ is thus
  diffeomorphic to $K/K_\tau$, but since $T$ acts trivially on $\t^*$,
  also diffeomorphic to $G/H$, with $H=K_\tau \ltimes T$.  We will
  assume (and will check) that $\eO_\tau$ admits a $K$-invariant
  measure.  Although $W$ is initially a representation of $K_\tau$, it
  can be seen as a representation of $H$ where $t \in T$ acts via the
  character $\chi_\tau$ defined by $\tau$.

\item Pick a (possibly only locally defined) coset representative
  $\sigma: \eO_\tau \to G$ for the orbit $\eO_\tau \cong G/H$.  Then
  for all $p \in \eO_\tau$ (in the domain of $\sigma$) and all $g \in G$,
  \begin{equation}\label{eq:g-action-on-coset-rep}
    g^{-1} \sigma(p) = \sigma(g^{-1}\cdot p) h(g^{-1},p),
  \end{equation}
  which defines $h(g^{-1},p) \in H$.
\item Let $\psi : \eO_\tau \to W$ and for $g \in G$, define
  \begin{equation}\label{eq:g-action-on-irrep}
    (g \cdot \psi)(p) = h(g^{-1},p) \cdot \psi(g^{-1}\cdot p).
  \end{equation}
  This defines a UIR of $G$ on the Hilbert space of square-integrable
  functions $\eO_\tau \to W$ relative to the inner product
  \begin{equation}\label{eq:hilb-space-inner-product}
    \left(\psi_1,\psi_2\right) = \int_{\eO_\tau} d\mu(p) \left<\psi_1(p),\psi_2(p)\right>_W,
  \end{equation}
  where $d\mu(p)$ is the invariant measure on $\eO_\tau$.
\end{enumerate}

We should remark that the above ``algorithm'' is an
over-simplification and the reader is urged to read
Appendix~\ref{sec:meth-induc-repr} for a more detailed exposition,
from where the above four points have been distilled.  In particular,
the ``functions'' $\psi: \eO_\tau \to W$ are actually sections of a
vector bundle $E_W = K \times_{K_\tau} W$ over $\eO_\tau$ associated to
the representation $W$ of $K_\tau$.  We can also describe this vector
bundle as $G \times_{K_\tau \ltimes T} W$ having extended the action
of $K_\tau$ on $W$ to the action of $K_\tau \ltimes T$ as discussed in
point (2) above.  In Appendix~\ref{sec:meth-induc-repr} we also
remind the reader that sections of $E_W$ can be equivalently described
as (Mackey) functions $K \to W$ which are equivariant under $K_\tau$ or,
even, functions $G \to W$ which are equivariant under $K_\tau \ltimes
T$.  The representation of $G$ carried by the sections of $E_W$ is
much more transparent when recast in the language of Mackey
functions.  It is in this language that the
formula~\eqref{eq:g-action-on-irrep} (which is
equation~\eqref{eq:G-action-sections-final}) is arrived at, departing
from equation~\eqref{eq:G-action-on-sections}, where $F : G \to W$ is
the corresponding Mackey function.

\subsection{Induced representations}
\label{sec:induc-repr-homog}

We shall construct UIRs of $G$ using the method of induced
representations familiar from the case of the Poincaré
group~\cite{Wigner:1939cj} and recalled above in
Section~\ref{sec:brief-recap-method} and in more detail in
Appendix~\ref{sec:meth-induc-repr}.

Let $\tau = (\p,E) \in \t^*$.  This defines a unitary character 
$\chi_\tau$ by
\begin{equation}\label{eq:translation-character}
  \chi_\tau(\ba,s) = e^{i(\p \cdot \ba + E s)}.
\end{equation}
Let $K_\tau \subset K$ denote the stabiliser of $\tau$.  Even though
$K$ is the universal cover of the euclidean group, its action on
$\t^*$ factors through the action of the euclidean group and hence the
$K$-orbit of $\tau$ is again the same $\eO_\tau$ introduced in
Section~\ref{sec:review-coadjoint-orbits}.  It is nevertheless 
$K$-equivariantly diffeomorphic to $K/K_\tau$, even when the
$K$-action is only locally effective.

We now choose a UIR $W$ of $K_\tau$ and construct the homogeneous
vector bundle
\begin{equation}
 E_W := K \times_{K_\tau} W \to \eO_\tau.
\end{equation}
Sections of $E_W$ are locally functions $\psi: \eO_\tau \to W$ and they
carry an action of $G$ as in equation~\eqref{eq:g-action-on-irrep}
which defines a UIR of $G$ on the Hilbert space of square-integrable
sections.

\subsubsection{Invariant measures}
\label{sec:invariant-measures}

The above of course depends on the existence of the invariant measure.
In Table~\ref{tab:structure-orbits} we see that there are three types
of orbits $\eO_\tau$: point-like orbits $\{(\bzero, 0)\}$, 2-spheres
$S^2_{\|\p\|}$ and three-dimensional affine hyperplanes
$\AA^3_{E=E_0}$. Invariant measures are nowhere-vanishing top-rank
forms on $\eO_\tau$ which are $K$-invariant and, by Frobenius
reciprocity, they are in bijective correspondence with
$K_\tau$-invariant elements in
$\wedge^{\mathrm{top}} (\fk/\fk_\tau)^* \cong
\wedge^{\mathrm{top}}\fk_\tau^0$, where $\fk_\tau^0 \subset \fk^*$ is
the annihilator of the Lie algebra $\fk_\tau$ of $K_\tau$ in the dual
of the Lie algebra $\fk^*$ of $K$. We will use this to deduce that the
orbits $\eO_\tau$ in Table~\ref{tab:structure-orbits} admit invariant
measures.

Let $J_i, B_j$ denote a basis for $\fk$ and $\lambda^i, \beta^i$ the
canonical dual basis for $\fk^*$.  Then the coadjoint action is given
by
\begin{equation}
  \begin{aligned}
    \ad^*_{J_i} \beta^j &= \epsilon_{ijk} \beta^k\\
    \ad^*_{J_i} \lambda^j &= \epsilon_{ijk} \lambda^k
  \end{aligned}
  \qquad\qquad
  \begin{aligned}
    \ad^*_{B_i} \beta^j &= -\epsilon_{ijk} \lambda^k\\
    \ad^*_{B_i} \lambda^j &= 0.
  \end{aligned}
\end{equation}
Ignoring the point-like orbit, we see that for the $2$-sphere orbits
$\fk_\tau$ is the span of $J_3, B_i$, whereas $\fk_\tau^0$ is the span
of $\lambda^1,\lambda^2$ and one can check that
$\lambda^1 \wedge \lambda^2 \in \wedge^2\fk_\tau^0$ is
$\fk_\tau$-invariant and hence, since $K_\tau$ is connected, also
$K_\tau$-invariant.  For the affine hyperplane orbits, $\fk_\tau$ is
spanned by $J_i$ and hence $\fk_\tau^0$ is spanned by $\beta^i$ and
it's not hard to see that $\beta^1 \wedge \beta^2 \wedge \beta^3 \in
\wedge^3\fk_\tau^0$ is $K_\tau$-invariant.  We conclude that all
orbits have invariant measures.

\subsection{Inducing representations}
\label{sec:induc-repr}

We must now determine the UIRs $W$ of $K_\tau$, the so-called inducing
representations.  Depending on the orbit, as seen in
Figure~\ref{fig:mom_lim}, we have three possible isomorphism classes
of stabilisers:
\begin{itemize}
\item $K_\tau$: $\Spin(3) \ltimes \RR^3$ for the point-like orbits,
\item  $(\Spin(2) \ltimes \RR^2) \times \RR$ for the 2-spheres and 
\item $\Spin(3)$ for the affine hyperplanes.
\end{itemize}
The simplest case, which has already been discussed
in~\cite{Levy1965}, is that of the affine hyperplanes, since all
irreducible representations of $\Spin(3)$ are well-known: they are
finite-dimensional, unitary and isomorphic to the spin-$s$
representation $V_s$ for some $2s$ a non-negative integer, which is of
dimension $2s + 1$.

\subsubsection{UIRs of $\Spin(3) \ltimes \RR^3$}
\label{sec:unit-irred-repr-2}

The three-dimensional euclidean group is again a semidirect product
and hence we use again the method of induced representations.  Now we let
$\bk \in \RR^3$ and $\chi_{\bk}$ be the unitary character defined by
\begin{equation}
  \chi_{\bk}(\bv) := e^{i \bk \cdot \bv}.
\end{equation}
The group $\Spin(3)$ acts on such characters as
$\chi_{\bk} \mapsto \chi_{R \bk}$, where in the expression $R\bk$ we
understand that $R \in \Spin(3)$ acts through its projection to
$\SO(3)$.  If $\bk = \bzero$ (so that $\chi_{\bk} \equiv 1$) the
induced representation is then simply a UIR of $\Spin(3)$, which as
mentioned above, is one of the spin-$s$ representations $V_s$.

If $\bk \neq \bzero$, the $\Spin(3)$-orbit of $\chi_{\bk}$ is a 2-sphere with
typical stabiliser $\U(1) \subset \SU(2) \cong \Spin(3)$.  The
UIRs of $\U(1)$ are indexed by the integers.  If $\lambda \in \U(1)$,
or equivalently $\lambda \in \CC$ with $|\lambda| = 1$, the
representation indexed by $n\in \ZZ$ is the one-dimensional complex
representation where $\lambda$ acts by multiplication by $\lambda^n$.
Let us call that representation $\CC_n$.

We now define complex line bundles over the 2-sphere associated to
such representations:
\begin{equation}
  \SU(2) \times_{\U(1)} \CC_n \to S^2.
\end{equation}
We can identify these bundles as follows.  First of all notice that
we can identify $\SU(2)$ with the unit sphere in $\CC^2$.  Indeed if
$(z_1,z_2) \in \CC^2$ with $|z_1|^2 + |z_2|^2 = 1$, we form the
special unitary matrix
\begin{equation}
  g(z_1,z_2) :=
  \begin{pmatrix}
    z_1 & \overline z_2\\ -z_2 & \overline z_1
  \end{pmatrix}
\end{equation}
and every special unitary matrix is of this form.  The 2-sphere is the
complex projective line, which is the quotient of
$\CC^2\setminus\{(0,0)\}$ by $\CC^\times = \CC\setminus\{0\}$.  We can
restrict to the 3-sphere in $\CC^2$ and quotient by the action of
$\U(1) \subset \CC$ given by right multiplication as follows:
\begin{equation}
  g(z_1,z_2)
  \begin{pmatrix}
    \lambda & 0 \\ 0 & \overline\lambda
  \end{pmatrix} = 
  \begin{pmatrix}
    z_1 & \overline z_2\\ -z_2 & \overline z_1
  \end{pmatrix}
  \begin{pmatrix}
    \lambda & 0 \\ 0 & \overline\lambda
  \end{pmatrix} = 
  \begin{pmatrix}
    \lambda z_1 & \overline\lambda \overline z_2\\ -\lambda z_2 &
    \overline\lambda \overline z_1
  \end{pmatrix}=  
  g(\lambda z_1, \lambda z_2),
\end{equation}
where $|\lambda|=1$.  Sections of the homogeneous line bundle
$\SU(2)\times_{\U(1)} \CC_n \to S^2$ are $\U(1)$-equivariant functions
$f : \SU(2) \to \CC$ such that $f(gh) = h^{-1}\cdot f(g)$, or
equivalently complex-valued functions of $z_1,z_2$ such that
$f(\lambda z_1, \lambda z_2) = \lambda^{-n} f(z_1,z_2)$.  These are
the sections of the line bundle $\mathscr{O}(-n)$ over $\CP^1$.  We
may define an inner product on the space of sections by integrating
the pointwise inner product on $\CC_n$ against the $\SU(2)$-invariant
measure given by the volume form of a round metric on the 2-sphere.
The resulting induced representation of the three-dimensional
euclidean group is then carried by the square-integrable sections of
$\mathscr{O}(-n)$ over $\CP^1$ for any $n\in \ZZ$.  We will discuss
them in more detail in Section~\ref{sec:uirreps-e=0}.

\subsubsection{UIRs of $(\Spin(2) \ltimes \RR^2)\times \RR$}
\label{sec:unit-irred-repr-1}

The stabiliser now is isomorphic to $(\Spin(2) \ltimes \RR^2)\times \RR$,
where $\Spin(2)$ can be identified with the $\U(1)$ subgroup of
$\SU(2)$ discussed in the previous section.  Indeed, the action of
$\SU(2)$ on $\RR^3$ is the adjoint representation, which is self-dual
and hence isomorphic to the coadjoint representation.  Choosing $\p
\in \RR^3$ to correspond to the Lie algebra element
\begin{equation}
  \begin{pmatrix}
    i p  & 0 \\ 0 & -i p
  \end{pmatrix}
\end{equation}
we see that the stabiliser of this element in $\SU(2)$ is
\begin{equation}
  \left\{
    \begin{pmatrix}
      \lambda & 0 \\ 0 & \overline \lambda
    \end{pmatrix} ~ \middle |~ |\lambda|=1 \right\} \cong \U(1).
\end{equation}

Irreducible representations of $\left( \Spin(2) \ltimes \RR^2\right)
\times \RR$ are tensor products of irreducible representations of
$\Spin(2)\ltimes \RR^2$ and of $\RR$. Complex irreducible
representations of an abelian Lie group are one-dimensional.  The
unitary irreducible representations of $\RR$ are complex
one-dimensional and given by unitary characters labelled by a real
number $w \in \RR$, where $\chi_w(s) = \exp(i w s)$ for all $s \in
\RR$.  It is however more convenient notationally to consider
$\Spin(2) \ltimes \RR^3$ even when $\Spin(2)$ leaves invariant the
third component of the vectors in $\RR^3$.

Let us then determine the unitary irreducible representations of
$\Spin(2) \ltimes \RR^3$.  Being also a semidirect product, we apply
again the method of induced representations. Let again $\bk \in \RR^3$
and $\chi_{\bk}$ be the unitary character:
\begin{equation}
  \chi_{\bk}(\bv) := e^{i \bk \cdot \bv}.
\end{equation}
The group $\Spin(2)$ acts on such characters by restricting the
adjoint action of $\SU(2)$.  If we take $\bk = (k_1, k_2,k_3)$, then
$\lambda \cdot \bk = \bk'$ where $\bk'= (k'_1,k'_2,k'_3)$ with
$k'_3 = k_3$ and
\begin{equation}\label{eq:weight-two}
  k'_1 + i k'_2 = \lambda^2 (k_1 + i k_2)
\end{equation}
as shown by the conjugation, where we have used that $\overline\lambda
= \lambda^{-1}$:
\begin{equation}
  \begin{pmatrix}
    \lambda & 0 \\ 0 & \lambda^{-1}
  \end{pmatrix}
  \begin{pmatrix}
    k_3 & k_1 + i k_2 \\ k_1 - i k_2 & - k_3
  \end{pmatrix}
  \begin{pmatrix}
    \lambda^{-1} & 0 \\ 0 & \lambda
  \end{pmatrix} =
  \begin{pmatrix}
    k_3 & \lambda^2 (k_1 + i k_2) \\ \lambda^{-2} (k_1 - i k_2) & -k_3
  \end{pmatrix}.
\end{equation}

Let us use the notation $\bk^\perp = (k_1, k_2, 0)$ to denote the
component of $\bk$ orthogonal to $\p$.  If $\bk^\perp = \bzero$, then
its $\Spin(2)$-orbit is a single point.  The stabiliser is $\Spin(2)$
itself and hence we choose a UIR $\CC_n$ of $\Spin(2)$, as already
discussed above.

If $\bk^\perp \neq \bzero$, the $\Spin(2)$-orbit of $\chi_{\bk}$ is a circle and
the stabiliser consists of all those $\lambda \in \Spin(2)$ with
$\lambda^2 = 1$; that is the subgroup corresponding to $\pm 1$, which
is of course isomorphic to $\ZZ_2$.  There are two irreducible
representations of $\ZZ_2$, both unitary and one-dimensional,
depending on whether $-1$ acts as $1$ or as $-1$: they are call,
respectively, the trivial and sign representations and we will denote
them by $\CC_\pm$, respectively.  The associated homogeneous line
bundles have an equivalent characterisation which may be more
familiar.   The group $\Spin(2)$ is the total space of the spin bundle
over the circle $\Spin(2)/\ZZ_2$ and the homogeneous line bundles
associated to the representations $\CC_\pm$ are the corresponding
spinor bundles $\Sigma_\pm \to S^1$.  Sections of $\Sigma_\pm$ are
typically known as Ramond spinors (for $+$) and Neveu--Schwarz spinors
(for $-$) on the circle.  The representation space of $\Spin(2)\ltimes
\RR^2$ is then the Hilbert space $L^2(S^1,\Sigma_\pm)$ of
square-integrable spinor fields on the circle.  But of course we are
interested in $\Spin(2) \ltimes \RR^3$.  The third component acts via
a character as explained above and hence we get a (trivial) line
bundle $L_w \to S^1$ by which we may twist the spinors.

In summary, the UIRs of $(\Spin(2)\ltimes \RR^2) \times \RR$ come in
several types:
\begin{itemize}
\item one-dimensional representations $\CC_n \otimes \CC_w$ where $n
  \in \ZZ$ and $w \in \RR$; and
\item infinite-dimensional representations $L^2(S^1, \Sigma_\pm
  \otimes L_w)$ for $w \in \RR$, where $L_w$ is a trivial line bundle
  over $S^1$.
\end{itemize}
We will discuss them in more detail in Section~\ref{sec:uirreps-e=0}.

\subsection{UIRs of the Carroll group}
\label{sec:summary}

We may summarise the above discussion by listing the UIRs of (the
universal cover of) the Carroll group and in so doing we shall give
them names.

\subsubsection{UIRs with $E_0 \neq 0$}
\label{sec:uirreps-eneq0}

For representations with the value $E_0$ of $H$ nonzero, the Hilbert
space consists of the square-integrable functions $\AA^3 \to V_s$,
with $\AA^3 \subset \t^*$ the affine hyperplane with $E =E_0$ and
$V_s$ the complex spin-$s$ representation of $\Spin(3)\cong \SU(2)$ of
dimension $2s + 1$, relative to the inner product
\begin{equation}\label{eq:massive-inner-product}
 (\psi_1,\psi_2) = \int_{\AA^3} d^3p \left<\psi_1(\p),\psi_2(\p)\right>_{V_s} \, ,
\end{equation}
where $\left<-,-\right>_{V_s}$ is an invariant hermitian inner product on
$V_s$.  These representations were already discussed in the original work
of Lévy-Leblond~\cite{Levy1965}. They admit field-theoretic
realisations on Carroll spacetime as we will review below in
Section~\ref{sec:an-explicit-example}.  We shall call denote these
UIRs by $\Romanbar{II}(s,E_0)$ with the understanding that $E_0 \neq 0$.

Let us write the explicit action of the Carroll group $G$ on these
representations.  The orbit $\eO_\tau \subset \t^*$ is the affine
hyperplane $\AA^3 \subset \t^*$ consisting of points $(E_0,\p)$ where
$\p \in \RR^3$. Let us choose the point $(E_0,\bzero)$ as the origin
and let us choose a coset representative $\sigma(\p) \in K$ so that
$\sigma(\p) \cdot (E_0,\bzero) = (E_0,\p)$. A quick calculation shows
that $\sigma(\p) = \exp(-\frac1{E_0}\p\cdot \bB)$ works.  In this
paper and in contrast with Part~I, we work with a different, more
standard, parametrisation of the Carroll group.  Let us
factorise $g = t k$, with $t \in T$ and $k \in K$ as follows:
\begin{equation}
  \label{eq:standard-param}
  g(R,\bv,\ba,s)= \underbrace{e^{sH + \ba \cdot \bP}}_{t \in T}
  \underbrace{e^{\bv\cdot \bB} R}_{k \in K}.
\end{equation}
We have that with such $g$,
\begin{equation}
  \begin{split}
    (g \cdot \psi)(\p) &= \chi_{(E_0,\p)}(t) (k \cdot \psi)(\p),
  \end{split}
\end{equation}
where $\chi_{(E_0,\p)}$ is the character given by
\begin{equation}
  \chi_{(\p,E_0)}(t) =e^{i (E_0 s + \ba \cdot \p)}
\end{equation}
and we calculate
\begin{equation}
  k^{-1} \sigma(\p) = \sigma(R^{-1}(\p + E_0 \bv)) R^{-1},
\end{equation}
so that from equation~\eqref{eq:g-action-on-irrep} we arrive at the
explicit expression for the action of the group element
$g(R,\bv,\ba,s) \in G$ on $\psi \in L^2(\AA^3,V_s)$:
\begin{equation}
  \label{eq:g-rep-massive}
  (g \cdot \psi)(\p) = e^{i ( E_0 s + \p \cdot \ba)} \rho(R) \psi(R^{-1}(\p + E_0\bv)) \, ,
\end{equation}
where $R \mapsto \rho(R)$ denotes the spin-$s$ representation of
$\Spin(3)$. It is understood that we are on the hyperplane where the
energy is restricted to $E=E_{0}$.

Let us write down the hermitian operators corresponding to angular
momentum $\hat \J$, energy $\hat H$, momentum $\hat{\p}$ and
centre-of-mass $\hat{\bm{B}}$. They are related to the Lie algebra
generators~\eqref{eq:3-carroll-algebra} via multiplication by $i$,
explicitly $X_{\mathrm{LieAlg}} = i \hat X$. For the representation at
hand they are given by
\begin{align}
  \label{eq:op-p-basis}
\hat \J &= - i \p \times \frac{\pd}{\pd \p} + \hat{\bm{S}}  & \hat{\bm{B}} &= -i E_{0} \frac{\pd}{\pd \p} &  \hat H &= E_{0} & \hat{\bm{P}}&= \p   \, .
\end{align}
Here $\hat{\bm{S}}$ are the infinitesimal generators of the spin-$s$
representation $\rho(R)$. Together with $\hat H$ we can use them to
uniquely label the $\Romanbar{II}(s,E_0)$ representation since
\begin{align}
  \hat H &= E_{0} & \hat{\bm{S}}^{2} &= s (s+1) \, ,
\end{align}
which are multiples of the identity.

For massive carrollions we can also define a position operator
$\hat{\bm{X}}$~\cite{Levy1965}
\begin{align}
  \hat{ \bm{X}} = \frac{1}{E_{0}} \hat{\bm{B}}
\end{align}
which transforms as expected under rotations and spatial translations.
This definition also agrees with the intuition that the centre of mass
of a massive Carroll particle is the energy multiplied by the position
(classically written as $\bm{k} = E_{0} \bm{x}$ as in, e.g., Part~I,
Section 3). When evaluated on the wavefunctions we recover the
canonical commutation relations
\begin{align}
  [\hat X_{i},\hat P_{j}] = - i \delta_{ij} \, .
\end{align}

We could have chosen to diagonalise not the momentum operator
$\hat{\bm{P}}$, but with respect to the centre-of-mass
$\hat{\bm{B}}$. In this case this leads us to ``boost'' or
centre-of-mass wavefunctions on which the symmetries act as
\begin{equation}
  \label{eq:g-rep-massive-boost-again}
  (g \cdot \tilde\psi)(\bk) = e^{i ( E_0 s + \bk \cdot \bv)} \rho(R) \tilde\psi(R^{-1}(\bk - E_0\ba)) \, .
\end{equation}
with 
\begin{align}
  \label{eq:op-b-basis}
\hat \J &= - i \bk \times \frac{\pd}{\pd \bk} + \hat{\bm{S}}  & \hat{\bm{B}} &= \bk &  \hat H &= E_{0} & \hat{\bm{P}}&=  i E_{0}\frac{\pd}{\pd \bk}   \, .
\end{align}
The momentum and boost eigenstates can be shown to be related via
Fourier transforms
\begin{align}
  \label{eq:psi-fourier}
  \tilde \psi(\bk) &= \int d^{3}p \, e^{-\frac{i}{E_{0}} \bk \cdot \p} \psi(\p) &  \psi(\p) &=\frac{1}{(2\pi |E_{0}|)^{3}} \int d^{3}k \, e^{\frac{i}{E_{0}} \bk \cdot \p} \tilde \psi(\bk)
\end{align}
and the inner product is given by
\begin{equation}
  \label{eq:k-massive-inner-product}
  (\tilde\psi_1,\tilde\psi_2) = \int_{\AA^3} d^3k \left<\tilde\psi_1(\bk),\tilde\psi_2(\bk)\right>_{V_s} \, .
\end{equation}
Here $\AA^{3}$ is now the $E=E_{0}$ hyperplane but in $(E,\bk)$ space.
The relation between the momentum and boost eigenstates is analogous
to the relation between momentum and position space eigenstates in
nonrelativistic (galilean) quantum mechanics. This can be traced back
to the commutation relation $[B_{i},P_{j}] = \delta_{ij}H$ which
mirrors the canonical commutation relation between position and
momentum operators with the energy playing the rôle of $\hbar$. This
also implies that the momentum and the centre-of-mass representations
are Fourier transforms (in $\p$-$\bk$ space) of each other.

Another way to see this relation between the momentum and boost basis
is to look at the particle action, as in~Section 3.1 in Part~I, for
instance.  The kinetic term in the canonical lagrangian can be written
as being proportional to $\p \cdot \dot \bk$. The analogous choice in
the path integral quantisation approach is then the choice to
calculate amplitudes either with regard to eigenstates of
$\hat{\bm{P}}$ or $\hat{\bm{B}}$.

\subsubsection{UIRs with $E_0 = 0$}
\label{sec:uirreps-e=0}

There are four classes of UIRs with $E_0 = 0$:
\begin{enumerate}[label=(\alph*)]
\item any finite-dimensional representation $V_s$ of $\Spin(3)$ with all
  other generators acting trivially;
\item the square integrable sections of $\mathscr{O}(-n)$ over the
  2-sphere for any $n\in \ZZ$, with the translations of the Carroll
  group acting trivially;
\item the square-integrable sections of a Hilbert bundle over the
  2-sphere, whose fibres are the square-integrable spinors (with
  respect to either of the two spin structures) on the circle twisted
  by a trivial line bundle $L_w$: $L^2(S^1, \Sigma_{\pm} \otimes L_w)$;
\item and the square-integrable sections of the line bundle over the
  2-sphere associated to the one-dimensional representation $\CC_n
  \otimes \CC_w$ of $(\Spin(2)\ltimes \RR^2)\times \RR$.
\end{enumerate}

We shall now discuss them in some detail and will discuss the possible
field theoretical realisations of some of these representations in
Section~\ref{sec:another-example}.

Representations of class (a) require no further discussion.  We shall
denote them by $\Romanbar{I}(s)$ with the understanding that $2s \in
\NN_0$.

\paragraph{Representations of classes (b) and (d).}

Let us first of all consider representations of class (d), under the
assumption that $w=0$, since as we shall see in
Section~\ref{sec:another-example}, these are the ones which can be
realised as finite-component fields on Carroll spacetime.  We take
$\tau = (0,\p)$ with $\p=(0,0,p)$ and $p>0$.  The stabiliser $K_\tau$
consists of the boosts and $\U(1)$ subgroup of $\SU(2)$ consisting of
diagonal matrices and $\eO_\tau$ is the $2$-sphere of radius $p$ in
$\RR^3$.  If we think of the sphere as the extended complex plane, we
can effectively work in the complex plane.  The round metric on the
unit sphere pulls back to the Fubini--Study metric (up to a factor of 4):
\begin{equation}
  \label{eq:fubini-study}
  g_{FS} = \frac{4 dz d\zbar}{(1 + |z|^2)^2}
\end{equation}
whose associated volume form is
\begin{equation}
  \label{eq:volume-form-sphere}
  \omega =\frac{2i dz \wedge d\zbar}{(1+|z|^2)^2}.
\end{equation}
As a coset representative $\sigma : \CC \to \SU(2)$ we may
take\footnote{Here and in the sequel we think of $z$ as a point in the
complex plane and not as a holomorphic coordinate.  Hence the notation
$\sigma(z)$ or $\psi(z)$ is not meant to denote a holomorphic function, but
simply a smooth function on the complex plane, as can be seen from the
explicit form of $\sigma(z)$.}
\begin{equation}\label{eq:sigma-z}
  \sigma(z) = \frac1{\sqrt{1+|z|^2}}
  \begin{pmatrix}
    z & -1 \\ 1 & \zbar
  \end{pmatrix}.
\end{equation}
Notice that the map $\eO_\tau \to \CC$ is the stereographic projection
from $\tau$.  Hence the complex plane $\CC$ parametrises the orbit
$\eO_\tau$ excised of the actual point $\tau$, which we only recover
in the limit $z \to \infty$.  (This may seem a little strange, but it
is fine.)  The action of $\sigma(z)$ on $\p = (0,0,p)$ results in
$\q(z)= \sigma(z)\cdot \p$.  To work out the expression for $\q(z)$,
let us identify $\RR^3$ with the space of hermitian traceless matrices
in such a way that $\p = (p_1,p_2,p_3)$ is represented by the matrix
\begin{equation}
  \begin{pmatrix}
    p_3 & p_1 + i p_2\\
    p_1 - i p_2 & -p_3
  \end{pmatrix}
\end{equation}
and the action of $\SU(2)$ on such hermitian matrices is via matrix
conjugation.  This is a linear action which preserves the trace (which
is zero) and the determinant (which is $-\|\p\|^2$).  For the chosen
$\tau$, we have $\p = (0,0,p)$, so that the matrix corresponding to $\q(z)$ is given by
\begin{equation}
  \begin{pmatrix}
    \pi_3(z) & \pi_1(z) + i \pi_2(z)\\
    \pi_1(z) - i \pi_2(z) & -\pi_3(z)
  \end{pmatrix} =
  \frac1{1+|z|^2}
  \begin{pmatrix}
    z & -1 \\ 1 & \zbar
  \end{pmatrix}
  \begin{pmatrix}
    p & 0\\
    0 & -p
  \end{pmatrix}
  \begin{pmatrix}
    \zbar & 1 \\ -1 & z
  \end{pmatrix},
\end{equation}
resulting in
\begin{equation}
  \label{eq:q-vector}
    \pi_1(z) = \frac{2p \Re(z)}{1+|z|^2}, \qquad 
    \pi_2(z) = \frac{2p \Im(z)}{1+|z|^2} \qquad\text{and}\qquad
    \pi_3(z) = \frac{(|z|^2-1)p}{1+|z|^2}.
\end{equation}
As expected, it satisfies $\|\q(z) \|^{2}=p^{2}$.

We consider functions $\psi : \CC \to \CC_n$, where $\CC_n$ is a copy
of the complex numbers with $n$ reminding us how $U(1)$ acts.  We
introduce the inner product
\begin{equation}
  \label{eq:inner-product-sphere}
  \left<\psi_1, \psi_2\right> := \int_\CC \frac{2i dz \wedge d\zbar}{(1+|z|^2)^2} \overline{\psi_1(z)}\psi_2(z)
\end{equation}
and we let $\mathscr{H}$ denote the Hilbert space of square-integrable
such functions.  This space carries a UIR of the Carroll group which
we now exhibit.  Let $g = g(R,\bv,\ba,s) \in G$ be given as in
equation~\eqref{eq:standard-param}.  Then writing $g = t k$,
\begin{equation}
  (g \cdot \psi)(z) = \chi_{(0,\q(z))}(t) (k \cdot \psi)(z).
\end{equation}
Let us write $k = R \beta$ with $R \in \SU(2)$ and $\beta$ a boost.   Then
\begin{equation}
  k^{-1} \sigma(z) = \beta^{-1} R^{-1} \sigma(z).
\end{equation}
We will first work out $R^{-1}\sigma(z)$.  Let's take
\begin{equation}\label{eq:R-and-inverse}
  R =
  \begin{pmatrix}
    \eta & \xi \\ - \overline \xi & \overline \eta
  \end{pmatrix} \implies R^{-1} =
  \begin{pmatrix}
    \overline \eta & - \xi \\ \overline{\xi} & \eta
  \end{pmatrix} \qquad\text{with}\qquad |\eta|^2 + |\xi|^2 = 1.
\end{equation}
A short calculation shows that
\begin{equation}
  \label{eq:rotated-coset-rep}
  R^{-1} \sigma(z) = \sigma(w) \lambda(R,z),
\end{equation}
where
\begin{equation}
  \label{eq:w-and-lambda}
  w = \frac{\overline \eta  z - \xi}{\eta + \overline \xi z}
  \qquad\text{and}\qquad \lambda(R,z) = \frac{\eta + \overline \xi
    z}{\left|\eta + \overline \xi z\right|}.
\end{equation}
Notice that $z \mapsto w$ is the fractional linear transformation
associated to $R^{-1}$, as expected.  Therefore,
\begin{equation}
  k^{-1} \sigma(z) = \beta^{-1} \sigma(w) \lambda(R,z) = \sigma(w)
  \underbrace{\sigma(w)^{-1} \beta^{-1} \sigma(w)}_{\beta'(w)} \lambda(R,z),
\end{equation}
where $\beta'(w)$ is another boost.  Since boosts act trivially on the
inducing representation $\CC_n$ and $\lambda \in \U(1)$ acts like
$\lambda^n$, we arrive at the following action of $g = g(R,\bv,\ba,s)$
on $\psi : \CC \to \CC_n$:
\begin{equation}
  \label{eq:g-action-on-massless-field}
  (g \cdot \psi)(z) = e^{i \ba \cdot \q(z)} \left( \tfrac{\eta +
      \overline{\xi} z}{| \eta + \overline{\xi} z|} \right)^{-n}
  \psi\left( \tfrac{\overline{\eta}z -\xi}{\eta + \overline{\xi} z} \right)\, ,
\end{equation}
with $\q(z)$ given by equation~\eqref{eq:q-vector} and $R \in\SU(2)$
is given by equation~\eqref{eq:R-and-inverse}. Notice that time
translations and carrollian boosts act trivially on momenta when the
energy vanishes, so there is no $s$ and $\bv$ on the right-hand side.
We denote these representations by $\Romanbar{III}(n,p)$ with the
understanding that $n\in\ZZ$ and $p>0$.

By applying an automorphism $\varphi: G\to G$, we may obtain other
representations from this one simply by precomposing: $\rho:G \to
\U(\eH)$ changes to $\rho \circ \varphi : G \to \U(\eH)$.  The outer
automorphism $\begin{pmatrix}\alpha & \beta \\ \gamma &
  \delta\end{pmatrix} \in \GL(2,\RR)$ acts on $G$ as in
equation~\eqref{eq:autos-on-G} and hence from
equation~\eqref{eq:g-action-on-massless-field} we read off the
following representation
\begin{equation}
  \label{eq:g-action-on-massless-field-after-autos}
  (g \cdot \psi)(z) = e^{i (\delta \ba + \beta \bv)\cdot \q(z)} \left( \tfrac{\eta +
      \overline{\xi} z}{| \eta + \overline{\xi} z|} \right)^{-n}
  \psi\left( \tfrac{\overline{\eta}z -\xi}{\eta + \overline{\xi} z} \right) \, ,
\end{equation}
By taking $\delta = 0$ (and hence $\beta \neq 0$) we obtain the
representation of class (b) in the list at the start of this section,
since in this case the translations act trivially.  We denote those
representations by $\Romanbar{III}'(n,k)$, with the understanding that
$n\in\ZZ$ and $k>0$.

Similarly, taking $\delta$ and $\beta$ both nonzero, we obtain the representation of
class (d) with $w = \beta$.  We denote these representations by
$\Romanbar{IV}_\pm(n,p,k)$ where the sign is the sign of $w$, and with
the understanding that $n\in\ZZ$ and $p,k>0$.  It may seem a little
odd that whereas twisting by an automorphism results in a
representation with the same underlying vector space, our description
of the representations of class (d) at the start of this section
exhibits such representations in terms of sections of a different
homogeneous line bundle.  The conjectural resolution is that these
representations are unitarily equivalent.

\paragraph{Representations of class (c).}

It remains to discuss representations in class (c).  The description
which follows by adhering to the method of induced representations
seems a little exotic, since the representation is described as being
carried by sections of an infinite-rank Hilbert bundle over the
$2$-sphere.  We will show, however, that there is an equivalent
description of these representations as honest functions on the
round $3$-sphere with values in a one-dimensional representation of
the nilpotent subgroup of the Carroll group generated by boosts and
translations.

To see this it is perhaps convenient to briefly recapitulate how one
might arrive at such a description.  We start by following the
description of the representation as square-integrable sections of a
Hilbert bundle over the $2$-sphere.  The fibre of the Hilbert bundle
in an infinite-dimensional Hilbert space $\eH^\pm_{k,\theta}$ which is
a UIR of $\Spin(2) \ltimes \RR^3$ and is described as follows.  We
pick a unitary character $\chi_{\bk}:\RR^3 \to \U(1)$ of $\RR^3$ given
by $\chi_{\bk}\left(e^{\bv\cdot \bB}\right) = e^{i \bk \cdot \bv}$,
with our chosen $\bk = (k\sin\theta, 0, k \cos\theta)$ with $k>0$ and
$\theta\in(0,\pi)$.  Remember that $\Spin(2)$ is the subgroup of
$\SU(2)$ consisting of diagonal matrices, so they are labelled by a
complex number $\zeta$, say, of unit modulus.  The orbit of the
character $\chi_{\bk}$ under $\Spin(2)$ is a circle of characters
$\chi_{\bk(\zeta)}$ where $\bk(\zeta) = (k_1(\zeta),k_2(\zeta),k
\cos\theta)$, where
\begin{equation}
  k_1(\zeta) + i k_2(\zeta) = \zeta^2 k\sin\theta,
\end{equation}
as was seen in equation~\eqref{eq:weight-two}.
The stabiliser of $\bk$ in $H := \Spin(2) \ltimes \RR^3$ is the subgroup
$H_{\bk}:= \ZZ_2 \times \RR^3$ consisting of elements $e^{\bv \cdot \bB} (\pm
\1)$, with $\1$ the identity matrix.  The UIRs of $H_{\bk}$ are
complex one-dimensional $\CC_\pm \otimes \CC_{\bk}$, where $\CC_\pm$ are
the trivial and sign representations of $\ZZ_2$, respectively, and
$\CC_{\bk}$ is the one-dimensional unitary representation of $\RR^3$
given by with character $\chi_{\bk}$.  Let us choose a coset
representative $\nu : S^1 \to \Spin(2)$, sending $\zeta \mapsto
\nu(\zeta)$, so that $\nu(\zeta) \bk = \bk(\zeta)$.  One such
possibility is
\begin{equation}
  \label{eq:nu-zeta}
  \nu(\zeta) = \begin{pmatrix} \zeta^{1/2} & 0 \\ 0 &  \zeta^{-1/2}\end{pmatrix},
\end{equation}
for some choice of square root function.  Then we may describe $\psi
\in \eH^\pm_{k,\theta}$ as complex-valued functions $\psi(\zeta)$ on
the circle with the inner product
\begin{equation}
  \left<\psi_1,\psi_2\right> = \int_{S^1} \frac{d\zeta}{i\zeta}
  \overline{\psi_1(\zeta)} \psi_2(\zeta).
\end{equation}
The unitary action of $h = e^{\bv \cdot \bB} \begin{pmatrix}\lambda &
  0 \\ 0 & \lambda^{-1} \end{pmatrix}\in \Spin(2) \ltimes \RR^3$ on
$\psi$ is given by (in the case of Ramond spinors, for definiteness)
\begin{equation}
  (h \cdot \psi)(\zeta) = e^{i \bv \cdot \bk(\zeta)} \psi(\lambda^{-2}
  \zeta).
\end{equation}

The UIR of the Carroll group is carried by sections over the Hilbert
bundle over the $2$-sphere $S^2_p \in \RR^3$ where $p= \|\p\| > 0$,
associated to the UIR $\eH^\pm_{k,\theta}$ of the stabiliser of $\p =
(0,0,p)$.  We denote these UIRs by $\Romanbar{V}_\pm(p,k,\theta)$ with
the understanding that $p,k>0$ and $\theta\in(0,\pi)$.

Locally, relative to a stereographic coordinate $z \in \CC$ for the
sphere, they are defined as functions
$\psi : \CC \to \eH^\pm_{k,\theta}$.  This means that
$\psi(z) \in \eH^\pm_{k,\theta}$ so it is itself a function
$\zeta \mapsto \psi(z)(\zeta)$ as described above.  Such a function is
the result of currying a function $\Psi : \CC \times S^1 \to \CC$;
that is, $\psi(z) = \Psi(z,-)$.  In this description, translations and
boosts are simultaneously diagonalised, which is possible for
representations with $E=0$, since in that case they commute.  We
therefore have that
\begin{equation}
  \begin{split}
    \left( e^{\ba \cdot \bP} \cdot \Psi \right)(z,\zeta) &= e^{i \ba \cdot \q(z)}\Psi(z,\zeta)\\
    \left( e^{\bv \cdot \bB} \cdot \Psi \right)(z,\zeta) &= e^{i \bv \cdot \bk(z,\zeta)} \Psi(z,\zeta),
  \end{split}
\end{equation}
where $\q(z) = \sigma(z)\p$ and $\bk(z,\zeta) = \sigma(z)\bk(\zeta)$.
Notice that $\q(z) \cdot \bk(z,\zeta) = \sigma(z)\p \cdot \sigma(z)\bk(\zeta) = \p \cdot
\bk(\zeta) = p k \cos\theta$.  Time translations act trivially, of
course, on massless representations, so it remains to discuss the
action of rotations. Let $R \in \SU(2)$ be given by
equation~\eqref{eq:R-and-inverse}. Then from
equation~\eqref{eq:rotated-coset-rep} we have that
\begin{equation}
  (R \cdot \Psi)(z,\zeta) = \lambda(R,z)^{-1} \Psi(w,\zeta),
\end{equation}
where $\lambda(R,z)$ and $w$ are given in
equation~\eqref{eq:w-and-lambda}.  Using the action of $\Spin(2)$
on the circle $S^1$, we arrive at
\begin{equation}
  (R \cdot \Psi)(z,\zeta) =\Psi\left( \frac{\overline\eta z -
      \xi}{\eta + \overline\xi z}, \left(\frac{\eta + \overline\xi
        z}{|\eta + \overline\xi z|} \right)^2\zeta \right).
\end{equation}
Putting it all together we arrive at the action of $g =
g(R,\bv,\ba,s)$ in equation~\eqref{eq:standard-param} in this
representation
\begin{equation}
  \label{eq:g-rep-class-iii}
  (g \cdot \Psi)(z,\zeta) = e^{i\left( \ba \cdot \q(z)  + \bv \cdot
      \bk(z,\zeta) \right)} \Psi\left( \frac{\overline\eta z - \xi}{\eta
      + \overline\xi z}, \left(\frac{\eta + \overline\xi z}{|\eta +
        \overline\xi z|} \right)^2\zeta \right).
\end{equation}
Finally, this representation is unitary relative to the hermitian
inner product given by
\begin{equation}
  \label{eq:inner-product-class-iii}
  \left<\Psi_1,\Psi_2\right> = \int_\CC \frac{2 dz \wedge
    d\zbar}{(1+|z|^2)^2} \int_{S^1} \frac{d\zeta}{\zeta} \overline{\Psi_1(z,\zeta)}\Psi_2(z,\zeta).
\end{equation}
We will now describe these representations in a simpler fashion.

\subsubsection{A simpler description of the massless UIRs}
\label{sec:simpl-descr-massl}

As shown in Appendix~\ref{sec:hopf-charts-su2}, the domain of
integration $\CC \times S^1$ in
equation~\eqref{eq:inner-product-class-iii} is a Hopf chart on the
$3$-sphere and the measure of integration is nothing but the volume
form of a round metric on the $3$-sphere, given in
equation~\eqref{eq:volume-form-round-3-sphere}.  This suggests very
strongly an equivalent (and perhaps less exotic) characterisation of
these UIRs, which we now describe.

Let $N\subset G$ denote the nilpotent subgroup of the Carroll group
generated by boosts and translations and let $\chi : N \to \U(1)$ be
defined by
\begin{equation}
  \label{eq:nilpotent-char}
  \chi\left( e^{sH + \ba \cdot \bP} e^{\bv \cdot \bB} \right) =
  e^{i\left( \ba \cdot \p + \bv \cdot \bk \right)},
\end{equation}
where $\p = (0,0,p)$ and $\bk = (k \sin\theta, 0, k\cos\theta)$, as
above.  Although $N$ is not abelian, $\chi$ does define a
one-dimensional representation since in any representation of $N$
where $H$ acts trivially, $N$ acts effectively as an abelian group.
Notice that $G/N \cong \SU(2)$, which is diffeomorphic to the
$3$-sphere.  The character $\chi$ of $N$ defines a homogeneous line
bundle $L_\chi := G \times_N \CC_\chi$ on $S^3$, where $\CC_\chi$ is a
copy of $\CC$ where $N$ acts via $\chi$.  Every smooth line bundle
over $S^3$ is trivial\footnote{We may trivialise the bundle on each
  hemisphere, with transition function from the equatorial $2$-sphere
  to the structure group.  Homotopic transition functions define
  isomorphic bundles, but the second homotopy group of any Lie group
  is trivial, so we may extend the trivialisation to the whole
  $S^3$.}, and hence sections of $L_\chi$ are just smooth functions
$S^3 \to \CC_\chi$.  Let $\eH = L^2(S^3,\CC_\chi)$ denote the
square-integrable such functions relative to the measure coming from a
round metric.  We define a coset representative $\sigma: S^3 \to G$ to
be the identification $S^3 \cong \SU(2)\subset G$.  So we can actually
write $\Psi(S)$ with $S \in \SU(2)$ and if we write
$g = g(R,\ba,\bv,s)$ as in equation~\eqref{eq:g-param}, its action on
such functions is given, as described in
Appendix~\ref{sec:meth-induc-repr} (see, e.g.,
equation~\eqref{eq:G-action}), by
\begin{equation}
  \label{eq:G-action-3-sphere}
  (g \cdot \Psi)(S) = e^{i\left( \ba \cdot S\p + \bv \cdot S\bk
    \right)} \Psi(R^{-1}S),
\end{equation}
which is manifestly unitary relative to the inner product
\begin{equation}
  \label{eq:inner-product-3-sphere}
  \left<\Psi_1,\Psi_2\right> = \int_{S^3} d\mu(S)
  \overline{\Psi_1(S)}\Psi_2(S),
\end{equation}
where $d\mu$ is a bi-invariant Haar measure on $\SU(2)$.  Restricting
to one of the Hopf charts $\CC \times S^1$ as described in
Appendix~\ref{sec:hopf-charts-su2}, we indeed recover the action of
$G$ given by equation~\eqref{eq:g-rep-class-iii} and the inner product
in equation~\eqref{eq:inner-product-class-iii}.  There is one small
detail: the representation just constructed is not actually
irreducible.  Why?  Because of the action of the centre of $\SU(2)$.
Given any function $\Psi: \SU(2) \to \CC$ we can decompose it into a
sum $\Psi_+ + \Psi_-$ of such functions where
$\Psi_\pm(-S) = \pm \Psi_{\pm}(S)$; namely,
\begin{equation}
  \Psi(S) = \underbrace{\tfrac12 (\Psi(S) + \Psi(-S))}_{=:\Psi_+(S)} +
  \underbrace{\tfrac12 (\Psi(S) - \Psi(-S))}_{=:\Psi_-(S)}.
\end{equation}
This decomposes $L^2(S^3,\CC_\chi)$ into the direct sum of two
orthogonal\footnote{This follows from the invariance of the Haar
  measure under multiplication by $-\1$:  $d\mu(-S)=d\mu(S)$.}
subspaces:
\begin{equation}
  L^2(S^3,\CC_\chi) =  L_+^2(S^3,\CC_\chi) \oplus L_-^2(S^3,\CC_\chi),
\end{equation}
where $\Psi \in L_\pm^2(S^3,\CC_\chi)$ if and only if $\Psi(-S) = \pm
\Psi(S)$.  This sign, of course, is the same sign distinguishing the
two spin structures on the circle.  The UIRs have underlying Hilbert
spaces $L_\pm^2(S^3,\CC_\chi)$ and they agree of course with the UIRs
denoted $\Romanbar{V}_\pm(p,k,\theta)$ above.

One could be forgiven for asking why we did not describe other massless
representations of $G$ in terms of functions on the $3$-sphere.  In 
fact, as we now explain, in a sense we have.  Let $\p,\bk \in
\RR^3$ be arbitrary and consider the unitary character of $N$ defined
by equation~\eqref{eq:nilpotent-char}.  What is the stabiliser $H_\chi
\subset \SU(2)$ of $\chi$?  It clearly depends on $\p$ and $\bk$.  The
stabiliser always includes the centre $\ZZ_2$, since the action of
$\SU(2)$ on $\RR^3$ is the adjoint representation and the centre lies
in its kernel, but it may be larger.  In all cases, Mackey's method
would instruct us to induce a representation of $G$ as
square-integrable sections of a homogeneous vector bundle over
$\SU(2)/H_\chi = G/(H_\chi \ltimes N)$ associated to a UIR of $H_\chi
\ltimes N$, obtained from a complex UIR of $H_\chi$ by having $N$ act
via the character $\chi$.  Let us now go through the different
possibilities.

\begin{itemize}
\item If $\p = \bk = \bzero$, then the character $\chi \equiv 1$ and
  $H_\chi = \SU(2)$. Thus we induce a representation of $G$ from a UIR
  of $\SU(2)$ (since $N$ acts trivially) which is carried by
  sections of the corresponding homogeneous vector bundle over
  $\SU(2)/\SU(2)$, which is a point, over which a vector bundle is
  just a vector space.  In other words, the representation is simply
  an UIR of $\SU(2)$ with $N$ acting trivially.  These are the Carroll
  UIRs of class $\Romanbar{I}(s)$.
  
\item If $\p,\bk$ are collinear (but not both zero), $\chi$ is
  stabilised by the $\Spin(2)$ subgroup of $\SU(2)$ which fixes the
  direction singled out by $\p,\bk$.  Thus we induce a UIR of $G$ from
  a UIR of $\Spin(2)\ltimes N$ which is carried by sections of the
  corresponding homogeneous vector bundle over $\SU(2)/\Spin(2)$, which
  is the $2$-sphere.  Since $N$ acts via $\chi$ and $\Spin(2)$
  stabilises $\chi$, the inducing representation is a representation
  of $\Spin(2) \times N$, hence a tensor product of a UIR of $\Spin(2)$,
  which is $\CC_n$ for some $n \in \ZZ$, and the UIR of $N$ defined by
  $\chi$.  The fibration $\SU(2) \to S^2$ is the Hopf fibration
  discussed in Appendix~\ref{sec:hopf-charts-su2}.  If $\p=\bzero$
  these are the Carroll UIRs of class $\Romanbar{III}'(n,k)$, whereas
  if $\p\neq \bzero$ they are the Carroll UIRs of class
  $\Romanbar{III}(n,p)$ if $k=0$ or of class
  $\Romanbar{IV}_\pm(n,p,k)$ where $\bk = \pm \tfrac{k}{p}\p$.

\item Finally if $\p$ and $\bk$ are not collinear, then they span a
  plane and $H_\chi = \ZZ_2$.  The UIRs are square-integrable sections
  of a homogeneous vector bundle over $\SU(2)/\ZZ_2 \cong \SO(3)$
  associated to either the trivial or the sign representation of
  $\ZZ_2$, or we could just remain on $\SU(2)$ as described earlier in
  this section and project to functions which are either odd or even
  under the action of the centre.  These are of course the Carroll
  UIRs of class $\Romanbar{V}_\pm(p,k,\theta)$.  The fact that
  $H_\chi$ is not connected is responsible for the sign, which labels
  two inequivalent quantisations of the same coadjoint orbit.
\end{itemize}

Since, as explained in Appendix~\ref{sec:meth-induc-repr}, sections of
homogeneous vector bundles over $G/H$ lift to $H$-equivariant
functions on $G$ (the so-called Mackey functions), all the massless
UIRs of the Carroll group can indeed be described in terms of functions on
$\SU(2)$, except that the functions are equivariant under the action
of $\SU(2)$, $\Spin(2)$ or $\ZZ_2$ with values in UIRs of $\SU(2)$,
$\Spin(2)$ or $\ZZ_2$, respectively.  That is precisely the
description of the UIRs detailed above.

\subsubsection{Relation with coadjoint orbits}
\label{sec:relat-with-coadj}

Following Rawnsley \cite{MR387499}, these induced representations are
obtained by geometric quantisation of the coadjoint orbits of the
Carroll group, whose classification was recalled in
Section~\ref{sec:review-coadjoint-orbits}.  This suggests a sort of
correspondence between coadjoint orbits of the Carroll group and the
induced representations we have determined, which is described in
Table~\ref{tab:uirreps-orbits}.

Any such correspondence between coadjoint orbits and UIRs needs to be
qualified.  Firstly, not every orbit is quantisable: indeed it is
clear from the table that not all orbits in classes $\#2,4,5,6,7_\pm$
are quantisable, since the angular momentum is quantised.  On the
other hand, the orbits of type $8$ admit two different quantisations:
corresponding to the two different spin structures on the circle.
This is due to the fact that the stabiliser of the character is
disconnected.  Both of these phenomena are well established, as
explained, for example, in the Introduction to \cite{MR1486137}.

Short of actually geometrically quantising the coadjoint orbits, the
correspondence in the table must remain largely conjectural, but we
are fairly confident in its validity, with the above caveats.

\begin{table}[h]
  \centering
    \caption{Coadjoint orbits and unitary irreducible representations
      of the Carroll group}
    \label{tab:uirreps-orbits}
  \begin{adjustbox}{max width=\textwidth}
    \begin{tabular}{l*{6}{>{$}c<{$}}}
      \toprule
                               & \# & \alpha \in \g^* & \eO_\tau & K_\tau & \text{inducing representation of $K_\tau$} & \text{UIR of $G$}\\\midrule \rowcolor{blue!7}
 $\Romanbar{II}(s=0,E_0)$      &   1& (\bzero,\bzero,\bzero,E_0\neq0) & \AA^3_{E=E_0} & \Spin(3) & \CC & L^2(\AA^3)\\
 $\Romanbar{II}(s\neq 0,E_0)$  &   2& (\bj\neq\bzero,\bzero,\bzero,E_0\neq0) & \AA^3_{E=E_0} & \Spin(3) & V_{s\neq 0} & L^2(\AA^3,V_s)\\ \rowcolor{blue!7}
 $\Romanbar{I}(s=0)$           &   3& (\bzero,\bzero,\bzero,0) & \{(\bzero,0)\} & K & \CC & \CC\\
 $\Romanbar{I}(s\neq 0)$       &   4& (\bj\neq\bzero,\bzero,\bzero,0) & \{(\bzero,0)\} & K & V_{s\neq 0} & V_s \\ \rowcolor{blue!7}
 $\Romanbar{III}'(n,k)$        &   5& (\bj,\bk\neq\bzero,\bzero,0)_{\bj\times\bk=\bzero} & \{(\bzero,0)\} & K & L^2(S^2,\mathscr{O}(-n)) & L^2(S^2,\mathscr{O}(-n))\\
 $\Romanbar{III}(n,p)$         &   6& (\bj,\bzero,\p\neq\bzero,0)_{\bj\times \p = \bzero} & S^2_{p} & \Spin(2) \ltimes \RR^3 & \CC_n & L^2(S^2,\mathscr{O}(-n))\\ \rowcolor{blue!7}
 $\Romanbar{IV}_\pm(n,p,k)$        &   7& (\bj,\bk\neq\bzero,\p\neq\bzero,0)_{\bk \times \p = \bj \times \bk = \bzero} & S^2_{p} & \Spin(2) \ltimes \RR^3 & \CC_n \otimes \CC_{w\neq0} & L^2(S^2,\mathscr{O}(-n))\\
 $\Romanbar{V}_\pm(p,k,\theta)$ &  8& (\bzero,\bk,\p,0)_{\bk\times\p \neq \bzero} & S^2_p & \Spin(2) \ltimes \RR^3 & \mathscr{H}^\pm_{k,\theta}:=L^2(S^1_{k\sin\theta}, \Sigma_\pm \otimes L_{k\cos\theta}) & L^2_\pm(S^3,\CC_{p,k,\theta})\\
      \bottomrule
    \end{tabular}
  \end{adjustbox}
  \vspace{1em}
  \caption*{The table lists a representative $\alpha$ of each class of
    coadjoint orbit $\eO_\alpha$, the base $\eO_\tau$ of the fibration
    which describes $\eO_\alpha$ and the little groups
    $K_\tau \subset K$ from which we induce the unitary irreducible
    representations of the Carroll group.  In each row we also list
    the vector space of the inducing representation of $K_\tau$ as
    well as that of the induced representation. The notation
    $L^2(X,V)$ means either $L^2$ functions $X\to V$, when $V$ is a
    vector space, or $L^2$ sections of a vector bundle $V$ over $X$.
    The bundles $\Sigma_\pm$ are the Ramond ($+$) and Neveu--Schwarz
    ($-$) spinor bundles over the circle and $L_{k\cos\theta}$ is a
    trivial line bundle over the circle associated to a unitary
    character of $\RR$, thought of as a one-dimensional abelian group.
    As explained in the bulk of the paper, the resulting description
    of this UIR in terms of square-integrable sections of a
    homogeneous Hilbert bundle over the $2$-sphere can be simplified
    to square-integrable functions on the $3$-sphere with values in a
    one-dimensional complex unitary representation of the subgroup of
    $G$ generated by boosts and translations labelled by
    $(p,k,\theta)$ which are either odd or even under the action of
    the centre $\ZZ_2$.}
\end{table}

\section{Carrollian fields}
\label{sec:carrollian-fields}

In this section we describe some of the UIRs of the Carroll group in
terms of fields in Carroll spacetime.  We will do one example with $E
\neq 0$ and one with $E=0$, which roughly correspond to electric and
magnetic Carroll fields, respectively.  Before doing so, we will
briefly recap the method, which we explain in more detail in
Appendix~\ref{sec:meth-induc-repr}.

Carroll spacetime is the homogeneous space $G/K$ and fields on Carroll
spacetime which transform in some representation of the Carroll group
are sections of homogeneous vector bundles associated to
representations of $K$.  In the previous section we exhibited the
UIRs of the Carroll group as sections of homogeneous vector bundles
over the orbits $\eO_\tau$ in momentum space associated to a UIR $W$
of $K_\tau$.  To view such representations as fields in Carroll
spacetime requires us firstly to make a choice of representation $V$
of $K$ which, when restricted to $K_\tau$, contains a
subrepresentation isomorphic to $W$.  This first step is known as
``covariantisation'' in the Physics literature.  We then apply a
group-theoretical Fourier transform to go from sections of homogeneous
vector bundles over $G/(K_\tau \ltimes T)$ to sections of homogeneous
vector bundles over $G/K$.

We may contrast this with the case of massive representations of the
Poincaré group.  In both cases, the little group $K_\tau \cong \SU(2)$
and $W$ is one of the complex spin-$s$ representations.
Covariantisation differs: in the Poincaré case we need to embed $W$ in
a finite-dimensional (and hence non-unitary) representation of Lorentz
group, whereas in the Carroll case we need to embed $W$ in a
finite-dimensional representation of the homogeneous Carroll group,
which is isomorphic to the three-dimensional euclidean group.  The
main difference, which will have very visible repercussions is that in
the Carroll case, the representation $W$ can be made itself covariant
by declaring the boosts, which form an ideal, to act trivially;
whereas in the Lorentz case this is not possible for any positive
spin.

\subsection{Brief recap of the method of induced representations (continued)}
\label{sec:brief-recap-method-contd}

We retake the algorithm started in
Section~\ref{sec:brief-recap-method} with the procedure of
covariantising the induced representations in terms of fields on
Carroll spacetime $M = G/K$.

\begin{enumerate}[label=(\arabic*),start=5]
\item Pick a representation $V$ of $K$ whose restriction to $K_\tau$
  contains a subrepresentation isomorphic to $W$.  This representation
  need not be unitary.  We typically take it to be irreducible, so as
  to minimise the number of extra degrees of freedom we are
  introducing.  We extend $V$ to a representation of $G$ by declaring
  $t \in T$ to act via the character $\chi_\tau$, just in the same way
  we extended $W$ to a representation of $H = K_\tau \ltimes T$.

\item Let $\zeta: M \to T$ be the coset representative for Carroll
  spacetime sending $x = (t,\x) \mapsto \zeta(x) = \exp(t H + \x
  \cdot \bP)$ and define $\phi: M \to V$ by
  \begin{equation}\label{eq:fourier-trans}
    \phi(x) = \int_{\eO_\tau} d\mu(p)
    \chi_{\sigma(p)\cdot\tau}(\zeta(x)^{-1}) \sigma(p) \cdot \psi(p),
  \end{equation}
  where $\psi : \eO_\tau \to V$.
\item Let $g \in G$ and define $k(g^{-1},x) \in K$ by
  \begin{equation}\label{eq:g-on-zeta}
    g^{-1}\zeta(x) = \zeta(g^{-1}\cdot x) k(g^{-1},x).
  \end{equation}
  Then the action of $g \in G$ on the carrollian field $\phi$ is given
  by
  \begin{equation}\label{eq:g-action-on-phi}
    (g \cdot \phi)(x) = k(g^{-1},x) \cdot \phi(g^{-1}\cdot x).
  \end{equation}

\item This representation of $G$ need neither be unitary nor
  irreducible. To remedy this we must impose field equations which say
  that $\phi$ lies in the kernel of a (pseudo-)differential operator
  which is obtained via the Fourier transform \eqref{eq:fourier-trans}
  from the projector from $V$ to $W$.
\end{enumerate}

Again, these steps are somewhat over-simplified and the more detailed
story is contained in Appendix~\ref{sec:meth-induc-repr}.  In
particular, the ``functions'' $\phi: M \to
V$ are actually sections of a homogeneous vector bundle $E_V = G
\times_K V$ over $M$. These sections are obtained via a
group-theoretical version of the Fourier transform at the level of
the Mackey functions.  Indeed, equation~\eqref{eq:fourier-trans},
which is equation~\eqref{eq:fourier-transform}, is obtained by
evaluating the Fourier-transformed Mackey function $\widehat F$
defined by equation~\eqref{eq:pre-Fourier-transform} at the coset
representative.  Similarly, the expression \eqref{eq:g-action-on-phi},
which agrees with equation \eqref{eq:G-action-fields-final} is
obtained from the more natural transformation
law~\eqref{eq:G-action-fields} of the Fourier-transformed Mackey
function $\widehat F$.

\subsection{Representations with $E\neq 0$ in terms of carrollian fields}
\label{sec:an-explicit-example}

Here we illustrate the method by working out an explicit example of
unitary irreducible representation of (the universal cover of) the
Carroll group with $E_0 \neq 0$ as fields on Carroll spacetime.  The
group element $g(R,\bv,\ba,s)$ is given by equation
\eqref{eq:standard-param}, where we assume that $G$ stands for the
universal cover of the Carroll group and $K$ for the universal cover
of the euclidean group.  In particular, $R$ in
equation~\eqref{eq:standard-param} belongs to $\Spin(3) \cong
\SU(2)$.

As discussed above, this representation is carried by the Hilbert
space of functions $\AA^3 \to V_s$ where $\AA^3 \subset \t^*$ is the
affine plane with $E = E_0$ and $V_s$ is the complex spin-$s$
representation of $\Spin(3)$, which are square-integrable relative to
the inner product defined by integrating the pointwise invariant
hermitian inner product on $V_s$ over $\AA^3$ relative to the
euclidean measure, as in equation~\eqref{eq:massive-inner-product}.
The action of the Carroll group $G$ on this representation was given
above in equation \eqref{eq:g-rep-massive}.

We are ultimately interested in fields defined on Carroll spacetime,
so we may now covariantise this representation by first of all
embedding the representation $V_s$ of $K_\tau$ into a representation
of $K$. One particularly economical choice is to declare that the
boosts act trivially. As described in the introduction this might
seem unconventional from the point of view of Lorentz invariant
theories, since this is impossible for the Lorentz group unless
$s=0$.  Indeed, the commutator of two boosts is a rotation and hence
if the boosts were to act trivially, so would the rotations.  As
explained in Appendix~\ref{sec:meth-induc-repr} and briefly recapped
above, we then need to Fourier transform to arrive at sections of a
homogeneous vector bundle over Carroll spacetime.

Let $\zeta : G/K \to T$ be a coset representative
sending\footnote{Now $t$ is the time coordinate and no longer a
  generic element of the translation subgroup $T$.} $x =
(t,x^a) \mapsto \exp(t H + x^a P_a) \in T$.  Recall that for $\p \in
\AA^3$, $\sigma(\p) = \exp(-\tfrac1{E_0} \p \cdot \bB)$ so that
\begin{equation}
  \sigma(\p)^{-1}\zeta(t, \x) \sigma(\p)= \zeta(t + E_0^{-1} \x \cdot \p, \x),
\end{equation}
so that
\begin{equation}
  \chi_\tau(\sigma(\p)^{-1}\zeta(t, \x)^{-1} \sigma(\p)) = e^{-i E_0 t
    - i \p \cdot \x}
\end{equation}
and hence $\phi(t,\x)$, which we remind the reader takes values in the
spin-$s$ representation of $\SU(2)$, is given by
\begin{equation}
  \label{eq:proper-fourier}
\phi(t,\x) = e^{-i E_0 t}\int_{\AA^3} d^3p e^{-i \p \cdot \x}\psi(\p) \, ,
\end{equation}
which is, up to the factor $e^{-i E_0 t}$ and normalisation, the
spatial Fourier transform of $\psi(\p)$.  Observe that these fields obey
\begin{equation}
  \label{eq:wave_eq_massive}
  \frac{\d\phi}{\d t} = - i E_0 \phi,
\end{equation}
suggesting that they are ``electric'' Carroll fields.

The action of the Carroll group on such fields is easily worked out.
Let $g = g(R,\bv,\ba,s) = g(0,\bzero,\ba,s) \exp(\bv \cdot \bB) R$, with
$R \in \Spin(3)$.  Then
\begin{equation}
  g^{-1}\zeta(x)= R^{-1} e^{-\bv\cdot \bB} t^{-1} \zeta(x) = R^{-1} e^{-\bv\cdot \bB} \zeta(x')
\end{equation}
where if $x = (t,\x)$ then
\begin{equation}
  x' = (t-s, \x - \ba).
\end{equation}
Continuing with the calculation,
\begin{equation}
  R^{-1} e^{-\bv\cdot \bB} \zeta(x') = \zeta(x'') R^{-1} e^{-\bv\cdot \bB},
\end{equation}
where
\begin{equation}
  x'' = (t - s - \bv \cdot (\x - \ba), R^{-1}(\x-\ba)),
\end{equation}
resulting in
\begin{equation}
  \label{eq:carroll-fields-trafos}
  (g\cdot \phi)(t, \x) = \rho(R) \phi(t-s-\bv\cdot (\x - \ba),R^{-1}(\x - \ba)),
\end{equation}
with $\rho$ the spin-$s$ representation of $\SU(2)$.

Notice that under the different kinds of Carroll transformations, the
field $\phi$ behaves as follows:
\begin{itemize}
\item under rotations with $R \in \Spin(3)$,
  \begin{equation}
    \label{eq:rotations}
    (R \cdot \phi)(t, \x) = \rho(R)\phi(t, R^{-1} \x);
  \end{equation}
\item under boosts,
  \begin{equation}
    \label{eq:boosts}
    (e^{\bv\cdot \bB}\cdot \phi)(t, \x) = \phi(t - \bv \cdot \x, \x);
  \end{equation}
\item and under translations,
  \begin{equation}
    \label{eq:translations}
    (e^{s H + \ba\cdot \bP}\cdot \phi)(t, \x) = \phi(t - s, \x - \ba).
  \end{equation}
\end{itemize}
It is interesting to note that the boosts have no action on the field
$\phi$, but only on coordinates, which is quite distinct to Poincaré
where Lorentz boosts also transform the fields.  This is easy to
explain: since the commutator of two Lorentz boosts is a rotation, we
cannot covariantise the inducing representation in the Poincaré case
by simply declaring the boosts to act trivially.

We may calculate the value of the casimirs $H$ and $W^2$ on such
fields.  Let $g(\lambda) = (R(\lambda), \bv(\lambda), \ba(\lambda),
s(\lambda))$ denote a curve through the identity on $G$; that is,
$R(0)=\mathbb{1}$, $\bv(0)=\ba(0) = \bzero$ and $s(0)=0$.  Let us now
act with $g(\lambda)$ on a field $\phi$ according to
equation~\eqref{eq:carroll-fields-trafos} to obtain a curve in the
space of fields.  Its derivative
$\left.\frac{d}{d\lambda}\right|_{\lambda = 0}$ gives
\begin{equation}
  \rho(R'(0)) \phi + \frac{\d\phi}{\d t} (-s'(0) - \bv'(0)\cdot \x) -
  (R'(0)\x) \cdot \boldsymbol{\d}\phi - \ba'(0)\cdot \boldsymbol{\d}\phi,
\end{equation}
where we also use the notation $\rho$ for the representation of
$\so(3)$ on $V_s$.  We can then read off the action of the generators
of $\g$ on $\phi$:
\begin{equation}\label{eq:g-as-vector-fields}
  \begin{split}
    H \phi &= - \frac{\d \phi}{\d t}\\
    P_i \phi &= - \frac{\d \phi}{\d x^i}\\
    B_i \phi &= - x_i \frac{\d \phi}{\d t}\\
    J_i \phi &= \rho(J_i)\phi + \epsilon_{ijk} x_k \frac{\d \phi}{\d x^j}.
  \end{split}
\end{equation}
From this it follows that $H\phi = i E_0 \phi$ and that $W_i := H J_i + \epsilon_{ijk} P_j B_k$ acts simply as
\begin{equation}
  \label{eq:g-action-on-fields}
  \begin{split}
    W_i \phi &= \underbrace{\rho(J_i) \frac{\d\phi}{\d t} + \epsilon_{ijk}x_k \frac{\d^2\phi}{\d t \d x^j}}_{HJ_i \phi} + \underbrace{\left( - \epsilon_{ijk} x_k \frac{\d^2\phi}{\d t \d x^j} \right) }_{\epsilon_{ijk}P_j B_k \phi}\\
    &= \rho(J_i) \frac{\d\phi}{\d t},
\end{split}
\end{equation}
and hence
\begin{equation}
  W^2 \phi = \rho(J^2) \frac{\d^2\phi}{\d t^2} = E_0^2 s(s+1) \phi,
\end{equation}
where we have used that $\rho(J^2)$ acts by scalar multiplication by $-s(s+1) $ on $V_s$.

\subsection{A class of $E=0$ representations in terms of carrollian fields}
\label{sec:another-example}

Let us consider the other class of representations of the Carroll
group which can be carried by (finite-component) fields on Carroll
spacetime: namely, the ones carried by the square-integrable sections
of the line bundle $\mathscr{O}(-n)$ over $\CP^1$, equivalently the
vector bundle over the sphere associated to the one-dimensional
representation $\CC_n$ of $\U(1) \subset \SU(2)$, with boosts acting
trivially.

Let us first see that relaxing this condition by inducing from $\CC_n
\otimes \CC_w$ with $w \neq 0$ results in infinite-component fields.
Indeed, as discussed in Appendix~\ref{sec:meth-induc-repr} and
recapped briefly in Section~\ref{sec:brief-recap-method-contd}, the
first step is to embed the representation $\CC_n \otimes \CC_w$ of
$(\U(1)\ltimes \RR^2)\times \RR$ into a representation of $K \cong
\SU(2) \ltimes \RR^3$.  It is not difficult to see, however, that any
such representation would have to be infinite-dimensional unless $w =
0$.  Indeed, suppose that there is a vector $\psi$ in that
representation with $e^{\bv \cdot \bB} \psi = e^{i\chi(\bv)} \psi$, with
$\chi$ the character $\chi(v_1,v_2,v_3) = w v_3$.  Let $R \in
\Spin(3)$ and consider the vector $R \psi$.  Then
\begin{equation}
  e^{\bv \cdot \bB} R \psi = R R^{-1} e^{\bv \cdot \bB} R \psi = R
  e^{R^{-1}\bv \cdot \bB} \psi = e^{i \chi(R^{-1}\bv)}
  R\psi = e^{i (R \chi)(\bv)} R\psi,
\end{equation}
so that the representations with characters $R \chi$ in the
$\SO(3)$-orbit of $\chi$ also appear.  These are labelled by the
sphere of radius $|w|$ in the dual of $\RR^3$, so that they are
uncountable unless $w = 0$.

Therefore let us take $w=0$ from now on and let us first of all embed
the representation $\CC_n$ of $\U(1)$ inside an irreducible
finite-dimensional representation $V$ of $\SU(2)$.  Every such $V$ is
isomorphic to $V_s$ for some spin $s$ and $\CC_n$ embeds in $V_s$
provided that $|n|\leq 2s$.  The smallest irreducible 
representation containing $\CC_n$ is therefore $V_s$ with $s = \left|
  \frac{n}2 \right|$, where it appears as either the subspace of
highest (if $n>0$) or lowest (if $n<0$) weight vectors in $V_s$.  Let
$V = V_{\left|\frac{n}2\right|}$ from now on; although of course other
representations are possible.\footnote{This ought to be familiar from
  relativistic fields, where, for instance, the representation
  describing a massive scalar field can be embedded as either a
  Lorentz scalar subject to the Klein--Gordon equation or as a massive
  vector field subject to the equation $\d_\mu \d \cdot V = m^2
  V_\mu$, which is the complementary equation to the Proca equation on
  the massive vector field.  In this case the only degree of freedom
  is precisely the divergence $\d \cdot V$ which is a scalar field
  obeying the Klein--Gordon equation.}

Now define $\zeta : G/K \to T$ by $\zeta(x) =
\exp(t H + \x \cdot \bP)$ as a coset representative for Carroll
spacetime $M$ and we obtain a field $\phi$ on $M$ taking values in $V$
given by
\begin{equation}\label{fourier-transform-massless}
  \phi(t,\x) =\int_\CC \frac{2i dz \wedge d\zbar}{(1+|z|^2)^2} e^{-i
    \x \cdot \q(z)} \rho(\sigma(z)) \psi(z),
\end{equation}
with $\rho$ the spin-$s$ representation of $\SU(2)$.  Notice that such
fields do not depend on the time coordinates in Carroll spacetime:
$\frac{\d \phi}{\d t} = 0$, so we may simply write $\phi(\x)$.

The action of the Carroll group on $\phi$ is easy to work out.
Let $g = \exp(s H + \ba \cdot
P) e^{\bv\cdot \bB} R \in G$, where $R \in
\SU(2)$.  For this choice of $g$,
\begin{equation}
  g^{-1} \zeta(x) = R^{-1} e^{-\bv\cdot \bB} \zeta(x')
  \quad\text{where}\quad x'=(t - s, \x - \ba ).
\end{equation}
Continuing with the calculation,
\begin{equation}
  R^{-1} e^{-\bv\cdot \bB} \zeta(x') = \zeta(x'') R^{-1} e^{-\bv \cdot \bB}
  \quad\text{where}\quad x''=\left(t - s - \bv \cdot (\x-\ba),
  R^{-1}(\x - \ba)\right).
\end{equation}
Since $\phi$ does not depend on time ($t$), we arrive at
\begin{equation}
  \label{eq:carroll-trafos}
  (g \cdot \phi)(\x) = \rho(R) \phi(R^{-1}(\x - \ba)),
\end{equation}
showing that both the boosts and the time translations act trivially.
Both casimirs $H$ and $W^2$ act trivially on such fields.

Since the fields are defined on the mass shell where $\|\p\|^2 = p^2$
is a constant, the field $\phi$ satisfies the Helmholtz equation
\begin{equation}
  \label{eq:helmholtz}
  (\bigtriangleup +p^2) \phi(\x) = 0,
\end{equation}
with $\bigtriangleup$ the laplacian in $\RR^3$.  This is easily
derived by simply inserting $0$ in the form $\|\q(z)\|^2 - p^2$ in the
integral~\eqref{fourier-transform-massless} and then noticing that
$\|\q(z)\|^2$ is (minus) the laplacian acting on $e^{-i\x \cdot
  \q(z)}$.

The irreducible representation is carried by those fields $\phi$ for
which the integral in equation~\eqref{fourier-transform-massless}
exists and for which $\psi$ is either a lowest or highest weight
vector in $V$.  For definiteness, let us assume that $\psi$ is a
highest weight vector so that it lies in the kernel of $\rho(J_+)$,
with $J_+ = \begin{pmatrix} 0 & 1 \\ 0 & 0\end{pmatrix}$.
As explained at the end of Appendix~\ref{sec:meth-induc-repr}, we may
derive an equation for such $\phi$ as follows.  We notice that since
$\rho(J_+)\psi(z)=0$, we have that
\begin{equation}
  \int_\CC \frac{2i dz \wedge d\zbar}{(1+|z|^2)^2} e^{-i
    \x \cdot \q(z)} \rho(\sigma(z)) \rho(J_+) \psi(z) = 0.
\end{equation}
Let us define $J_+(z) := \sigma(z) J_+ \sigma(z)^{-1}$, so that the
above equation becomes
\begin{equation}
  \int_\CC \frac{2i dz \wedge d\zbar}{(1+|z|^2)^2} e^{-i
    \x \cdot \q(z)} \rho(J_+(z)) \rho(\sigma(z)) \psi(z) = 0.
\end{equation}
We work out $J_+(z)$ explicitly to be
\begin{equation}
  \begin{split}
      J_+(z) = \sigma(z) J_+ \sigma(z)^{-1} &= \frac1{1+|z|^2}
  \begin{pmatrix}
    z & -1 \\ 1 & \zbar
  \end{pmatrix}
  \begin{pmatrix}
    0 & 1 \\ 0 & 0
  \end{pmatrix}
  \begin{pmatrix}
    \zbar & 1 \\ -1 & z
  \end{pmatrix}\\
  &= \frac1{1+|z|^2}
  \begin{pmatrix}
    -z & z^2 \\ -1 & z
  \end{pmatrix}.
  \end{split}
\end{equation}

To recognise the (pseudo-)differential operator defined by
$\rho(J_+(z))$, let us now calculate the spatial derivatives of
$\phi(\x)$ by differentiating under the integral sign:
\begin{equation}
   \frac{\d}{\d x^j} \phi(\x) =  \int_\CC \frac{2dz \wedge
      d\zbar}{(1+|z|^2)^2} \pi_j(z) e^{-i
    \x \cdot \q(z)} \rho(\sigma(z)) \psi(z),
\end{equation}
so that, using the components of $\q(z)$ in equation~\eqref{eq:q-vector},
\begin{equation}
  \begin{split}
    \frac{\d}{\d x^1} \phi(\x) &=  \int_\CC \frac{4p dz \wedge d\zbar}{(1+|z|^2)^3} \Re(z) e^{-i \x \cdot \q(z)} \rho(\sigma(z)) \psi(z)\\
    \frac{\d}{\d x^2} \phi(\x) &=  \int_\CC \frac{4p dz \wedge  d\zbar}{(1+|z|^2)^3} \Im(z) e^{-i \x \cdot \q(z)} \rho(\sigma(z)) \psi(z)\\
    \frac{\d}{\d x^3} \phi(\x) &=  \int_\CC \frac{2p dz \wedge  d\zbar}{(1+|z|^2)^3} (|z|^2-1) e^{-i \x \cdot \q(z)} \rho(\sigma(z)) \psi(z).
  \end{split}
\end{equation}
This sets up the following dictionary:
\begin{equation}
  \label{eq:dictionary}
  \begin{split}
    \frac{z}{1+|z|^2} &\squigto \frac{i}{2p} (\d_1 + i \d_2)\\
    \frac{1}{1+|z|^2} &\squigto -\frac{i}{2p} (\d_3 + i p)\\
    \frac{z^2}{1+|z|^2} &\squigto -\frac{i}{2p} \frac{(\d_1 + i \d_2)}{\d_3 + i p}.
  \end{split}
\end{equation}
It is then a relatively simple matter to take the $\rho(J_+(z))$ in the
integrand and write it as a (pseudo-)differential operator acting on
the integral, obtaining the somewhat formal expression
\begin{equation}
  \label{eq:non-local}
  \rho\left(
    \begin{pmatrix}
      \d_1 + i \d_2 & \frac{(\d_1 + i \d_2)^2}{\d_3 + i p}\\
      -(\d_3 + i p) & - (\d_1 + i \d_2)
    \end{pmatrix}
\right) \phi(\x)= 0,
\end{equation}
where we have omitted an inconsequential overall factor of $1/(2ip)$
in the left-hand side.

A similar calculation starting with $J_- = \begin{pmatrix} 0 & 0 \\ 1
  & 0 \end{pmatrix}$, whose kernel consists of the lowest weight
vectors, results in
\begin{equation}
  \label{eq:non-local-lwv}
  \rho\left(
    \begin{pmatrix}
      -(\d_1 - i \d_2) & -(\d_3 + i p) \\
      \frac{(\d_1 - i \d_2)^2}{\d_3 + i p} & \d_1 - i \d_2
    \end{pmatrix}
\right) \phi(\x)= 0,
\end{equation}
which is the relevant equation in the case of negative helicity.

These equations look non-local due to the presence of the resolvent
$(\d_3 + i p)^{-1}$, but we can try to make sense of them.  Let us
discuss some explicit cases.

\subsubsection{Helicity $0$}
\label{sec:helicity-0}

This case needs no discussion, since there $\rho$ is the trivial
representation and the only condition on the scalar field is the
Helmholtz equation~\eqref{eq:helmholtz}.

\subsubsection{Helicity $\tfrac12$}
\label{sec:helicity-tfrac12}

In this case, $\rho$ is the defining representation of $\SU(2)$, so
the identity map.  Here $\phi(\x)$ is a $2$-component field and once
we take into account the Helmholtz equation~\eqref{eq:helmholtz}, the
solutions to equation~\eqref{eq:non-local} are also the solutions to
the Dirac-like equation
\begin{equation}
  \label{eq:dirac-like}
  (\slashed{\d} + i p) \phi = 0,
\end{equation}
where $\slashed{\d} = \gamma^i \d_i$ with $\gamma^i$ the representation
of $\Cl(0,3)$ given by
\begin{equation}
  \gamma^1=
  \begin{pmatrix}
    0 & -1 \\ -1 & 0
  \end{pmatrix},\qquad
  \gamma^2=
  \begin{pmatrix}
    0 & -i \\ i & 0
  \end{pmatrix} \qquad\text{and}\qquad
  \gamma^3=
  \begin{pmatrix}
    1 & 0 \\ 0 & -1
  \end{pmatrix}.  
\end{equation}

\subsubsection{Helicity $1$}
\label{sec:helicity-1}

The representation $\rho$ now is the adjoint representation of
$\SU(2)$.  Let $R\in \SU(2)$ be given by
equation~\eqref{eq:R-and-inverse}.  Then the adjoint representation of
$R$ is given by
\begin{equation}
  \label{eq:adjoint-rep}
  \rho(R) =
  \begin{pmatrix}
    \Re(\eta^2 - \xi^2) & - \Im(\eta^2 + \xi^2) & -2\Re (\eta\xi)\\
    \Im(\eta^2 - \xi^2) & \Re (\eta^2 + \xi^2) & -2 \Im (\eta \xi)\\
    2 \Re(\eta\overline\xi) & -2 \Im (\eta \overline\xi) & |\eta|^2 - |\xi|^2
  \end{pmatrix},
\end{equation}
which can be checked to be a matrix in $\SO(3)$.  The matrix $J_+(z)$
belongs to $\fsl(2,\CC)$, the complexification of $\su(2)$.  The way
we calculate $\rho(J_+(z))$ is as follows.  We consider $X \in \su(2)$
and exponentiate $\exp(t X) \in \SU(2)$ and differentiate $\rho(\exp(t
X))$ at $t=0$.  This gives us a map (also denoted) $\rho: \su(2) \to
\so(3)$, which we extend complex-linearly to $\rho: \fsl(2,\CC) \to
\so(3,\CC)$, which can then be evaluated at $J_+(z)$.

Let $X \in \su(2)$ be given by
\begin{equation}
  X =
  \begin{pmatrix}
    i \gamma & \alpha + i \beta\\
    -\alpha + i \beta & -i \gamma
  \end{pmatrix}
\end{equation}
and let $\|X\| := \sqrt{\alpha^2 + \beta^2 + \gamma^2}$.  Then $\exp(t
X) \in SU(2)$ is given by a matrix of the form in
equation~\eqref{eq:R-and-inverse} with
\begin{equation}
  \begin{split}
    \eta(t) &= \cos\left( t \|X\| \right) + i \gamma \frac{\sin\left( t \|X\| \right)}{\|X\|}\\
    \xi(t) &= (\alpha + i \beta) \frac{\sin\left( t \|X\| \right)}{\|X\|}.
  \end{split}
\end{equation}
We insert this into equation~\eqref{eq:adjoint-rep} and differentiate
with respect to $t$ and evaluate at $t=0$ to obtain
\begin{equation}
  \rho(X) =
  \begin{pmatrix}
    0 & - 2 \gamma & -2\alpha\\
    2\gamma & 0 & -2\beta\\
    2\alpha & 2 \beta & 0
  \end{pmatrix},
\end{equation}
which we extend to $\rho : \fsl(2,\CC) \to \so(3,\CC)$ simply by
allowing $\alpha,\beta,\gamma \in \CC$.  For $X = J_+(z)$, we have
\begin{equation}
  2\alpha = \frac{1 + z^2}{1 + |z|^2}, \qquad 2 \beta = \frac{i
    (1-z^2)}{1+|z|^2} \qquad\text{and}\qquad \gamma = \frac{i z}{1+|z|^2},
\end{equation}
so that
\begin{equation}
  \rho(J_+(z)) = \frac1{1+|z|^2}
  \begin{pmatrix}
    0 & -2 i z & - (1 + z^2)\\
    2 i z & 0 & i (z^2 - 1)\\
    1 + z^2 & i (1-z^2) & 0
  \end{pmatrix}.
\end{equation}
Using the dictionary in equation~\eqref{eq:dictionary}, we may write
the pseudo-differential equation for helicity-$1$ fields as
\begin{equation}
  \label{eq:helicity-one-pseudo-diff}
  \begin{pmatrix}
    0 & \d_1 + i \d_2   & \frac{i}{2} \left( \d_3 + i p + \frac{(\d_1 + i \d_2)^2}{\d_3 + i p} \right)\\
    -(\d_1 + i\d_2) & 0 & -\frac12 \left( \d_3 + i p - \frac{(\d_1 + i \d_2)^2}{\d_3 + i p} \right)\\
    -\frac{i}{2} \left( \d_3 + i p + \frac{(\d_1 + i \d_2)^2}{\d_3 + i p}
    \right) & \frac12 \left( \d_3 + i p - \frac{(\d_1 + i
        \d_2)^2}{\d_3 + i p} \right) & 0
  \end{pmatrix} \phi(\x) = 0,
\end{equation}
where $\phi(\x)$ is now a $3$-component field subject to the Helmholtz
equation~\eqref{eq:helmholtz}.  Let us write $\phi =
(\phi_1,\phi_2,\phi_3)$.  Then the solutions of the
equation~\eqref{eq:helicity-one-pseudo-diff} agree with the solutions
of the following differential equations:
\begin{equation}
  \label{eq:helicity-one-local}
  \begin{split}
    (\d_1 + i \d_2)(\phi_1 - i \phi_2) &= - (\d_3 + i p) \phi_3\\
    (\d_1 + i \d_2)\phi_3 &= (\d_3 + i p)(\phi_1 + i \phi_2).
  \end{split}
\end{equation}
Under the assumption that $\phi_i$ are real fields, these
equations imply the Helmholtz equations $\left( \d_1^2 +
  \d_2^2 + \d_3^2  \right) \phi_i = -p^2 \phi_i$ for $i=1,2,3$, and in
fact, breaking the equations up into real and imaginary parts, we find
that they are equivalent to a much simpler equation:
\begin{equation}
  p \phi_i = \epsilon_{ijk} \d_j \phi_k,
\end{equation}
or, equivalently, $d\phi = p \star \phi$, thinking of $\phi \in
\Omega^1(\RR^3)$.  We recognise this equation as the field equation of
euclidean topologically massive Maxwell theory
\cite{Deser:1981wh,Deser:1982vy}, with $\phi$ playing the rôle of the
Hodge dual of the Maxwell field-strength.

\subsubsection{Remarks}
\label{sec:remarks}

The massless UIRs with helicity are such that time translations and
boosts act trivially.  Hence they are actually UIRs of the
three-dimensional euclidean group and, presumably, any euclidean
three-dimensional field theory should serve as a starting point to
constructing massless carrollian field theories with vanishing
centre-of-mass charge.  Since the centre of mass vanishes for these
theories, the symmetries and consequently the particles and theories
are basically represented by aristotelian symmetries where the boosts
play no rôle (see, e.g.,~\cite{Figueroa-OFarrill:2018ilb}).

In massless relativistic theories, the field equations typically
contain solutions with both signs of the helicity.  This is not the
case here.  Indeed, performing a similar calculation for helicity
$-\tfrac12$ results in the ``opposite''  Dirac equation:
\begin{equation}
  \label{eq:dirac-like-negative}
  (\slashed{\d} - i p) \phi = 0.
\end{equation}
Clearly only the only field obeying this and
equation~\eqref{eq:dirac-like} simultaneously is $\phi =0$.

A final remark is that the above equations for helicity $\tfrac12$ and
helicity $1$ are first-order partial differential equations which
imply the Helmholtz equation, which is a second-order partial
differential equation.  This is typically one of the (mathematical)
signatures of supersymmetry and it would be interesting to explore
this further.

\section{Fractonic particles and fields}
\label{sec:fractonic-fields}

The correspondence between Carroll and fracton particles established
in Part~I persists upon quantisation.  As recalled in Part~I, the
dipole algebra corresponding to our notion of fracton is isomorphic to
a trivial one-dimensional central extension of the Carroll algebra.
The additional generator is the time translation generator of the
aristotelian spacetime underlying the fracton dynamics.  The time
translation generator of the Carroll spacetime is the electric charge
generator of the dipole algebra, whereas the Carroll boost generators
are the dipole generators.  There are, of course, no boosts in an
aristotelian spacetime.  A UIR of the dipole group is the tensor
product of a UIR of the Carroll group and a one-dimensional UIR of the
one-dimensional Lie group generated by the additional central
generator.  Any one-dimensional UIR is characterised by a character,
which in the absence of any criterion which would imply the
quantisation of energy of the quantum fracton, is simply given by a
real number.  On the other hand, if we do demand (as might seem
reasonable) that the electric charge be quantised, then this would
restrict the Carroll UIRs with a fractonic interpretation to those
where $E_0$ is quantised in units of an elementary fracton charge.

\subsection{Unitary irreducible representations of the dipole group}
\label{sec:unit-irred-repr-3}

Let $\widetilde G = G \times \RR$ denote the dipole group.  The
aristotelian spacetime $A$ underlying the fractonic theory is a
homogeneous space of $\widetilde G$ with stabiliser $\widetilde K = K
\times \RR$, where $\widetilde K$ is the group generated by rotations,
boosts and time translations (in the Carroll language) or by
rotations, dipole and electric charge generators (in the fracton
language).  In other words, we have a $\widetilde G$-equivariant
diffeomorphism
\begin{equation}
  A \cong \widetilde G/\widetilde K  \cong\left( G/(K\times \RR) \right) \times \RR,
\end{equation}
where the copy of $\RR$ in the ``denominator'' is the Carroll time
translation subgroup, whereas that in the ``numerator'' is the fracton
time translation subgroup.  In this section we will let $t$ denote the
aristotelian time coordinate.  We will introduce $H_F$ as the
aristotelian time translation generator and we shall relabel the
Carroll generators $H,\bB$ to $Q,\bD$, with brackets
$[D_i, P_j] = \delta_{ij} Q$.  We will also introduce the notation
$\widetilde T = T \times \RR$ and $\widetilde\t = \t \oplus \RR H_F$.
A typical element $\widetilde\tau \in \widetilde\t^*$ is denoted now
$(q,\p,E)$ corresponding to the electric charge, the spatial momentum
and the energy, respectively: in other words,
$\left<\widetilde\tau,Q\right>=q$,
$\left<\widetilde\tau,\bP\right>=\p$ and
$\left<\widetilde\tau,H_F\right>=E$.

UIRs of the dipole group have a constant (fracton) energy.  Letting
$E$ now denote the fracton energy, we have the following UIRs of
the dipole group, where $\CC_E$ denotes the one-dimensional
representation of the aristotelian time translation subgroup
associated with the character $\chi(e^{s H_F}) = e^{i sE}$:
\begin{description}
\item[$\Romanbar{I}(s,E)$] $\cong \Romanbar{I}(s) \otimes \CC_E$, with $2 s \in \NN_0$ and $E \in \RR$;
\item[$\Romanbar{II}(s,q,E)$] $\cong \Romanbar{II}(s,q)\otimes \CC_E$, with $2 s \in \NN_0$, $q \in \RR$ (or $\ZZ$ if quantised) and $E \in \RR$;
\item[$\Romanbar{III}(n,p,E)$] $\cong \Romanbar{III}(n,p)\otimes \CC_E$, with $n \in \ZZ$, $p>0$ and $E \in \RR$;
\item[$\Romanbar{III}'(n,d,E)$] $\cong \Romanbar{III}'(n,k)\otimes \CC_E$, with $n \in \ZZ$, $d>0$ and $E \in \RR$; 
\item[$\Romanbar{IV}_\pm(n,p,d,E)$] $\cong \Romanbar{IV}_\pm(n,p,k)\otimes \CC_E$, with $n \in \ZZ$, $p>0$, $d>0$ and $E \in \RR$;
\item[$\Romanbar{V}_\pm(p,d,\theta,E)$] $\cong \Romanbar{V}_\pm(p,d,\theta) \otimes \CC_E$, with $p,d>0$, $\theta \in (0,\pi)$ and $E \in \RR$.
\end{description}
These representations have the same underlying Hilbert space as the
corresponding representations of the Carroll group and the group
element $(g,s)\in G \times \RR$ acts as
\begin{equation}
  (g,s)\cdot \psi = e^{iEs} g\cdot \psi.
\end{equation}

Let $\widetilde \tau = (q,\p,E) \in \widetilde \t^*$ and let
$\chi_{\widetilde\tau}$ denote the unitary character defined by
\begin{equation}
  \chi_{\widetilde\tau}\left(e^{\theta Q + \ba\cdot \bP + s H_F}\right) =
  e^{i(\theta q + \ba\cdot\p + s E)}.
\end{equation}

The first step in deriving the expressions for the aristotelian fields
is to extend the $K_\tau$-representation $W$ first to a representation
$W \otimes \CC_E$ of $K_\tau \times \RR$, and then to a representation
$V \otimes \CC_E$ of $K \times \RR$.  As we did earlier with the
Carroll group, we may extend $V \otimes \CC_E$ to a representation of
the whole dipole group via the character $\chi_{\widetilde\tau}$
above.  These unitary representations are described as sections of
homogeneous vector bundles over $\eO_{\widetilde\tau} \cong \eO_\tau$.
Let $\zeta : A \to \widetilde T$ send the aristotelian coordinates
$(t,\x) \mapsto \exp(\x \cdot \bP + t H_F)$.

\subsection{Charged aristotelian fields}
\label{sec:charg-arist-fields}

Let us consider the case of UIRs with nonzero electric
charge, so of class $\Romanbar{II}(s,q,E)$ with $q \neq 0$.   They
correspond to $\widetilde\tau = (q,\bzero,E)$.  The coset
representative $\sigma: \eO_{\widetilde\tau} \to K$ is given as before
by $\sigma(\p) = \exp(-\tfrac1q \p \cdot \bD)$.  Then a short
calculation gives
\begin{equation}
  \zeta(t,\x) \sigma(\p)  =\sigma(\p) \zeta(t,\x) \exp\left( \tfrac1q
    \p\cdot\x Q \right),
\end{equation}
which results in the aristotelian field
\begin{equation}
  \phi(t,\x) = e^{-iE t} \int_{\AA^3} d^3p e^{-i\p\cdot\x} \psi(\p) \,,
\end{equation}
taking values in the spin-$s$ representation of $\SU(2)$.  To work out
the action of $\widetilde G$ on such fields, we first consider the
general group element
\begin{equation}
  g=g(R,\bom,\ba,\theta,s) = e^{s H_F + \theta Q + \ba \cdot \bP}
  e^{\bom \cdot \bD} R,
\end{equation}
with $R \in \SU(2)$ and calculate
\begin{equation}
  g^{-1} \zeta(t,\x) = \zeta(t-s, R^{-1}(\x - \ba))
  e^{-\theta - \bom \cdot (\x - \ba)Q} R^{-1} e^{-\bom\cdot \bD},
\end{equation}
from where we deduce (as in Appendix~\ref{sec:meth-induc-repr}) that
\begin{equation}
  (g \cdot \phi)(t,\x) = e^{i q(\theta + \bom\cdot (\x - \ba))} \rho(R)
  \phi(t-s, R^{-1}(\x - \ba)) \, ,
\end{equation}
with $\rho$ the spin-$s$ representation of $\SU(2)$. Let us emphasise
that, as expected, there is no action of the dipole transformations on
the coordinates, cf.~\eqref{eq:proper-fourier}, and they indeed
transform as expected under charge and dipole transformations
\begin{equation}
  \label{eq:charge-dip-rot}
  \left(e^{\theta Q + \bom \cdot \bD} \cdot \phi \right)(t,\x) = e^{i q(\theta + \bom\cdot \x)} \phi(t, \x ) \, .
\end{equation}

\subsection{Neutral aristotelian fields}
\label{sec:neutr-arist-fields}

These provide field theoretical realisations of UIRs of class
$\Romanbar{III}(n,p,E)$.  Here $\widetilde \tau = (0,\p, E)$ and
$\eO_{\widetilde \tau}$ is the $2$-sphere of radius $\|\p\|$.  We
think of the sphere as the extended complex plane as we did when
discussing carrollian fields, with coset representative $\sigma(z)$
given by \eqref{eq:sigma-z} and
$\zeta(t,\x) = \exp(\x \cdot \bP + t H_F)$ as above.  The expression
for the neutral field is now
\begin{equation}
  \phi(t,\x) = e^{-i E t} \int_\CC \frac{2i dz \wedge d\zbar}{(1 +
    |z|^2)^2} e^{-i \x \cdot \q(z)} \rho(\sigma(z)) \psi(z) \, ,
\end{equation}
where $\q(z)$ is given by equation~\eqref{eq:q-vector}.  To work out
the $\widetilde G$-action, we consider $g = g(R,\bom,\ba,\theta,s)$ as
above and (via equation~\eqref{eq:pre-Fourier-transform}) work out that
\begin{equation}
  (g \cdot \phi)(t,\x) = \rho(R) \phi(t-s, R^{-1}(\x - \ba)) \, ,
\end{equation}
since $q = 0$ now.  Again the irreducible representations are those
where $\phi$ is either a highest (if $n>0$) or lowest (if $n<0$)
weight vector of the spin-$|\frac{n}2|$ representation $\rho$ of
$\SU(2)$.  The discussion in Section~\ref{sec:another-example} applies
mutatis mutandis.  The field equations in that section for the
cases of $s=0,\frac12,1$ are to be supplemented by the condition
$\frac{\d\phi}{\d t} = -i E \phi$.

\section{Quantum field theory}
\label{sec:carr-fract-quant}

The unitary irreducible representations of the Carroll and dipole
groups can be used as a starting point for the development of
carrollian\footnote{We are aware of a forthcoming work which also
  discusses Carroll quantum field theories~\cite{deBoer:2023fnj}.} and
fractonic quantum field theories, via the process known as second
quantisation. Multiparticle states formed from the quantum states of
the unitary irreducible representations presented in
Section~\ref{sec:summary-1} will then span the corresponding Fock
spaces.

The following discussion centers around non-interacting quantum field
theories, with a particular focus on scalar representations
corresponding to massive carrollion/fractonic monopoles
$\Romanbar{II}(s=0)$ and massless carrollions/aristotelions
$\Romanbar{III}(n=0,p)$.

\subsection{Massive carrollions/fractonic monopole}
\label{sec:mass-carr-monop}

While most of what we are saying can be generalised to generic spin we
will in the following restrict to spin $s=0$. This means the we focus
on massive carrollions and fractonic monopoles, corresponding to UIRs of
class $\Romanbar{II}(s=0,E_0)$ and $\Romanbar{II}(s=0,q,E)$,
respectively.

The quantum field theory of an interacting real massive scalar field
with action
\begin{equation}
  I=\int dtd^{3}x\left(\frac{1}{2}\dot{\phi}^{2}-\frac{1}{2}E_{0}^{2}\phi^{2}-V(\phi)\right),
\end{equation}
was studied by Klauder in the early 70s~\cite{Klauder:1970cs,
  Klauder:1971zz} (see also the references therein and Chapter 10
in~\cite{Klauder:2000ud} for an useful overview). Due to the absence
of gradient terms $(\pd\phi)^{2}$ and the resulting independence of
the evolution of the dynamics at each point in space, he called this
scalar field theory ``ultralocal.'' In recent times this theory has
reemerged as ``electric'' Carroll scalar field
theory~\cite{Henneaux:2021yzg,deBoer:2021jej}. The space of solutions
of the free theory, where $V(\phi)=0$, coincides with that of the free
field equation for the massive carrollion
in~\eqref{eq:wave_eq_massive} when we consider simultaneously
solutions with positive and negative energies. Therefore, in the
context of second quantisation, the Fock space of this model will be
spanned by multiparticle states formed from the vectors of the unitary
irreducible representation $\Romanbar{II}(s=0,\pm E_0)$. In
particular, when considering configurations involving both positive
and negative energies, according to~\eqref{eq:proper-fourier} the wave
function takes the form
\begin{equation}
  \phi\left(t,\boldsymbol{x}\right)
  =\int_{\mathbb{A}^{3}}d^{3}p\left(e^{-i\left(E_{0}t+\p \cdot \x\right)}\psi\left(\boldsymbol{p}\right)
    +e^{i\left(E_{0}t+\p \cdot \x\right)}\psi^{\dagger}\left(\boldsymbol{p}\right)\right).\label{eq:expansion_scalar_field}
\end{equation}
The operators $\psi^{\dagger}\left(\boldsymbol{p}\right)$ and
$\psi\left(\boldsymbol{p}\right)$ are interpreted as creation and
annihilation operators. Thus, if the vacuum state is denoted by
$\ket{0}$, then one-particle states belonging to the
representation $\Romanbar{II}(s=0,E_0)$ are given by
\begin{equation}
  \label{eq:1-particle-state-massive} 
  \ket{\p} =\psi^{\dagger}\left(\p\right)\ket{0} \,
\end{equation}
or $\ket{E_{0},\p}$ if we wish to make the specific energy explicit.
 
This theory can be derived from an ultrarelativistic limit
$(c \rightarrow 0)$ of a real Klein-Gordon field. In this limit, all
the frequencies of the relativistic scalar field
$\omega=\sqrt{E_0^{2}+ c^{2}\boldsymbol{p}^{2}}$ collapse to a fixed
value $\omega=E_0$. Geometrically, this is the limit in which a
hyperboloid in momentum space tends to the plane, pictured in
Figure~\ref{fig:mom_lim} (the full limit is given in Part I, Figure
2). Therefore, every one-particle
state~\eqref{eq:1-particle-state-massive} possesses exactly the same
energy $E_{0}$, irrespective of its momentum. Consequently, there
exist an infinite degeneracy of states with identical energy, which
follows from the fact that the Hamiltonian is a central element of the
Carroll algebra.

The main points discussed in the carrollian context are also valid for
fractons, albeit with a different physical interpretation. For
example, the free part of the action that describes the complex scalar
field model of fractons introduced by Pretko~\cite{Pretko:2018jbi} is
given by
\begin{equation}
  I=\int dtd^{3}x\left[\dot{\phi}^{*}\dot{\phi}-E^{2}\phi^{*}\phi\right],\label{eq:I_scalar_fracton}
\end{equation}
This theory is described by wave functions that belong to the unitary
irreducible representation $\Romanbar{II}\left(s=0,q,E\right)$,
which corresponds to the charged aristotelian fields discussed in
Section~\ref{sec:charg-arist-fields}.

Similar to the Carroll case, there is also an infinite degeneracy of
eigenstates having the same energy. In the context of fractons, the
degeneracy also extends to the vacuum because the representations of
class $\Romanbar{II}$ allow for a vanishing value of the energy. For
example we are free to create a large number of monopoles with zero
energy and arbitrary momentum. The degeneracy in the energy makes the
study of statistical mechanics of fractons with dipole symmetries a
non-trivial task. A possibility explored earlier in the literature
(see footnote~\ref{fn:1} and~\cite{Jensen:2022iww} and references
therein) is to describe these models on a finite lattice, that acts as
a UV regulator for the momentum. However, the continuum limit of
quantities such as the partition function in the canonical ensemble,
or the entropy in the microcanonical ensemble, is not clearly
understood. This can be attributed to the fact that the Gibbs operator
$e^{-\beta \hat H}$ is not trace class.\footnote{For a single massive
  Carroll particle we compute the character for a pure imaginary time
  translations, or equivalently, the canonical partition function.
  Formally this is given by
  $\mathrm{Tr} (e^{-\beta \hat H}) \propto \int d^{3}p e^{-\beta
    E_{0}} \delta^{3}(0)$ where $E_{0}$ is independent of $\p$. We see
  two divergences, the first $\delta^{3}(0)$ is due to the infinite
  volume in space and therefore an infrared divergence. It is also
  present for massive Poincaré particles (see,
  e.g.,~\cite{Oblak:2016eij}) and one way to regulate it is by putting
  the particle in a box. The second divergence is due to $\int d^{3}p$
  the infinite volume integral in momentum space (see
  Figure~\ref{fig:mom_lim}) and it is therefore an ultraviolet
  divergence. This is a carrollian feature due to the fact that the
  energy is independent of the momentum and can potentially be
  regulated by putting the particle on a lattice. Similar remarks
  apply to the monopoles.} One potential approach to address this
problem could be to add additional operators to the partition function
in order to lift the infinite degeneracy, akin to what is usually done
in the representation theory of infinite-dimensional algebras (e.g.,
W-algebras) in order to compute characters. Therefore, the
thermodynamics for carrollian or fracton theories of this type is
subtle.

The previous discussions were mainly centred on free field theories.
However, it is natural to explore the implications of incorporating
interactions. A potential that is invariant under dipole
transformations and is of quartic order in the field $\phi$ is given
by~\cite{Pretko:2018jbi}
\begin{equation}
V=\lambda\left|\phi\d_{i}\d_{j}\phi-\d_{i}\phi\d_{j}\phi\right|^{2}+\lambda'\phi^{*2}\left(\phi\d_{i}\d^{i}\phi-\d_{i}\phi\d^{i}\phi\right) \, ,
\end{equation}
where $\lambda$ and $\lambda'$ are coupling constants. The unitary
irreducible representations of class $\Romanbar{II}$ can be used to
describe asymptotic states for this class of interacting quantum field
theories.

Note that in theories with conserved dipole charges, the notion of a
scattering process between fractonic monopoles is subtle. If one
assumes that the initial and final states are separated by large
enough distances such that the interaction between them is negligible,
then the free monopoles will not move and the scattering process will
never take place. For a scattering process to exhibit non-trivial
behaviour, it is necessary for at least some of the interaction to
influence the asymptotic monopole states. Alternatively, due to their
unrestricted mobility, the scattering between dipoles emerges as a
natural physical process. A composite dipole can arise from the
binding of two monopole states possessing charges of equal magnitude
but opposite signs. Conversely, fundamental dipoles correspond to the
unitary irreducible representations of classes $\Romanbar{III}'$,
$\Romanbar{IV}_\pm$ and $\Romanbar{V}_{\pm}$.

\subsection{Massless carrollions/aristotelions}
\label{sec:massl-carr}

Massless carrollions/aristotelions are built on UIRs of class
$\Romanbar{III}(n=0,p)$ where we again focus on vanishing helicity,
hence $n=0$.

As discussed in Section~\ref{sec:another-example}, the wave function
for scalar massless carrollions/aristotelions is, basically by
construction, time-independent:
\begin{align}
  \label{eq:tzero}
  \pd_{t} \phi = 0,
\end{align}
and satisfies the Helmholtz equation
\begin{equation}
  \label{eq:helmholtz-1}
  (\bigtriangleup +p^2) \phi(\x) = 0.
\end{equation}
If the wave function $\phi(\x)$ is real, it will be promoted to a
hermitian operator in the second quantisation. Following
\eqref{fourier-transform-massless} the general solution to the
Helmholtz equation can be written as
\begin{equation}
  \phi(\x) =\int_\CC \frac{2i dz \wedge d\zbar}{(1+|z|^2)^2} \left(e^{-i
    \x \cdot \q(z)} \psi(z)+e^{i
    \x \cdot \q(z)} \psi^{\dagger}(z)\right).
\end{equation}
The functions $\psi^{\dagger}(z)$ and $\psi(z)$ can be promoted to
creation and annihilation operators, respectively.  Therefore,
one-particle states belonging to the representation of class
$\Romanbar{III}(n=0,p)$ can be obtained by acting on the vacuum
$\ket{0}$ according to
\begin{equation}
\ket{z} =\psi^{\dagger}\left(z\right)\ket{0} \, ,
\end{equation}
where $z$ denotes a point on the $2$-sphere defined by the constraint
$\|\p\|^{2}=p^{2}$. All states associated with massless
carrollions/aristotelions have zero energy hence standard
thermodynamics and statistical mechanics for massless carrollian
quantum field theories is trivial. This observation is in line with
the time-independence of the field $\phi(\x)$ and the absence of a
notion of ergodicity.

It is worth mentioning that whilst this theory shares similarities
with the scalar field theory referred to as the ``magnetic'' Carroll scalar
field theory proposed in \cite{Henneaux:2021yzg, deBoer:2021jej},
the theories are not equivalent.  The field equation of the magnetic
Carroll scalar field is described by the Helmholtz equation
supplemented with a source term
\begin{equation}
  \label{eq:helmholtz-2}
  (\bigtriangleup +p^2) \phi(\x) = \dot{\pi}\left(t,\boldsymbol{x}\right),
\end{equation}
where $\pi\left(t,\boldsymbol{x}\right)$ is the canonical momentum
conjugate to $\phi(\x)$. Thus, does not describe the free theory
corresponding to a UIR of the Carroll group. To characterise a UIR of
the Carroll group, additional conditions corresponding to
$\dot{\pi}=0$ need to be taken into account.
In~\cite{Figueroa-OFarrill:2023vbj} we proposed an action of the form
\begin{align}
  \label{eq:mag-other}
  I[\phi,\pi,u]
  = \int d t d^{3}x 
  \left(
  \pi\dot \phi - u (\bigtriangleup + p^{2})\phi
  \right) .
\end{align}
As the magnetic theory, characterized by eq.~\eqref{eq:helmholtz-2}, does not describe a UIR, it opens the possibility for its statistical mechanics to be non-trivial. We leave this question for future investigations.

\section{Discussion}
\label{sec:discussion}

We classified and related unitary irreducible representations of the
Carroll and dipole groups, i.e., we defined quantum Carroll and fracton
particles, which are summarised in Section~\ref{sec:summary-1} and
Table~\ref{tab:carrvsfrac}. This lifts the correspondence between
elementary Carroll and fracton
particles~\cite{Figueroa-OFarrill:2023vbj} to the quantum world.

The UIRs have distinctive features depending on whether the Carroll
energy/fracton charge vanishes or not. Isolated massive carrollions and
monopoles have quantised spin and a position operator and, when
isolated, do not move. The can be represented either as wave functions
of momenta or in centre-of-mass/dipole-moment space, which are related
via a Fourier transform.

On the other hand for vanishing energy (or fracton charge) we find
massless carrollions or neutral fractons. When the centre-of-mass or
dipole moment vanishes, massless carrollions or aristotelions have
quantised helicity and are intrinsically very similar.

We also constructed field theories in Carroll or aristotelian space
which we supplemented by free field equations projecting to the UIR.
Alternatively, the wave equations for scalar fields could be derived
by applying Dirac quantisation to the constraints that define the
classical particle models\footnote{A different class of classical
  fracton particle models with subsystem symmetries were constructed
  in~\cite{Casalbuoni:2021fel}.} described in Part~I.

Finally, we commented on the Carroll and fracton quantum field
theories for which the particles can be thought of as elementary
excitations. We highlighted the relation of massive Carroll and
charged fracton theories to ultralocal
theories~\cite{Klauder:1970cs,Klauder:1971zz} and emphasised that 
massive Carroll and (charged) fracton theories share similarities,
like puzzling thermodynamic behaviour related to their infinite
degeneracy with regard to energy and charge. We also pointed out that
the magnetic scalar theory has a different number of degrees of
freedom with respect to the free theory and presented an alternative
action~\eqref{eq:mag-other} which does.

There are various interesting points for further exploration:
\begin{description}
\item[Interacting theories.] 

  Even though most of the discussion has concentrated on free
  theories, it is worth exploring the potential application of our
  analysis to interacting theories. Some carrollian/fractonic models
  with non-trivial interactions include self-interacting fracton
  scalar field theories and their coupling to fracton
  electrodynamics~\cite{Pretko:2018jbi,Seiberg:2019vrp,Gorantla:2022eem,Jensen:2022iww},
  Carroll scalar fields and their coupling to Carroll
  electrodynamics~\cite{Klauder:1970cs,Brauner:2020rtz,
    Mehra:2023rmm,Banerjee:2023jpi}, theories of Yang--Mills
  type~\cite{Islam:2023rnc,Islam:2023iju}, Carrollian
  gravity~\cite{Isham:1975ur,
    Teitelboim:1981ua,Henneaux:1979vn,Hartong:2015xda,Bergshoeff:2017btm,Henneaux:2021yzg,Perez:2021abf,Hansen:2021fxi,
    Perez:2022jpr,Figueroa-OFarrill:2022mcy,Campoleoni:2022ebj,Sengupta:2022rbd},
  lower dimensional Chern--Simons
  theories~\cite{Bergshoeff:2016soe,Matulich:2019cdo,Huang:2023zhp} or
  extensions
  thereof~\cite{Ravera:2019ize,Concha:2021jnn,Concha:2022muu},
  Carrollian JT or dilaton
  gravity~\cite{Grumiller:2020elf,Gomis:2020wxp} or their supergravity
  versions~\cite{Ravera:2022buz}, and theories with spacetime
  subsystem symmetries~\cite{Baig:2023yaz,Kasikci:2023tvs}. An extensive list of
  references can be found in the reviews \cite{Nandkishore:2018sel,
    Pretko:2020cko, Bergshoeff:2022eog}. It would be interesting to
  explore whether these theories can be regarded as descriptions of
  interacting particles that belong to the UIRs of the Carroll and/or
  dipole groups, and to investigate the role played by the elementary
  particles in the development of perturbation theory.

  It might be interesting to contrast this with the existing
  techniques used for ultralocal quantum field theories
  (e.g.,~\cite{Klauder:2000ud} Section 10).

\item[Relation to timelike symmetries.] 

  In~\cite{Gorantla:2022eem}, ``timelike'' higher-form global
  symmetries were used to describe fractons. It might be interesting
  to understand the relation between these generalised symmetries and
  the definition of fractons as UIRs.

\item[Lattice field theory.]

  In this work we have focused on continuous symmetries. It might be
  interesting to understand the description of these symmetries on the
  lattice.  See \cite{Gorantla:2022eem}
  and~\cite[Appendix~D]{Jensen:2022iww} for interesting comments.

\item[Generic massless and dipole particles.]

  The generic massless or generic dipole representations have quite
  exotic properties, but let us highlight some of their intricate
  features.

  Similar to the massless continuous- or infinite-spin
  representations of the Poincaré group, they are actually the generic
  case in the massless/neutral sector (e.g.,~\cite{Bekaert:2017khg} provides a
  review of the Poincaré case). Continuous-spin particles share some
  similarities with this case and it would be interesting to further
  contrast their interesting properties. Continuous-spin UIRs
  are often discarded by causality~\cite{Abbott:1976bb,Hirata:1977ss}
  arguments, but it is not clear if this applies to the case at
  hand, after all the theories are not Poincaré
  invariant. Continuous-spin particles too cannot be localised to a 
  point~(see, e.g.,~\cite{Longo:2015tra} and references therein) and
  again maybe this is a feature, rather than a bug, for carrollions
  or dipoles?

\item[Carroll/fractons and flat space holography.]

  Carrollian physics naturally emerges in the study of the asymptotic
  structure of spacetime in the absence of a cosmological
  constant~(see, e.g., \cite{Duval:2014uva, Gibbons:2019zfs,
    Figueroa-OFarrill:2021sxz, Donnay:2022aba, Bagchi:2022emh,
    Bekaert:2022oeh, Saha:2023hsl, Salzer:2023jqv, Nguyen:2023vfz}), due to the
  underlying equivalence between BMS and conformal Carroll
  algebras~\cite{Duval:2014uva}. Although our work does not focus on 
  the conformal extension, some of the structure of flat space
  holography is already dictated by Carroll symmetries
  alone:\footnote{Some of the relevance of Carroll symmetries derives
    from the conformal extension of $2+1$ dimensional Carroll
    symmetries. Even though the Carroll algebra in $2+1$ dimensions 
    allows for nontrivial central extensions, the conformal extension
    does not (basically because it is isomorphic to the Poincaré
    algebra). Therefore most of our results apply to this case.}

  \begin{enumerate}[label=(\roman*)]
  \item As shown, e.g., in~\cite{Satishchandran:2019pyc,Bekaert:2022ipg}
    the asymptotic behaviour of a massless scalar field on a flat
    spacetime at null infinity exhibits two alternative sectors: a
    radiative and a non-radiative one. From the perspective of
    (conformal) carrollian symmetries these two branches can be traced
    back to the imposition of Carroll symmetries and the choice of
    vanishing or non-vanishing energy, as we have also presented in
    this work and in Part~I.
    
    One way to see this is by looking at the two-point functions of
    Carroll field theories which have two branches by only imposing
    Carroll
    symmetries~\cite{deBoer:2021jej,Chen:2021xkw,Bagchi:2022emh}, with
    further refinements once the extended symmetries are taken into
    account.
    
  \item If we focus on the non-radiative sector and we denote the
    leading and subleading terms of the scalar field by $\phi_{0}$ and
    $\phi_{1}$, respectively, then they satisfy the following
    equations:
    \begin{equation}
      \label{eq:flat-exp}
      \dot{\phi}_{0}=0\,,\qquad\left(\Delta_{S^{2}}+h\left(h-1\right)\right)\phi_{0}=-2h\dot{\phi}_{1},
    \end{equation}
    where $h$ is a real number and $\Delta_{S^{2}}$ is the Laplacian
    on the round two-sphere. Note that these equations coincide with
    that of the magnetic Carroll scalar field~\eqref{eq:helmholtz-2},
    with $p^{2}=h\left(h-1\right)$ and $\pi=-2h\phi_{1}$. However,
    equation~\eqref{eq:flat-exp} is defined on the $2$-sphere as a
    consequence of the topology of future (or past) null infinity. It
    could be interesting to explore this relation and the possible
    role of the UIR $\Romanbar{III}(n,p)$ for flat space holography.

  \item Massless carrollions (representations $\Romanbar{III}(n,p)$)
    are zero-energy states naturally defined on the ``celestial
    sphere.'' It might be worth to explore the potential connection
    with soft degrees of freedom.
  \end{enumerate}

\item[Timelike infinity and fractons on hyperbolic space.] We also
  proposed to generalise this correspondence to curved
  space~\cite{Figueroa-OFarrill:2023vbj}. This means to add
    \begin{align}
      \label{eq:curved-gen}
      [P_{a},P_{b}] &= - \Lambda J_{ab} & [P_{a},H] &= \Lambda B_{a}
    \end{align}
    to the Carroll algebra~\eqref{eq:3-carroll-algebra}, leading to
    AdS Carroll~\cite{Figueroa-OFarrill:2018ilb,Morand:2018tke} and
    \begin{align}
      \label{eq:frac-hyp}
      [P_{a},P_{b}] &= - \Lambda J_{ab} & [P_{a},Q] &= \Lambda D_{a}
    \end{align}
    to the dipole algebra~\eqref{eq:fracton-algebra-intro}. This means
    that the underlying space geometry is now not flat but given by
    three-dimensional hyperbolic space or the $3$-sphere, for
    $\Lambda<0$ and $\Lambda >0$, respectively.  The relation to
    timelike infinity is given by the fact that AdS Carroll is the
    homogeneous model for the blow up of timelike infinity of
    asymptotically flat spacetimes~\cite{Figueroa-OFarrill:2021sxz}.

    For the quantum version of the correspondence on curved space it
    is interesting to note that AdS Carroll is a homogeneous space of
    the Poincaré group.  Consequently, the quantum particles of the
    algebras described above are the same as the ones classified by
    Wigner~\cite{Wigner:1939cj}. In this sense AdS Carroll, fractons
    on curved space and flat space (holography) are indeed connected.
    (In this context see
    also~\cite{Gromov:2017vir,Slagle:2018kqf,Bidussi:2021nmp,Jain:2021ibh}.
    It could be interesting to understand if there is any relation to
    the models discussed in~\cite{Yan:2022yix,Yan:2023lmj}).
\end{description}
The tools used in this work are not restricted to the Carroll and
dipole groups and we will discuss other particles with restricted
mobility in a future work~\cite{IP:2023}.

\acknowledgments

We thank Andrew Beckett, Jelle Hartong, Emil Have and Jakob Salzer for useful
discussions and Simon Pekar for a careful reading of an earlier
version of this draft.  The research of AP is partially supported by
Fondecyt grants No 1211226, 1220910 and 1230853. SP is supported by
the Leverhulme Trust Research Project Grant (RPG-2019-218) ``What is
Non-Relativistic Quantum Gravity and is it Holographic?''.

\appendix

\section{The method of induced representations}
\label{sec:meth-induc-repr}

In this appendix we will summarise the salient points of the method of
induced representations to construct unitary irreducible
representations of a group with an abelian normal subgroup and their
description as spacetime fields subject to the free field equations.
This was pioneered by Wigner \cite{Wigner:1939cj} for the case of the
Poincaré group and developed into a mathematical theory by Mackey.
Most of this material is standard.  See, for example, \cite{MR0495836}
or \cite{MR0213463,MR0479129,MR0479084,MR0479085}.

\subsection{Induced representations à la Mackey}
\label{sec:induc-repr-la}

Let $G = K \ltimes T$ be a connected Lie group with $T$ an abelian
normal subgroup which we will assume to be simply connected, so
isomorphic to $\RR^n$ for some $n$.  We shall let
$\g = \fk \ltimes \t$ denote their respective Lie algebras.  Due to
the semidirect product structure, $K$ acts on $T$ and hence on $\t^*$.
Let $\tau \in \t^*$ and let $\eO_\tau = K \cdot \tau$ denote its
$K$-orbit.  Let $K_\tau \subset K$ denote the stabiliser subgroup, so
that $\eO_\tau$ is $K$-equivariantly diffeomorphic to $K/K_\tau$.  We
will use in the sequel an equivalent description of $\eO_\tau$ as
$G/(K_\tau \ltimes T)$; although $G$ does not act effectively, since
$T$ acts trivially on $\eO_\tau$.  Every $\tau \in \t^*$ defines a
character $\chi_\tau$ of $T$ and hence a one-dimensional unitary
representation: if $t = \exp(X) \in T$, for some $X \in \t$, then
$\chi_\tau(t) = e^{i\left<\tau, X\right>}$ with $\left<-,-\right>
:\g^* \times \g \to \RR$ denoting the dual pairing.  We will let $W$
denote a complex unitary irreducible representation of $K_\tau$.  It
is a representation of $K_\tau \ltimes T$ via
\begin{equation}
  \label{eq:KT-on-W}
  (t h) \cdot w = \chi_\tau(t) h\cdot w
\end{equation}
for all $t \in T$, $h \in K_\tau$ and $w \in W$.  In this
representation, $K_\tau$ and $T$ commute.  Hence $t h$ and $h t$ act
in the same way.

We will let $E_W := K \times_{K_\tau} W$ denote the homogeneous vector
bundle over $\eO_\tau$ associated to $W$.  We will let $\Gamma(E_W)$
denote the space of smooth sections of $E_W \to \eO_\tau$.  Such
sections admit a different characterisation in terms of so-called
Mackey functions: smooth functions $f : K \to W$ which are
$K_\tau$-equivariant: that is, for all $k \in K$ and $h \in K_\tau$,
\begin{equation}
  f(k h) = h^{-1} \cdot f(k),
\end{equation}
where $\cdot$ stands for the linear $K_\tau$-action on $W$.
We will let $C^\infty_{K_\tau}(K,W)$ denote the vector space of Mackey
functions.  Functions on $K_\tau$ pull back to $K_\tau$-invariant
functions on $K$ and in this way, $C^\infty_{K_\tau}(K,W)$ becomes a
$C^\infty(\eO_\tau)$-module.

\begin{lemma}\label{lem:sections-are-mackey-functions}
  $\Gamma(K \times_{K_\tau} W)$ and $C^\infty_{K_\tau}(K,W)$ are
  isomorphic as $C^\infty(\eO_\tau)$-modules.
\end{lemma}

\begin{proof}[Proof (sketch)]
  Let $\sigma : \eO_\tau \to K$ be a coset representative. This may only
  be partially defined, but we will assume that it is defined in an open
  dense subset of $\eO_\tau$ which is of measure zero relative a
  $K$-invariant measure on $\eO_\tau$, which we will also assume
  exists.\footnote{By considering multiplier representations, we need
    only assume that the measure is quasi-invariant, but in the examples
    we have in mind, we will always have $K$-invariant measures.}

  If $f \in C^\infty_{K_\tau}(K,W)$, we define $\psi \in \Gamma(E_W)$ by
  $\psi(p) = f(\sigma(p))$ for all $p \in \eO_\tau$.  Conversely, if
  $\psi \in \Gamma(E_W)$, we define $f \in C^\infty_{K_\tau}(K,W)$ as
  follows: $f(\sigma(p) h) = h^{-1} \cdot \psi(p)$, where $h \in K_\tau$.
  This results in a $K_\tau$-equivariant function by construction and it
  is not hard to show that it is independent of the choice of coset
  representative.
\end{proof}

The vector bundle $E_W$ can also be described as a
homogeneous vector bundle $G \times_{K_\tau \ltimes T} W$, with the
action of $K_\tau \ltimes T$ on $W$ given by \eqref{eq:KT-on-W} and
its sections can therefore be described equivalently as
($K_\tau \ltimes T$)-equivariant functions $G \to W$.  In other words,
as in Lemma~\ref{lem:sections-are-mackey-functions}, we have an
isomorphism $C^\infty_{K_\tau}(K,W) \cong C^\infty_{K_\tau \ltimes
  T}(G,W)$ of $C^\infty(\eO_\tau)$-modules, where we lift
$C^\infty(\eO_\tau)$ to $G$ as the ($K_\tau \ltimes T$)-invariant
functions.  Under this isomorphism, a given $f \in
C^\infty_{K_\tau}(K,W)$ is sent to $F : G \to W$, defined by
\begin{equation}
  \label{eq:FGW}
  F(t k) := \chi_{k \cdot \tau}(t^{-1}) f(k)
\end{equation}
for all $t \in T$ and $k \in K$.  The function $F$ is $(K_\tau \ltimes
T)$-equivariant by construction.

The virtue of the Mackey functions $C^\infty_{K_\tau \ltimes T}(G,W)$
is that they admit a natural action of $G$ which is very easy to
describe.  For $F \in C^\infty_{K_\tau \ltimes T}(G,W)$ and all $g, g'
\in G$ we have
\begin{equation}
  \label{eq:G-action}
  (g \cdot F)(g') = F(g^{-1} g').
\end{equation}
Let $g = t h$ with $t \in T$ and $h \in K$ and choose $g'=
k \in K$.  Then,
\begin{align*}
  (t h \cdot F)(k) &= F(h^{-1} t^{-1} k)\\
                   &= F(h^{-1} k k^{-1} t^{-1} k)\\
                   &= \chi_\tau (k^{-1} t k) F(h^{-1} k)\\
                   &= \chi_{k\cdot \tau} (t) F(h^{-1}k).
\end{align*}
Therefore we conclude that $g = t h$ acts on a Mackey function $f : K \to
W$ as
\begin{equation}
  \label{eq:indrep-group-action}
  (g \cdot f)(k) = \chi_{k\cdot \tau}(t)\, f(h^{-1}k)
\end{equation}
where $k \in K$.

\begin{lemma}
  The transformed function $g\cdot f$ is $K_\tau$-equivariant.
\end{lemma}

\begin{proof}
  This follows from the fact that left- and right-multiplications
  commute and $h\cdot \tau = \tau$ for $h \in K_\tau$, so that
  \begin{align*}
    (g \cdot f)(k h') &= \chi_{k h' \cdot \tau}(t) f(h^{-1}k h')\\
                      &= \chi_{k \cdot\tau}(t) (h')^{-1} \cdot f(h^{-1}k)\\
                      &= (h')^{-1}\cdot \left( \chi_{k\cdot\tau}(t) f(h^{-1}k) \right)\\
                      &= (h')^{-1}\cdot (g\cdot f)(k).
  \end{align*}
\end{proof}

Therefore we get a representation of $G$ on $C^\infty_{K_\tau}(K,W)$.

\begin{proposition}
  \label{eq:unitarity-of-induced-rep}
  The representation of $G$ on $C^\infty_{K_\tau}(K,W)$ just described
  is unitary relative to the hermitian inner product
  \begin{equation}
    (f_1,f_2) = \int_{\eO_\tau} d\mu(k\cdot\tau) \left<f_1,f_2\right>_{W}(k\cdot \tau),
  \end{equation}
  where $d\mu$ is a $K$-invariant measure on $\eO_\tau$ and
  $\left<f_1,f_2\right>_{W}$ is function on $\eO_\tau$ which pulls
  back, via the orbit map $K \to \eO_\tau$ sending
  $k \mapsto k \cdot \tau$, to the function
  $k \mapsto \left<f_1(k),f_2(k)\right>_{W}$.
\end{proposition}

\begin{proof}
  By assumption, $W$ is a unitary representation of $K_\tau$ with
  hermitian inner product $\left<-,-\right>_{W}$ and hence for all
  $h \in K_\tau$ and $k \in K$,
  \begin{align*}
    \left<f_1(kh),f_2(kh)\right>_{W} &=  \left<h^{-1}\cdot f_1(k), h^{-1} \cdot f_2(k)\right>_{W} &\tag{since $f_1,f_2$ are equivariant}\\
                                     &=  \left<f_1(k), f_2(k)\right>_{W}, &\tag{since $\left<-,-\right>_{W}$ is $K_\tau$-invariant}\\
  \end{align*}
  hence the function is the pull-back of a unique function on $\eO_\tau$
  and it is that function that we integrate against the invariant
  measure.
  
  Unitarity of the $G$-representation now follows because with $g = th$
  \begin{align*}
    (g \cdot f_1, g\cdot f_2)  &= \int_{\eO_\tau} d\mu(k \cdot \tau) \left<g \cdot f_1, g \cdot f_2 \right>_{W}(k \cdot \tau)\\
                               &= \int_{\eO_\tau} d\mu(k \cdot \tau) \left< \chi_{k \cdot \tau}(t) f_1(h^{-1} k), \chi_{k \cdot \tau}(t) f_2(h^{-1} k) \right>_{W}\\
                               &= \int_{\eO_\tau} d\mu(k \cdot \tau) \left< f_1(h^{-1} k), f_2(h^{-1} k) \right>_{W} &\tag{since $\left<-,-\right>_{W}$ is hermitian}\\
                               &= \int_{\eO_\tau} d\mu(k \cdot \tau) \left< f_1, f_2\right>_{W}(h^{-1}k\cdot  \tau)\\
                               &= \int_{\eO_\tau} d\mu(hk'\cdot\tau)) \left< f_1, f_2\right>_{W}(k'\cdot\tau) &\tag{changing variables to $k'= h^{-1}k$}\\
                               &= \int_{\eO_\tau} d\mu(k'\cdot\tau) \left< f_1, f_2\right>_{W}(k'\cdot\tau) &\tag{invariance of the measure}\\
                               &= (f_1, f_2).
  \end{align*}
\end{proof}

We can transport this unitary representation on
$C^\infty_{K_\tau}(K,W)$ to a unitary representation of $G$ on
sections of $E_W$.  Explicitly, if $\psi \in \Gamma(E_W)$ and $g \in
G$,
\begin{equation}
  \label{eq:G-action-on-sections}
  (g \cdot \psi)(p) := (g \cdot F)(\sigma(p)) = F(g^{-1}\sigma(p)).
\end{equation}
Using the product
\begin{equation}
  g^{-1}\sigma(p) = \sigma(g^{-1}\cdot p) h(g^{-1},p),
\end{equation}
where $h(g^{-1},p) \in K_\tau \ltimes T$, we can rewrite
equation~\eqref{eq:G-action-on-sections} as
\begin{equation}
  \label{eq:G-action-sections-final}
  (g\cdot \psi)(p) = h(g^{-1},p)^{-1} \cdot \psi(g^{-1}\cdot p).
\end{equation}
The function $h: G \times \eO_\tau \to K_\tau \ltimes T$ satisfies
some cocycle properties which guarantee that the above is indeed a
representation of $G$.

It is a fundamental result in this subject that if $W$ is both unitary
and irreducible as a representation of $K_\tau$, then so is (the
Hilbert space completion of the square-integrable sections in)
$\Gamma(K \times_{K_\tau} W)$ as a representation of $G$, which we
denote by $L^2(\eO_\tau, K \times_{K_\tau} W)$.  We proved
unitarity in Proposition~\ref{eq:unitarity-of-induced-rep}.
Irreducibility follows from a straightforward application of the SNAG
theorem, as we now briefly sketch.

\begin{proposition}
  \label{prop:irreducibility-of-induced-rep}
  Let $\eH = L^2(\eO_\tau,K\times_{K_\tau} W)$ and let $U : G \to
  \U(\eH)$ be the unitary induced representation constructed above.
  Then if $W$ is an irreducible representation of $K_\tau$, $\eH$ is
  an irreducible representation of $G$.
\end{proposition}

\begin{proof}[Proof (sketch)]
  Let $\eH'\subset \eH$ be a $G$-invariant subspace of $\eH$.  Then the
  orthogonal projection onto $\eH'$ is a continuous operator on $\eH$
  which commutes with the action of $G$.  We are done if we show that
  any such operator is necessarily a multiple of the identity, so that
  $\eH'$ is not then a proper subspace.

  Let $A$ be a continuous operator on $\eH$ which commutes with the
  action of $G$.  In particular it commutes with the action of the
  translation subgroup $T \subset G$.  By the SNAG theorem (see, e.g.,
  \cite[Section~6.2]{MR0495836}), $A$ acts pointwise
  $(A\cdot \psi)(p) = A(p) \cdot \psi(p)$ for all
  $p = \sigma(p) \cdot \tau \in \eO_\tau$ and since $A$ commutes with
  the action of $\sigma(p)$, it follows that
  \begin{align*}
    A(p) \cdot \psi(p) &= (A \cdot \psi)(p)\\
                       &= \left( U(\sigma(p)^{-1}) \cdot A \cdot \psi \right)(\tau)\\
                       &= (A \cdot U(\sigma(p)^{-1}) \cdot \psi)(\tau)\\
                       &= A(\tau) \cdot (U(\sigma(p)^{-1}) \cdot \psi)(\tau)\\
                       &= A(\tau) \cdot \psi(p),
  \end{align*}
  so that $A(p) = A(\tau)$ for all $p \in \eO_\tau$.  But $A(\tau) \in
  \End W$ commutes with the action of $K_\tau$ and since by
  hypothesis $W$ is a complex irreducible representation of $K_\tau$,
  Schur's Lemma guarantees that $A(\tau)$ is a multiple of the
  identity, and hence so is $A$.
\end{proof}

\subsection{Induced representations as free field theories }
\label{sec:induc-repr-as}

Mackey theory exhibits UIRs of $G$ as (square-integrable) sections of
bundles over $\eO_\tau$, which in the case of kinematical groups such
as the Poincaré, Galilei or Carroll groups, is an orbit in momentum
space.  It is however often desirable to exhibit the representation as
fields on the kinematical spacetime: Minkowski, Galilei or Carroll,
say.  Such a spacetime is a homogeneous space of $G$ which is
$G$-equivariantly diffeomorphic to $G/K$.  Such fields are therefore
sections of homogeneous vector bundles over $G/K$ associated to
representations of $K$.  Since we only have $W$, which is a
representation of $K_\tau$, this entails a \emph{choice}: namely, that
of a representation $V$ of $K$ which, when restricted to $K_\tau$,
contains a subrepresentation isomorphic to $W$.  The representation
$V$ need not be unitary, of course.

Given $f \in C^\infty_{K_\tau}(K,V)$ we get $F \in C^\infty_{K_\tau
  \ltimes T}(G,V)$ as before, by having $T$ act via the character
$\chi_\tau$.  We now define $\widehat F : G \to V$ by
\begin{equation}
  \label{eq:pre-Fourier-transform}
  \widehat F(g) := \int_{\eO_\tau} d\mu(p) \sigma(p) \cdot F(g \sigma(p))
\end{equation}
where $d\mu$ is a $K$-invariant measure on $\eO_\tau$ and we are
assuming that the coset representative $\sigma$ is defined in the
complement of a set of measure zero.  We now show that $\widehat F$ is
$K$-equivariant.

\begin{proposition}
 $\widehat F \in C^\infty_K(G,V)$.
\end{proposition}

\begin{proof}
  For all $k \in K$ and $g \in G$,
  \begin{equation}
    \widehat F(gk) = \int_{\eO_\tau} d\mu(p) \sigma(p) F(gk\sigma(p)).
  \end{equation}
  We now have
  \begin{equation}
    k\sigma(p) = \sigma(k\cdot p) h(k,p)
  \end{equation}
  for some $h(k,p) \in K_\tau$.  Since $F$ is in particular
  $K_\tau$-equivariant,
  \begin{equation}
    F(gk\sigma(p)) = F(g\sigma(k\cdot p)h(k,p)) = h(k,p)^{-1} \cdot F(g\sigma(k\cdot p)).
  \end{equation}
  But now notice that $\sigma(p) h(k,p)^{-1} = k^{-1} \sigma(k\cdot p)$,
  so that
  \begin{equation}
    \widehat F(gk) = \int_{\eO_\tau} d\mu(p) k^{-1}\sigma(k\cdot p) F(g\sigma(k\cdot p)).
  \end{equation}
  Let $k \cdot p = p'$.  By the invariance of the measure, $d\mu(p) =
  d\mu(p')$ and hence
  \begin{equation}
    \widehat F(gk) = \int_{\eO_\tau} d\mu(p') k^{-1}\sigma(p')
    F(g\sigma(p')) = k^{-1} \cdot \widehat F(g).
  \end{equation}
\end{proof}

It follows that $\widehat F$ defines a section of the homogeneous
vector bundle $G \times_K V$ over $G/K$, which we can describe
explicitly as follows.   Let $\zeta : G/K \to T$  be a coset
representative.  Again this may only be locally defined, but we assume
it is defined on the complement of a set of measure zero on $G/K$
relative to a $G$-invariant measure.  We define a section $\phi$ of $G
\times_K V$ by
\begin{equation}
  \phi(x) := \widehat F(\zeta(x)) = \int_{\eO_\tau} d\mu(p) \sigma(p)
  \cdot F(\zeta(x)\sigma(p)).
\end{equation}
Since $T$ is a normal subgroup, we can write $\zeta(x)\sigma(p) =
\sigma(p) \zeta(x')$, where
\begin{equation}
  \zeta(x') = \sigma^{-1}(p) \zeta(x) \sigma(p),
\end{equation}
and moreover
\begin{equation}
  \chi_\tau(\zeta(x')^{-1}) =\chi_{\sigma(p)\cdot \tau} (\zeta(x)^{-1}),
\end{equation}
so that
\begin{equation}
  \label{eq:fourier-transform}
  \phi(x) =\int_{\eO_\tau} d\mu(p) \chi_{\sigma(p)\cdot \tau} (\zeta(x)^{-1}) \sigma(p) \cdot \psi(p),
\end{equation}
where we have used that $F(\sigma(p)) = \psi(p)$.  If $\zeta(x) =
\exp(X)$, then
\begin{equation}
  \chi_{\sigma(p)\cdot \tau}(\zeta(x)^{-1}) = e^{-i \left<
      \sigma(p)\cdot\tau, X\right>}
\end{equation}
so that
\begin{equation}
  \phi(x) = \int_{\eO_\tau} d\mu(p) e^{-i \left<\sigma(p)\cdot\tau, X\right>} \sigma(p) \cdot \psi(p)
\end{equation}
is seen to be a group-theoretical generalisation of the Fourier
transform: it relates a section $\psi$ of $K \times_{K_\tau} W$ over
$\eO_\tau$ to a section $\phi$ of $G \times_K V$ over $G/K$.  It bears
reminding that we have made a choice of representation $V$.  Other
choices (such as coset representatives) are immaterial.

Finally, the $G$-action on $\phi$ is given by
\begin{equation}
  \label{eq:G-action-fields}
  (g \cdot \phi)(x) := \widehat F(g^{-1}\zeta(x)).
\end{equation}
We may expand this using
\begin{equation}
  g^{-1} \zeta(x) = \zeta(g^{-1}\cdot x) k(g^{-1},x)
\end{equation}
where $k : G \times G/K \to K$ is thus defined.  Then
\begin{equation}
  \widehat F(g^{-1}\zeta(x)) = \widehat F(\zeta(g^{-1}\cdot x)
 k(g^{-1},x)) = k(g^{-1},x)^{-1} \cdot \widehat F(\zeta(g^{-1} \cdot
 x)),
\end{equation}
or, finally,
\begin{equation}
  \label{eq:G-action-fields-final}
  (g \cdot \phi)(x) = k(g^{-1},x)^{-1} \cdot \phi(g^{-1}\cdot x).
\end{equation}

In those cases where $V$ properly contains $W$, the representation of
$G$ on sections of $G \times_K V$ is not irreducible.  To restore
irreducibility we need to somehow project our fields to $W$.  This is
an algebraic operation on the fibre of $K \times_{K_\tau} V$ at the
identity coset in $\eO_\tau$.  We can then extend it to a
point-dependent projector to the sub-bundle $K \times_{K_\tau} W
\subset K \times_{K_\tau} V$ and via the generalised Fourier transform
they become (pseudo-)differential operators which ought to be
interpreted as free field equations for sections of $G\times_K V$.

To see how this goes about, let us assume that $W = \ker \Phi$ for
some $\Phi \in \End V$.  (This does not mean that $\Phi$ is
$K_\tau$-equivariant, by the way.)  Let $F \in C^\infty_{K_\tau
  \ltimes T}(G,V)$ actually take values in $W$, so that it is in the
image of the natural embedding $C^\infty_{K_\tau \ltimes T}(G,W)
\subset C^\infty_{K_\tau \ltimes T}(G,V)$.  Then for all $g \in G$,
$\Phi F(g)=0$ and hence integrating,
\begin{equation}
  \int_{\eO_\tau} d\mu(p) \sigma(p) \Phi F(g\sigma(p)) = 0.
\end{equation}
In particular this is true for $g = \zeta(x)$.  We can rewrite this as
\begin{equation}
  0 = \int_{\eO_\tau} d\mu(p) \underbrace{\sigma(p) \Phi \sigma(p)^{-1}}_{=:\Phi_p} \sigma(p) F(\zeta(x)\sigma(p)).
\end{equation}
As above, we have that $\zeta(x)\sigma(p) = \sigma(p)\zeta(x')$ and hence
\begin{align*}
  0 &= \int_{\eO_\tau} d\mu(p) \Phi_p \sigma(p) F(\sigma(p)\zeta(x')\\
    &= \int_{\eO_\tau} d\mu(p) \chi_\tau(\zeta(x')^{-1}) \Phi_p \sigma(p) F(\sigma(p)) & \tag{using the equivariance of $F$}\\
    &=\int_{\eO_\tau} d\mu(p) \chi_{\sigma(p)\cdot \tau}(\zeta(x))^{-1} \Phi_p \sigma(p) \psi(p).
\end{align*}
The above integral defines the action of a pseudo-differential
operator $\widehat\Phi_x$ on the field $\phi(x)$ given by
\begin{align*}
  \widehat\Phi_x \phi(x) &= \widehat \Phi_x  \int_{\eO_\tau} d\mu(p) \chi_{\sigma(p)\cdot \tau} (\zeta(x)^{-1}) \sigma(p) \cdot \psi(p) &\tag{by equation~\eqref{eq:fourier-transform}}\\
  &= \int_{\eO_\tau} d\mu(p) \chi_{\sigma(p)\cdot \tau}(\zeta(x))^{-1} \Phi_p \sigma(p) \psi(p).
\end{align*}
In some cases, depending on the form of the Fourier-like kernel
$\chi_{\sigma(p)\cdot   \tau}(\zeta(x))^{-1}$ and the form of
$\Phi_p$, $\widehat\Phi_x$ is an honest differential operator; but in
any case, we obtain a field equation (which may be non-local):
$\widehat\Phi_x \phi(x) = 0$.

\section{Hopf charts on $\SU(2)$}
\label{sec:hopf-charts-su2}

In this appendix we record some useful charts on $\SU(2)$ adapted to
the Hopf fibration $\SU(2)\to S^2$ which play a rôle in our
description of UIRs of the Carroll group of class $\Romanbar{V}$.

The Lie group $\SU(2)$ is diffeomorphic to the $3$-sphere and this is
made transparent by embedding the $3$-sphere in $\CC^2$ as the unit
sphere: $S^3 = \left\{(z_1,z_2) \in \CC^2 ~\middle |~ |z_1|^2 + |z_2|^2 =
  1\right\}$ and then identifying $(z_1,z_2) \in S^3 \subset \CC^2$
with the special unitary matrix $\begin{pmatrix} z_1 & z_2 \\ -\zbar_2
  & \zbar_1\end{pmatrix}$.  Let us define the following open subsets
of $S^3$:
\begin{equation}
  \begin{split}
    U_1 &= \left\{ (z_1,z_2) \in S^3   ~\middle |~ z_1 \neq 0\right\}\\
    U_2 &= \left\{ (z_1,z_2) \in S^3   ~\middle |~ z_2 \neq 0\right\}.
  \end{split}
\end{equation}
Clearly, $S^3 = U_1 \cup U_2$.  We define surjective maps
$\pi_1 : U_1 \to \CC$ and $\pi_2 : U_2 \to \CC$ by $\pi_1(z_1,z_2) =
z_2/z_1$ and $\pi_2(z_1,z_2) = z_1/z_2$, which are nothing but the
restriction of the Hopf fibration $\SU(2) \to S^2$ to each of $U_1$
and $U_2$ composed with stereographic projection from (the complement
of a point in) $S^2$ to $\CC$.

For $z \in \CC$, the fibre $\pi_1^{-1}(z)$ consists of those
$(z_1,z_2)$ such that $z_2/z_1 = z$ and $|z_1|^2 + |z_2|^2 = 1$; that
is,
\begin{equation}
  \pi_1^{-1}(z) = \left\{ \left( \frac{\zeta}{\sqrt{1 + |z|^2}},
      \frac{z \zeta}{\sqrt{1 + |z|^2}} \right) ~ \middle |~  |\zeta| =
  1\right\},
\end{equation}
which is diffeomorphic to a circle.  Similarly,
\begin{equation}
  \pi_2^{-1}(z) = \left\{ \left( \frac{z \zeta}{\sqrt{1 + |z|^2}},
      \frac{\zeta}{\sqrt{1 + |z|^2}} \right) ~ \middle |~  |\zeta| =
  1\right\}.
\end{equation}
This allows us to establish charts\footnote{We use the word loosely,
  since $\CC \times S^1$ is not diffeomorphic to an open subset of
  $\RR^3$.} $\varphi_1 : U_1 \to \CC \times S^1$ and $\varphi_2 : U_2
\to \CC \times S^1$ by
\begin{equation}
  \varphi_1(z_1,z_2) = \left( \frac{z_2}{z_1}, \frac{z_1}{|z_1|}
  \right)\qquad\text{and}\qquad
  \varphi_2(z_1,z_2) = \left( \frac{z_1}{z_2}, \frac{z_2}{|z_2|}
  \right),
\end{equation}
whose inverses give parametrisations of $U_1$ and $U_2$ in terms of
$\CC \times S^1$:
\begin{equation}
  \varphi_1^{-1}(z,\zeta) = \left(\frac{\zeta}{\sqrt{1+|z|^2}},
    \frac{z \zeta}{\sqrt{1+|z|^2}}  \right)
  \qquad\text{and}\qquad
  \varphi_2^{-1}(z,\zeta) = \left(\frac{z \zeta}{\sqrt{1+|z|^2}},
    \frac{\zeta}{\sqrt{1+|z|^2}}  \right).
\end{equation}
On the overlap $U_1 \cap U_2$, we have that the transition functions
are given by
\begin{equation}
  \varphi_1 \circ \varphi_2^{-1} : (z,\zeta) \mapsto \left( \frac1z, \frac{z}{|z|}\zeta \right)
\end{equation}
and a formally identical expression for $\varphi_2 \circ \varphi_1^{-1}$.

Let $g_1: \CC \times S^1 \to \SU(2)$ and $g_2 : \CC \times S^1 \to
\SU(2)$ be the compositions of the parametrisations with the
identification between $S^3 \subset \CC^2$ and $\SU(2)$.  Explicitly,
\begin{equation}
  g_1(z,\zeta) = \frac1{\sqrt{1 + |z|^2}}
  \begin{pmatrix}
    \zeta & z \zeta \\-\zbar \zeta^{-1} & \zeta^{-1}
  \end{pmatrix}
\end{equation}
and
\begin{equation}
  g_2(z,\zeta) = \frac{1}{\sqrt{1 + |z|^2}}
  \begin{pmatrix}
    z \zeta & \zeta \\ - \zeta^{-1} & \zeta^{-1} z
  \end{pmatrix}.
\end{equation}
We may use these maps to pull-back to $\CC \times S^1$ the
left-invariant (say) Maurer--Cartan one-form on $\SU(2)$.  Let us work
with $g_2$ for definiteness:
\begin{equation}
  g_2^{-1} dg_2 = \frac{1}{1+|z|^2}
  \begin{pmatrix}
    \tfrac12 (\zbar d z - z d\zbar) + (|z|^2-1) \frac{d\zeta}{\zeta} &  -d\zbar + 2 \zbar \frac{d\zeta}{\zeta} \\
     dz + 2 z \frac{d\zeta}{\zeta} & -\tfrac12 (\zbar d z - z d\zbar)- (|z|^2-1) \frac{d\zeta}{\zeta}
   \end{pmatrix}
\end{equation}
The round metric on $S^3$ agrees with the natural bi-invariant metric
on $\SU(2)$ (both defined up to scale), whose volume form is given by
\begin{equation}
  \label{eq:volume-form-round-3-sphere}
  \mathrm{dvol} = -\tfrac13 \Tr \left(g_2^{-1}dg_2\right)^3 = \frac{2 dz \wedge
  d\zbar}{(1+|z|^2)^2} \frac{d\zeta}{\zeta},
\end{equation}
which agrees with the invariant measure in the inner product of the
Carroll UIRs of class $\Romanbar{V}$.

\providecommand{\href}[2]{#2}\begingroup\raggedright\endgroup


\begin{thebibliography}{10}

\bibitem{Levy1965}
J.-M. Lévy-Leblond, ``Une nouvelle limite non-relativiste du groupe de
  {P}oincaré,'' {\em Annales de l'I.H.P. Physique théorique} {\bfseries 3}
  no.~1, (1965) 1--12. \url{http://eudml.org/doc/75509}.

\bibitem{SenGupta1966OnAA}
N.~D.~S. Gupta, ``On an analogue of the {G}alilei group,'' {\em Il Nuovo
  Cimento A (1965-1970)} {\bfseries 44} (1966) 512--517.

\bibitem{Gromov:2018nbv}
A.~Gromov, ``{Towards classification of Fracton phases: the multipole
  algebra},'' \href{http://dx.doi.org/10.1103/PhysRevX.9.031035}{{\em Phys.
  Rev. X} {\bfseries 9} no.~3, (2019) 031035},
  \href{http://arxiv.org/abs/1812.05104}{{\ttfamily arXiv:1812.05104
  [cond-mat.str-el]}}.

\bibitem{Figueroa-OFarrill:2023vbj}
J.~Figueroa-O'Farrill, A.~P\'erez, and S.~Prohazka, ``{Carroll/fracton
  particles and their correspondence},''
  \href{http://dx.doi.org/10.1007/JHEP06(2023)207}{{\em JHEP} {\bfseries 06}
  (5, 2023) 207}, \href{http://arxiv.org/abs/2305.06730}{{\ttfamily
  arXiv:2305.06730 [hep-th]}}.

\bibitem{Nandkishore:2018sel}
R.~M. Nandkishore and M.~Hermele, ``{Fractons},''
  \href{http://dx.doi.org/10.1146/annurev-conmatphys-031218-013604}{{\em Ann.
  Rev. Condensed Matter Phys.} {\bfseries 10} (2019) 295--313},
  \href{http://arxiv.org/abs/1803.11196}{{\ttfamily arXiv:1803.11196
  [cond-mat.str-el]}}.

\bibitem{Pretko:2020cko}
M.~Pretko, X.~Chen, and Y.~You, ``{Fracton Phases of Matter},''
  \href{http://dx.doi.org/10.1142/S0217751X20300033}{{\em Int. J. Mod. Phys. A}
  {\bfseries 35} no.~06, (2020) 2030003},
  \href{http://arxiv.org/abs/2001.01722}{{\ttfamily arXiv:2001.01722
  [cond-mat.str-el]}}.

\bibitem{Grosvenor:2021hkn}
K.~T. Grosvenor, C.~Hoyos, F.~Pe\~na Benitez, and P.~Sur\'owka,
  ``{Space-Dependent Symmetries and Fractons},''
  \href{http://dx.doi.org/10.3389/fphy.2021.792621}{{\em Front. in Phys.}
  {\bfseries 9} (2022) 792621},
  \href{http://arxiv.org/abs/2112.00531}{{\ttfamily arXiv:2112.00531
  [hep-th]}}.

\bibitem{Wigner:1939cj}
E.~P. Wigner, ``{On Unitary Representations of the Inhomogeneous Lorentz
  Group},'' \href{http://dx.doi.org/10.2307/1968551}{{\em Annals Math.}
  {\bfseries 40} (1939) 149--204}.

\bibitem{Bidussi:2021nmp}
L.~Bidussi, J.~Hartong, E.~Have, J.~Musaeus, and S.~Prohazka, ``{Fractons,
  dipole symmetries and curved spacetime},''
  \href{http://dx.doi.org/10.21468/SciPostPhys.12.6.205}{{\em SciPost Phys.}
  {\bfseries 12} no.~6, (2022) 205},
  \href{http://arxiv.org/abs/2111.03668}{{\ttfamily arXiv:2111.03668
  [hep-th]}}.

\bibitem{Marsot:2022imf}
L.~Marsot, P.~M. Zhang, M.~Chernodub, and P.~A. Horvathy, ``{Hall effects in
  Carroll dynamics},'' \href{http://arxiv.org/abs/2212.02360}{{\ttfamily
  arXiv:2212.02360 [hep-th]}}.

\bibitem{Bergshoeff:2014jla}
E.~Bergshoeff, J.~Gomis, and G.~Longhi, ``{Dynamics of Carroll Particles},''
  \href{http://dx.doi.org/10.1088/0264-9381/31/20/205009}{{\em Class. Quant.
  Grav.} {\bfseries 31} no.~20, (2014) 205009},
\href{http://arxiv.org/abs/1405.2264}{{\ttfamily arXiv:1405.2264 [hep-th]}}.

\bibitem{deBoer:2021jej}
J.~de~Boer, J.~Hartong, N.~A. Obers, W.~Sybesma, and S.~Vandoren, ``{Carroll
  Symmetry, Dark Energy and Inflation},''
  \href{http://dx.doi.org/10.3389/fphy.2022.810405}{{\em Front. in Phys.}
  {\bfseries 10} (2022) 810405},
  \href{http://arxiv.org/abs/2110.02319}{{\ttfamily arXiv:2110.02319
  [hep-th]}}.

\bibitem{Zhang:2023jbi}
P.~M. Zhang, H.-X. Zeng, and P.~A. Horvathy, ``{MultiCarroll dynamics},''
  \href{http://arxiv.org/abs/2306.07002}{{\ttfamily arXiv:2306.07002 [gr-qc]}}.

\bibitem{griffiths1999introduction}
D.~Griffiths, {\em Introduction to Electrodynamics}.
\newblock Prentice Hall, 1999.

\bibitem{Duval:2014lpa}
C.~Duval, G.~Gibbons, and P.~Horvathy, ``{Conformal Carroll groups},''
  \href{http://dx.doi.org/10.1088/1751-8113/47/33/335204}{{\em J. Phys. A}
  {\bfseries 47} no.~33, (2014) 335204},
  \href{http://arxiv.org/abs/1403.4213}{{\ttfamily arXiv:1403.4213 [hep-th]}}.

\bibitem{Figueroa-OFarrill:2021sxz}
J.~Figueroa-O'Farrill, E.~Have, S.~Prohazka, and J.~Salzer, ``{Carrollian and
  celestial spaces at infinity},''
  \href{http://dx.doi.org/10.1007/JHEP09(2022)007}{{\em JHEP} {\bfseries 09}
  (2022) 007}, \href{http://arxiv.org/abs/2112.03319}{{\ttfamily
  arXiv:2112.03319 [hep-th]}}.

\bibitem{Gibbons:2019zfs}
G.~W. Gibbons, ``{The Ashtekar-Hansen universal structure at spatial infinity
  is weakly pseudo-Carrollian},''
  \href{http://arxiv.org/abs/1902.09170}{{\ttfamily arXiv:1902.09170 [gr-qc]}}.

\bibitem{Travaglini:2022uwo}
G.~Travaglini {\em et~al.}, ``{The SAGEX review on scattering amplitudes},''
  \href{http://dx.doi.org/10.1088/1751-8121/ac8380}{{\em J. Phys. A} {\bfseries
  55} no.~44, (2022) 443001}, \href{http://arxiv.org/abs/2203.13011}{{\ttfamily
  arXiv:2203.13011 [hep-th]}}.

\bibitem{Jain:2021ibh}
A.~Jain and K.~Jensen, ``{Fractons in curved space},''
  \href{http://dx.doi.org/10.21468/SciPostPhys.12.4.142}{{\em SciPost Phys.}
  {\bfseries 12} no.~4, (2022) 142},
  \href{http://arxiv.org/abs/2111.03973}{{\ttfamily arXiv:2111.03973
  [hep-th]}}.

\bibitem{Duval:2014uoa}
C.~Duval, G.~Gibbons, P.~Horvathy, and P.~Zhang, ``{Carroll versus Newton and
  Galilei: two dual non-Einsteinian concepts of time},''
  \href{http://dx.doi.org/10.1088/0264-9381/31/8/085016}{{\em Class. Quant.
  Grav.} {\bfseries 31} (2014) 085016},
  \href{http://arxiv.org/abs/1402.0657}{{\ttfamily arXiv:1402.0657 [gr-qc]}}.

\bibitem{Henneaux:2021yzg}
M.~Henneaux and P.~Salgado-Rebolledo, ``{Carroll contractions of
  Lorentz-invariant theories},''
  \href{http://dx.doi.org/10.1007/JHEP11(2021)180}{{\em JHEP} {\bfseries 11}
  (2021) 180}, \href{http://arxiv.org/abs/2109.06708}{{\ttfamily
  arXiv:2109.06708 [hep-th]}}.

\bibitem{Klauder:1970cs}
J.~R. Klauder, ``{Ultralocal scalar field models},''
  \href{http://dx.doi.org/10.1007/BF01649449}{{\em Commun. Math. Phys.}
  {\bfseries 18} (1970) 307--318}.

\bibitem{Klauder:1971zz}
J.~R. Klauder, ``{Ultralocal quantum field theory},''
  \href{http://dx.doi.org/10.1007/978-3-7091-8284-0_10}{{\em Acta Phys.
  Austriaca Suppl.} {\bfseries 8} (1971) 227--276}.

\bibitem{levygalgr}
J.-M. Levy-L\'{e}blond,
  \href{http://dx.doi.org/10.1016/B978-0-12-455152-7.50011-2}{``Galilei group
  and galilean invariance,''} in {\em Group Theory and its Applications}, E.~M.
  Loebl, ed., pp.~221 -- 299.
\newblock Academic Press, 1971.

\bibitem{Figueroa-OFarrill:2018ilb}
J.~Figueroa-O'Farrill and S.~Prohazka, ``{Spatially isotropic homogeneous
  spacetimes},'' \href{http://dx.doi.org/10.1007/JHEP01(2019)229}{{\em JHEP}
  {\bfseries 01} (2019) 229},
\href{http://arxiv.org/abs/1809.01224}{{\ttfamily arXiv:1809.01224 [hep-th]}}.

\bibitem{Deser:1981wh}
S.~Deser, R.~Jackiw, and S.~Templeton, ``{Topologically Massive Gauge
  Theories},'' \href{http://dx.doi.org/10.1006/aphy.2000.6013,
  10.1016/0003-4916(82)90164-6}{{\em Annals Phys.} {\bfseries 140} (1982)
  372--411}.
[Annals Phys.281,409(2000)].

\bibitem{Deser:1982vy}
S.~Deser, R.~Jackiw, and S.~Templeton, ``{Three-Dimensional Massive Gauge
  Theories},''
\href{http://dx.doi.org/10.1103/PhysRevLett.48.975}{{\em Phys. Rev. Lett.}
  {\bfseries 48} (1982) 975--978}.

\bibitem{Jensen:2022iww}
K.~Jensen and A.~Raz, ``{Large $N$ fractons},''
  \href{http://arxiv.org/abs/2205.01132}{{\ttfamily arXiv:2205.01132
  [hep-th]}}.

\bibitem{MR387499}
J.~H. Rawnsley, ``Representations of a semi-direct product by quantization,''
  \href{http://dx.doi.org/10.1017/S0305004100051793}{{\em Math. Proc. Cambridge
  Philos. Soc.} {\bfseries 78} no.~2, (1975) 345--350}.

\bibitem{Oblak:2016eij}
B.~Oblak, \href{http://dx.doi.org/10.1007/978-3-319-61878-4}{{\em {BMS
  Particles in Three Dimensions}}}.
\newblock PhD thesis, Brussels U., 2016.
\newblock
\href{http://arxiv.org/abs/1610.08526}{{\ttfamily arXiv:1610.08526 [hep-th]}}.
\newblock

\bibitem{MR0495836}
A.~O. Barut and R.~R\k{a}czka, {\em Theory of group representations and
  applications}.
\newblock PWN---Polish Scientific Publishers, Warsaw, 1977.

\bibitem{MR1486137}
W.~Graham and D.~A. Vogan, Jr., ``Geometric quantization for nilpotent
  coadjoint orbits,'' in {\em Geometry and representation theory of real and
  {$p$}-adic groups ({C}\'{o}rdoba, 1995)}, vol.~158 of {\em Progr. Math.},
  pp.~69--137.
\newblock Birkh\"{a}user Boston, Boston, MA, 1998.

\bibitem{deBoer:2023fnj}
J.~de~Boer, J.~Hartong, N.~A. Obers, W.~Sybesma, and S.~Vandoren, ``{Carroll
  stories},'' \href{http://arxiv.org/abs/2307.06827}{{\ttfamily
  arXiv:2307.06827 [hep-th]}}.

\bibitem{Klauder:2000ud}
J.~R. Klauder, {\em {Beyond conventional quantization}}.
\newblock Cambridge University Press, 12, 2005.

\bibitem{Pretko:2018jbi}
M.~Pretko, ``{The Fracton Gauge Principle},''
  \href{http://dx.doi.org/10.1103/PhysRevB.98.115134}{{\em Phys. Rev. B}
  {\bfseries 98} no.~11, (2018) 115134},
  \href{http://arxiv.org/abs/1807.11479}{{\ttfamily arXiv:1807.11479
  [cond-mat.str-el]}}.

\bibitem{Casalbuoni:2021fel}
R.~Casalbuoni, J.~Gomis, and D.~Hidalgo, ``{Worldline description of
  fractons},'' \href{http://dx.doi.org/10.1103/PhysRevD.104.125013}{{\em Phys.
  Rev. D} {\bfseries 104} no.~12, (2021) 125013},
  \href{http://arxiv.org/abs/2107.09010}{{\ttfamily arXiv:2107.09010
  [hep-th]}}.

\bibitem{Seiberg:2019vrp}
N.~Seiberg, ``{Field Theories With a Vector Global Symmetry},''
  \href{http://dx.doi.org/10.21468/SciPostPhys.8.4.050}{{\em SciPost Phys.}
  {\bfseries 8} no.~4, (2020) 050},
  \href{http://arxiv.org/abs/1909.10544}{{\ttfamily arXiv:1909.10544
  [cond-mat.str-el]}}.

\bibitem{Gorantla:2022eem}
P.~Gorantla, H.~T. Lam, N.~Seiberg, and S.-H. Shao, ``{Global Dipole Symmetry,
  Compact Lifshitz Theory, Tensor Gauge Theory, and Fractons},''
  \href{http://arxiv.org/abs/2201.10589}{{\ttfamily arXiv:2201.10589
  [cond-mat.str-el]}}.

\bibitem{Brauner:2020rtz}
T.~Brauner, ``{Field theories with higher-group symmetry from composite
  currents},'' \href{http://dx.doi.org/10.1007/JHEP04(2021)045}{{\em JHEP}
  {\bfseries 04} (2021) 045}, \href{http://arxiv.org/abs/2012.00051}{{\ttfamily
  arXiv:2012.00051 [hep-th]}}.

\bibitem{Mehra:2023rmm}
A.~Mehra and A.~Sharma, ``{Towards Carrollian quantization: renormalization of
  Carrollian electrodynamics},''
  \href{http://arxiv.org/abs/2302.13257}{{\ttfamily arXiv:2302.13257
  [hep-th]}}.

\bibitem{Banerjee:2023jpi}
K.~Banerjee, R.~Basu, B.~Krishnan, S.~Maulik, A.~Mehra, and A.~Ray, ``{One-Loop
  Quantum Effects in Carroll Scalars},''
  \href{http://arxiv.org/abs/2307.03901}{{\ttfamily arXiv:2307.03901
  [hep-th]}}.

\bibitem{Islam:2023rnc}
M.~Islam, ``{Carrollian Yang-Mills theory},''
  \href{http://dx.doi.org/10.1007/JHEP05(2023)238}{{\em JHEP} {\bfseries 05}
  (2023) 238}, \href{http://arxiv.org/abs/2301.00953}{{\ttfamily
  arXiv:2301.00953 [hep-th]}}.

\bibitem{Islam:2023iju}
M.~Islam, ``{BRST Symmetry of Non-Lorentzian Yang-Mills Theory},''
  \href{http://arxiv.org/abs/2306.04241}{{\ttfamily arXiv:2306.04241
  [hep-th]}}.

\bibitem{Isham:1975ur}
C.~J. Isham, ``{Some Quantum Field Theory Aspects of the Superspace
  Quantization of General Relativity},''
  \href{http://dx.doi.org/10.1098/rspa.1976.0138}{{\em Proc. Roy. Soc. Lond. A}
  {\bfseries 351} (1976) 209--232}.

\bibitem{Teitelboim:1981ua}
C.~Teitelboim, ``{Quantum Mechanics of the Gravitational Field},''
  \href{http://dx.doi.org/10.1103/PhysRevD.25.3159}{{\em Phys. Rev. D}
  {\bfseries 25} (1982) 3159}.

\bibitem{Henneaux:1979vn}
M.~Henneaux, ``{Geometry of Zero Signature Space-times},''
{\em Bull. Soc. Math. Belg.} {\bfseries 31} (1979) 47--63.

\bibitem{Hartong:2015xda}
J.~Hartong, ``{Gauging the Carroll Algebra and Ultra-Relativistic Gravity},''
  \href{http://dx.doi.org/10.1007/JHEP08(2015)069}{{\em JHEP} {\bfseries 08}
  (2015) 069},
\href{http://arxiv.org/abs/1505.05011}{{\ttfamily arXiv:1505.05011 [hep-th]}}.

\bibitem{Bergshoeff:2017btm}
E.~Bergshoeff, J.~Gomis, B.~Rollier, J.~Rosseel, and T.~ter Veldhuis,
  ``{Carroll versus Galilei Gravity},''
  \href{http://dx.doi.org/10.1007/JHEP03(2017)165}{{\em JHEP} {\bfseries 03}
  (2017) 165},
\href{http://arxiv.org/abs/1701.06156}{{\ttfamily arXiv:1701.06156 [hep-th]}}.

\bibitem{Perez:2021abf}
A.~P\'erez, ``{Asymptotic symmetries in Carrollian theories of gravity},''
  \href{http://dx.doi.org/10.1007/JHEP12(2021)173}{{\em JHEP} {\bfseries 12}
  (2021) 173}, \href{http://arxiv.org/abs/2110.15834}{{\ttfamily
  arXiv:2110.15834 [hep-th]}}.

\bibitem{Hansen:2021fxi}
D.~Hansen, N.~A. Obers, G.~Oling, and B.~T. S\o{}gaard, ``{Carroll Expansion of
  General Relativity},''
  \href{http://dx.doi.org/10.21468/SciPostPhys.13.3.055}{{\em SciPost Phys.}
  {\bfseries 13} no.~3, (2022) 055},
  \href{http://arxiv.org/abs/2112.12684}{{\ttfamily arXiv:2112.12684
  [hep-th]}}.

\bibitem{Perez:2022jpr}
A.~P\'erez, ``{Asymptotic symmetries in Carrollian theories of gravity with a
  negative cosmological constant},''
  \href{http://dx.doi.org/10.1007/JHEP09(2022)044}{{\em JHEP} {\bfseries 09}
  (2022) 044}, \href{http://arxiv.org/abs/2202.08768}{{\ttfamily
  arXiv:2202.08768 [hep-th]}}.

\bibitem{Figueroa-OFarrill:2022mcy}
J.~Figueroa-O'Farrill, E.~Have, S.~Prohazka, and J.~Salzer, ``{The gauging
  procedure and carrollian gravity},''
  \href{http://dx.doi.org/10.1007/JHEP09(2022)243}{{\em JHEP} {\bfseries 09}
  (2022) 243}, \href{http://arxiv.org/abs/2206.14178}{{\ttfamily
  arXiv:2206.14178 [hep-th]}}.

\bibitem{Campoleoni:2022ebj}
A.~Campoleoni, M.~Henneaux, S.~Pekar, A.~P\'erez, and P.~Salgado-Rebolledo,
  ``{Magnetic Carrollian gravity from the Carroll algebra},''
  \href{http://dx.doi.org/10.1007/JHEP09(2022)127}{{\em JHEP} {\bfseries 09}
  (2022) 127}, \href{http://arxiv.org/abs/2207.14167}{{\ttfamily
  arXiv:2207.14167 [hep-th]}}.

\bibitem{Sengupta:2022rbd}
S.~Sengupta, ``{Hamiltonian form of Carroll gravity},''
  \href{http://dx.doi.org/10.1103/PhysRevD.107.024010}{{\em Phys. Rev. D}
  {\bfseries 107} no.~2, (2023) 024010},
  \href{http://arxiv.org/abs/2208.02983}{{\ttfamily arXiv:2208.02983 [gr-qc]}}.

\bibitem{Bergshoeff:2016soe}
E.~Bergshoeff, D.~Grumiller, S.~Prohazka, and J.~Rosseel, ``{Three-dimensional
  Spin-3 Theories Based on General Kinematical Algebras},''
  \href{http://dx.doi.org/10.1007/JHEP01(2017)114}{{\em JHEP} {\bfseries 01}
  (2017) 114},
\href{http://arxiv.org/abs/1612.02277}{{\ttfamily arXiv:1612.02277 [hep-th]}}.

\bibitem{Matulich:2019cdo}
J.~Matulich, S.~Prohazka, and J.~Salzer, ``{Limits of three-dimensional gravity
  and metric kinematical Lie algebras in any dimension},''
  \href{http://dx.doi.org/10.1007/JHEP07(2019)118}{{\em JHEP} {\bfseries 07}
  (2019) 118},
\href{http://arxiv.org/abs/1903.09165}{{\ttfamily arXiv:1903.09165 [hep-th]}}.

\bibitem{Huang:2023zhp}
X.~Huang, ``{A Chern-Simons theory for dipole symmetry},''
  \href{http://arxiv.org/abs/2305.02492}{{\ttfamily arXiv:2305.02492
  [cond-mat.str-el]}}.

\bibitem{Ravera:2019ize}
L.~Ravera, ``{AdS Carroll Chern-Simons supergravity in 2 + 1 dimensions and its
  flat limit},'' \href{http://dx.doi.org/10.1016/j.physletb.2019.06.026}{{\em
  Phys. Lett. B} {\bfseries 795} (2019) 331--338},
  \href{http://arxiv.org/abs/1905.00766}{{\ttfamily arXiv:1905.00766
  [hep-th]}}.

\bibitem{Concha:2021jnn}
P.~Concha, D.~Pe\~nafiel, L.~Ravera, and E.~Rodr\'\i{}guez,
  ``{Three-dimensional Maxwellian Carroll gravity theory and the cosmological
  constant},'' \href{http://dx.doi.org/10.1016/j.physletb.2021.136735}{{\em
  Phys. Lett. B} {\bfseries 823} (2021) 136735},
  \href{http://arxiv.org/abs/2107.05716}{{\ttfamily arXiv:2107.05716
  [hep-th]}}.

\bibitem{Concha:2022muu}
P.~Concha, C.~Henr\'\i{}quez-B\'aez, and E.~Rodr\'\i{}guez, ``{Non-relativistic
  and ultra-relativistic expansions of three-dimensional spin-3 gravity
  theories},'' \href{http://dx.doi.org/10.1007/JHEP10(2022)155}{{\em JHEP}
  {\bfseries 10} (2022) 155}, \href{http://arxiv.org/abs/2208.01013}{{\ttfamily
  arXiv:2208.01013 [hep-th]}}.

\bibitem{Grumiller:2020elf}
D.~Grumiller, J.~Hartong, S.~Prohazka, and J.~Salzer, ``{Limits of JT
  gravity},'' \href{http://dx.doi.org/10.1007/JHEP02(2021)134}{{\em JHEP}
  {\bfseries 02} (2021) 134}, \href{http://arxiv.org/abs/2011.13870}{{\ttfamily
  arXiv:2011.13870 [hep-th]}}.

\bibitem{Gomis:2020wxp}
J.~Gomis, D.~Hidalgo, and P.~Salgado-Rebolledo, ``{Non-relativistic and
  Carrollian limits of Jackiw-Teitelboim gravity},''
  \href{http://dx.doi.org/10.1007/JHEP05(2021)162}{{\em JHEP} {\bfseries 05}
  (2021) 162}, \href{http://arxiv.org/abs/2011.15053}{{\ttfamily
  arXiv:2011.15053 [hep-th]}}.

\bibitem{Ravera:2022buz}
L.~Ravera and U.~Zorba, ``{Carrollian and non-relativistic
  Jackiw\textendash{}Teitelboim supergravity},''
  \href{http://dx.doi.org/10.1140/epjc/s10052-023-11239-x}{{\em Eur. Phys. J.
  C} {\bfseries 83} no.~2, (2023) 107},
  \href{http://arxiv.org/abs/2204.09643}{{\ttfamily arXiv:2204.09643
  [hep-th]}}.

\bibitem{Baig:2023yaz}
S.~A. Baig, J.~Distler, A.~Karch, A.~Raz, and H.-Y. Sun, ``{Spacetime Subsystem
  Symmetries},'' \href{http://arxiv.org/abs/2303.15590}{{\ttfamily
  arXiv:2303.15590 [hep-th]}}.

\bibitem{Kasikci:2023tvs}
O.~Kasikci, M.~Ozkan, and Y.~Pang, ``{A Carrollian Origin of Spacetime
  Subsystem Symmetry},'' \href{http://arxiv.org/abs/2304.11331}{{\ttfamily
  arXiv:2304.11331 [hep-th]}}.

\bibitem{Bergshoeff:2022eog}
E.~Bergshoeff, J.~Figueroa-O'Farrill, and J.~Gomis, ``{A non-lorentzian
  primer},'' \href{http://arxiv.org/abs/2206.12177}{{\ttfamily arXiv:2206.12177
  [hep-th]}}.

\bibitem{Bekaert:2017khg}
X.~Bekaert and E.~D. Skvortsov, ``{Elementary particles with continuous
  spin},'' \href{http://dx.doi.org/10.1142/S0217751X17300198}{{\em Int. J. Mod.
  Phys. A} {\bfseries 32} no.~23n24, (2017) 1730019},
  \href{http://arxiv.org/abs/1708.01030}{{\ttfamily arXiv:1708.01030
  [hep-th]}}.

\bibitem{Abbott:1976bb}
L.~F. Abbott, ``{Massless Particles with Continuous Spin Indices},''
  \href{http://dx.doi.org/10.1103/PhysRevD.13.2291}{{\em Phys. Rev. D}
  {\bfseries 13} (1976) 2291}.

\bibitem{Hirata:1977ss}
K.~Hirata, ``{Quantization of Massless Fields with Continuous Spin},''
  \href{http://dx.doi.org/10.1143/PTP.58.652}{{\em Prog. Theor. Phys.}
  {\bfseries 58} (1977) 652--666}.

\bibitem{Longo:2015tra}
R.~Longo, V.~Morinelli, and K.-H. Rehren, ``{Where Infinite Spin Particles Are
  Localizable},'' \href{http://dx.doi.org/10.1007/s00220-015-2475-9}{{\em
  Commun. Math. Phys.} {\bfseries 345} no.~2, (2016) 587--614},
  \href{http://arxiv.org/abs/1505.01759}{{\ttfamily arXiv:1505.01759
  [math-ph]}}.

\bibitem{Duval:2014uva}
C.~Duval, G.~W. Gibbons, and P.~A. Horvathy, ``{Conformal Carroll groups and
  BMS symmetry},'' \href{http://dx.doi.org/10.1088/0264-9381/31/9/092001}{{\em
  Class. Quant. Grav.} {\bfseries 31} (2014) 092001},
\href{http://arxiv.org/abs/1402.5894}{{\ttfamily arXiv:1402.5894 [gr-qc]}}.

\bibitem{Donnay:2022aba}
L.~Donnay, A.~Fiorucci, Y.~Herfray, and R.~Ruzziconi, ``{Carrollian Perspective
  on Celestial Holography},''
  \href{http://dx.doi.org/10.1103/PhysRevLett.129.071602}{{\em Phys. Rev.
  Lett.} {\bfseries 129} no.~7, (2022) 071602},
  \href{http://arxiv.org/abs/2202.04702}{{\ttfamily arXiv:2202.04702
  [hep-th]}}.

\bibitem{Bagchi:2022emh}
A.~Bagchi, S.~Banerjee, R.~Basu, and S.~Dutta, ``{Scattering Amplitudes:
  Celestial and Carrollian},''
  \href{http://dx.doi.org/10.1103/PhysRevLett.128.241601}{{\em Phys. Rev.
  Lett.} {\bfseries 128} no.~24, (2022) 241601},
  \href{http://arxiv.org/abs/2202.08438}{{\ttfamily arXiv:2202.08438
  [hep-th]}}.

\bibitem{Bekaert:2022oeh}
X.~Bekaert, A.~Campoleoni, and S.~Pekar, ``{Carrollian conformal scalar as
  flat-space singleton},''
  \href{http://dx.doi.org/10.1016/j.physletb.2023.137734}{{\em Phys. Lett. B}
  {\bfseries 838} (2023) 137734},
  \href{http://arxiv.org/abs/2211.16498}{{\ttfamily arXiv:2211.16498
  [hep-th]}}.

\bibitem{Saha:2023hsl}
A.~Saha, ``{Carrollian Approach to $1+3$D Flat Holography},''
  \href{http://arxiv.org/abs/2304.02696}{{\ttfamily arXiv:2304.02696
  [hep-th]}}.

\bibitem{Salzer:2023jqv}
J.~Salzer, ``{An Embedding Space Approach to Carrollian CFT Correlators for
  Flat Space Holography},'' \href{http://arxiv.org/abs/2304.08292}{{\ttfamily
  arXiv:2304.08292 [hep-th]}}.

\bibitem{Nguyen:2023vfz}
K.~Nguyen and P.~West, ``{Carrollian conformal fields and flat holography},''
  \href{http://arxiv.org/abs/2305.02884}{{\ttfamily arXiv:2305.02884
  [hep-th]}}.

\bibitem{Satishchandran:2019pyc}
G.~Satishchandran and R.~M. Wald, ``{Asymptotic behavior of massless fields and
  the memory effect},''
  \href{http://dx.doi.org/10.1103/PhysRevD.99.084007}{{\em Phys. Rev. D}
  {\bfseries 99} no.~8, (2019) 084007},
  \href{http://arxiv.org/abs/1901.05942}{{\ttfamily arXiv:1901.05942 [gr-qc]}}.

\bibitem{Bekaert:2022ipg}
X.~Bekaert and B.~Oblak, ``{Massless scalars and higher-spin BMS in any
  dimension},'' \href{http://dx.doi.org/10.1007/JHEP11(2022)022}{{\em JHEP}
  {\bfseries 11} (2022) 022}, \href{http://arxiv.org/abs/2209.02253}{{\ttfamily
  arXiv:2209.02253 [hep-th]}}.

\bibitem{Chen:2021xkw}
B.~Chen, R.~Liu, and Y.-f. Zheng, ``{On Higher-dimensional Carrollian and
  Galilean Conformal Field Theories},''
  \href{http://arxiv.org/abs/2112.10514}{{\ttfamily arXiv:2112.10514
  [hep-th]}}.

\bibitem{Morand:2018tke}
K.~Morand, ``{Embedding Galilean and Carrollian geometries I. Gravitational
  waves},'' \href{http://dx.doi.org/10.1063/1.5130907}{{\em J. Math. Phys.}
  {\bfseries 61} no.~8, (2020) 082502},
  \href{http://arxiv.org/abs/1811.12681}{{\ttfamily arXiv:1811.12681
  [hep-th]}}.

\bibitem{Gromov:2017vir}
A.~Gromov, ``{Chiral Topological Elasticity and Fracton Order},''
  \href{http://dx.doi.org/10.1103/PhysRevLett.122.076403}{{\em Phys. Rev.
  Lett.} {\bfseries 122} no.~7, (2019) 076403},
  \href{http://arxiv.org/abs/1712.06600}{{\ttfamily arXiv:1712.06600
  [cond-mat.str-el]}}.

\bibitem{Slagle:2018kqf}
K.~Slagle, A.~Prem, and M.~Pretko, ``{Symmetric Tensor Gauge Theories on Curved
  Spaces},'' \href{http://dx.doi.org/10.1016/j.aop.2019.167910}{{\em Annals
  Phys.} {\bfseries 410} (2019) 167910},
  \href{http://arxiv.org/abs/1807.00827}{{\ttfamily arXiv:1807.00827
  [cond-mat.str-el]}}.

\bibitem{Yan:2022yix}
H.~Yan, K.~Slagle, and A.~H. Nevidomskyy, ``{Y-cube model and fractal structure
  of subdimensional particles on hyperbolic lattices},''
  \href{http://arxiv.org/abs/2211.15829}{{\ttfamily arXiv:2211.15829
  [quant-ph]}}.

\bibitem{Yan:2023lmj}
H.~Yan, C.~B. Jepsen, and Y.~Oz, ``{$p$-adic Holography from the Hyperbolic
  Fracton Model},'' \href{http://arxiv.org/abs/2306.07203}{{\ttfamily
  arXiv:2306.07203 [hep-th]}}.

\bibitem{IP:2023}
J.~Figueroa-O'Farrill, S.~Pekar, A.~P\'erez, and S.~Prohazka, ``{Particles with
  restricted mobility (in preparation)},''
  \href{http://arxiv.org/abs/23XX.XXXXX}{{\ttfamily arXiv:23XX.XXXXX
  [hep-th]}}.

\bibitem{MR0213463}
R.~Hermann, {\em Lie groups for physicists}.
\newblock W. A. Benjamin, Inc., New York-Amsterdam, 1966.

\bibitem{MR0479129}
R.~Hermann, {\em Fourier analysis on groups and partial wave analysis}.
\newblock Mathematics Lecture Note Series. W. A. Benjamin, Inc., New York,
  1969.

\bibitem{MR0479084}
U.~H. Niederer and L.~O'Raifeartaigh, ``Realizations of the unitary
  representations of the inhomogeneous space-time groups. {I}. {G}eneral
  structure,'' \href{http://dx.doi.org/10.1002/prop.19740220302}{{\em Fortschr.
  Physik} {\bfseries 22} (1974) 111--129}.

\bibitem{MR0479085}
U.~H. Niederer and L.~O'Raifeartaigh, ``Realizations of the unitary
  representations of the inhomogeneous space-time groups. {II}. {C}ovariant
  realizations of the {P}oincar\'{e} group,''
  \href{http://dx.doi.org/10.1002/prop.19740220303}{{\em Fortschr. Physik}
  {\bfseries 22} (1974) 131--157}.

\end{thebibliography}

\end{document}